\documentclass[a4paper,12pt,parskip=half,english,numbers=noenddot]{scrartcl}

\usepackage{babel}
\usepackage[T1]{fontenc}
\usepackage{lmodern,textcomp,latexsym}
\usepackage{amsthm,amssymb,amsmath}
\usepackage{graphicx,psfrag,color,rotate}
\usepackage{paralist,threeparttable}
\usepackage[footnotesize,labelfont={bf}]{caption}
\usepackage{alphalph}
\usepackage{longtable,booktabs}
\usepackage{tabularx,nth}
\usepackage{tikz}
\usepackage{pifont}
\usepackage{enumerate}
\usepackage{ulem}
\usepackage{cancel}

\usepackage{hyperref}

\usepackage{verbatim}

\DeclareCaptionStyle{table_new}{font={footnotesize},justification=raggedright}
\setdefaultenum{1)}{\theenumi.1)}{}{}

\newtheorem{remark}{Remark}
\newtheorem{lemma}{Lemma}
\newtheorem{proposition}{Proposition}

\renewcommand{\vec}{\boldsymbol}
\renewcommand{\d}[1][]{\ensuremath{\,\mathrm{d}#1}}

\newcommand{\op}[1]{\ensuremath{\operatorname{#1}}}

\makeatletter
\renewcommand*{\env@matrix}[1][*\c@MaxMatrixCols c]{%
  \hskip -\arraycolsep
  \let\@ifnextchar\new@ifnextchar
  \array{#1}}
\makeatother

\graphicspath{{./pictures/}}

\DeclareOldFontCommand{\it}{\normalfont\itshape}{\mathit}

\newcommand{\X}{\vec{X}}
\newcommand{\x}{\vec{x}}
\newcommand{\y}{\vec{y}}

\newcommand{\abs}[1]{\left|{#1}\right|}

\newcommand{\inv}{^{-1}}
\newcommand{\tp}{\ensuremath{^{\scriptstyle \mathrm{T}}}}

\renewcommand{\exp}[1]{e^{#1}}

\renewcommand{\i}{\mathrm{i}}
\renewcommand{\d}{\operatorname{d}\!}

\renewcommand{\O}{\mathcal{O}}

\newcommand{\VPTen}[1]{\left[#1\right]_{\times}}
\renewcommand{\det}{\operatorname{det}}
\newcommand{\cof}{\operatorname{cof}}

\renewcommand{\vec}[1]{\pmb{#1}}
\newcommand{\vct}[1]{\vec{#1}}
\newcommand{\ten}[1]{\pmb{#1}}

\newcommand{\drv}[2]{\frac{\partial #1}{\partial #2}} %{\partial_{#1} #2}
\newcommand{\R}{\mathbb{R}} % real
\renewcommand{\L}[1]{\op{L}^{\!#1}} % Lesbegues space
\renewcommand{\H}[1]{\op{H}^{#1}} % Sobolev space
\newcommand{\Cn}[1]{\op{C}^{#1}} % Sobolev space
\newcommand{\TrialSpace}{\mathcal{C}} % space of trial functions
\newcommand{\TestSpace}{\mathcal{V}}  % space of test functions

\newcommand{\fib}[1]{\tilde{#1}}
\newcommand{\mtx}[1]{#1}
\newcommand{\axl}{\operatorname{axl}}

\newcommand{\EC}[1]{\left.\left[\!\!\left[ #1 \right]\!\!\right] \right|_{0}^{L}}

\newcommand{\CH}[1]{{\color{black} #1}}
\newcommand{\BW}[1]{{\color{black} #1}}
\newcommand{\UK}[1]{{\color{black} #1}}

\newcommand{\review}[1]{{\color{black} #1}}

\usepackage{tikz}
\usetikzlibrary{decorations.pathmorphing} % for snake lines
\usetikzlibrary{matrix} % for block alignment
\usetikzlibrary{arrows, arrows.meta} % for arrow heads
\usetikzlibrary{calc} % for manimulation of coordinates
\usetikzlibrary{shapes}
\usetikzlibrary{shapes.geometric}

\usepgfmodule{nonlineartransformations}
\usepgflibrary{curvilinear}

\usepackage{tikz-3dplot} % needs tikz-3dplot.sty in same folder

\usepackage{bm}

\tikzstyle{block} = [draw,rectangle,thick,minimum height=2em,minimum width=2em]
\tikzstyle{sum} = [draw,circle,inner sep=0mm,minimum size=2mm]
\tikzstyle{connector} = [->,thick]
\tikzstyle{line} = [thick]
\tikzstyle{branch} = [circle,inner sep=0pt,minimum size=1mm,fill=black,draw=black]
\tikzstyle{guide} = []
\tikzstyle{snakeline} = [connector, decorate, decoration={pre length=0.2cm,
	post length=0.2cm, snake, amplitude=.4mm,
	segment length=2mm},thick, magenta, ->]

\begin{document}

\begin{center}
\large{\textbf{Multidimensional coupling: A variationally consistent approach to fiber-reinforced materials.}}

{\large Ustim Khristenko$^{a}$, Stefan Schu{\ss}$^{b}$, Melanie Kr{\"u}ger$^{b}$, Felix Schmidt$^{b}$, Barbara Wohlmuth$^{a}$, Christian Hesch$^{b}$\footnote{Corresponding author. E-mail address: christian.hesch@uni-siegen.de}}

{\small
\(^a\) Faculty of Mathematics, Technical University of Munich, Garching, Germany\\
\(^b\) Chair of Computational Mechanics, University of Siegen, Siegen, Germany}

\end{center}

\vspace*{-0.1cm}\textbf{Abstract}

\BW{A novel mathematical model for fiber-reinforced materials is proposed. It is based on a 1-dimensional beam model for the thin fiber structures, a flexible and general 3-dimensional elasticity model for the matrix and an overlapping domain decomposition approach. From a computational point of view, this is  motivated by the fact that matrix and fibers can easily meshed independently. Our main interest is in fiber reinforce polymers where the Young's modulus are quite different. Thus the modeling error from the overlapping approach is of no significance. }%A novel coupling algorithm for the coupling of 1-dimensional fibers considered as beam models to the surrounding 3-dimensional matrix material is presented.  An overlapping domain decomposition is used to embed the fibers within the matrix material.   
The coup\-ling conditions acknowledge both, the forces and the moments of the beam model and transfer them to the background material.  A suitable static condensation procedure is applied to remove the beam balance equations. \BW{The condensed system then forms our starting point for a numerical approximation in terms of isogeometric analysis. The choice of our discrete basis functions of higher regularity is motivated by the fact, that as a result of the static condensation, we obtain second gradient terms in fiber direction. }% in the continuum setting and afterwards in the coupling constraints in the discrete setting, reducing drastically the computational cost.  
Eventually, a series of benchmark tests demonstrate the flexibility and robustness of the proposed methodology. \BW{As a proof-of-concept, we show that our new model is able to capture bending, torsion and shear dominated situations.}

\section{Introduction}

Fiber reinforced materials are subjected to various physical mechanisms on different scales, depending on the size, orientation and distribution within a suitable matrix material. Applications contain, e.g.\ steel reinforced ultra-high performance concrete or fiber reinforced polymers.  Beside such technical materials, many biological tissues are reinforced by certain types of fibers.  A model containing fibers fully resolved as a discretized continua, subsequently referred to as Cauchy continuum approach, is far out of the range of today's computational capabilities.  Therefore, we develop a model containing the fibers as continuum degenerated 1-dimensional beams within the matrix material in the sense of an overlapping domain decomposition method. 

% General second gradient formulation and homogenisation
This intermediate scale model can be considered as a material with microstructures.  We refer here to dell'Isola et al.\ \cite{hesch2019a},  where panthographic mechanisms as prototypical material with dedicated microstructure have been investigated.  In Giorgio \cite{giorgio2016} a detailed ana\-ly\-sis demonstrates that a microscale model using a Cauchy continuum model can be homogenized using a macro second gradient model,  see also dell'Isola et al.\ \cite{delIsola2015} for a suitable Piola homogenization of rod-like structures.  More generally,  we consider a higher-gradient framework as proposed by Mindlin \cite{mindlin1964,mindlin1965a},  see also the work of Germain \cite{germain1973b}, Toupin \cite{toupin1962,toupin1964} as well as Eringen \cite{eringen1999}. We refer also to  Asmanoglo \& Menzel \cite{menzel2017a,menzel2017c} for higher-order formulations used in the context of composites based on the early work of Spencer \& Soldatos \cite{spencer2007} and Soldatos \cite{soldatos2010}.  For the application to Kirchhoff-Love shell elements see Schulte et al.\ \cite{Schulte2020b} and Dittmann et al.\ \cite{dittman2020straingradient} for the application of strain gradient formulation on porous-ductile fracture.

% Beam formulations
For the embedded beams,  several models can be taken into account.  Geometrically exact beam formulations using finite elements are presented foremost in the Simo-Reissner beam model,  see Simo \cite{Simo1985}, Simo \& Vu-Quoc \cite{Simo1986} and Reissner \cite{Rei1981}.  For the parametrization of the rotation, director interpolations can be employed to the  Simo-Reissner beam theory as shown in Romero \& Armero \cite{romero2001} and Betsch \& Steinmann \cite{betsch2002b,betsch2002c} and further advanced in Eugster et al.\ \cite{eugster2014}. Quaternions are used e.g.\ in McRobie \& Lasenby \cite{mcrobie1999}, see also in Weeger et al.\ \cite{Weeger2017} for the application of quaternions in the context of isogeometric collocation methods.  In Meier et al.\ \cite{meier2014},  geometrically exact Kirchhoff rods are used and in Meier et al.\ \cite{meier2018} applied for the modeling of fiber based materials using contact algorithms.

% Immersed technologies
In contrast to the surface coupling of dedicated surfaces between beam and solid using either Dirichlet or Neumann conditions (e.g.\ within staggered algorithms or via Lagrange multipliers using Mortar methods),  overlapping domain decomposition allows for a flexible and efficient embedding of the fibers.  This idea has been applied in the context of fluid-structure interaction problems (FSI) using the terminus immersed technologies,  see \cite{peskin2002,liu2006,liu2007,gil2010,hesch2012b,hesch2014b}.  From a methodological point of view, immersed methods for FSI are part of the so-called Fictitious Domain (FD) philosophy,  introduced by Glowinski et al.\ in \cite{glowinski1994} for the resolution of boundary value problems in complex geometrical settings.  For the application on two elastic domains,  we refer to Sanders \& Puso \cite{puso2012}.  An overlapping domain decomposition for beam-solid interaction problems using position constraints is presented in Steinbrecher et al.\ \cite{steinbrecher2019c}. 

% What we do
To be more specific, \BW{we propose a new mathematical model for fiber-reinforced materials.}
%the model proposed in this contribution derives the coupling conditions between a 3-dimensional matrix material and a dimensional reduced beam model.  
The model is derived from a full surface-to-volume beam-matrix coupling via Lagrange multiplier, representing the coupling force. \BW{For the matrix we apply a large strain elasticity model in 3 dimensions and for the fiber a 1-dimensional beam model.} This full model is \BW{then} reduced via its Fourier expansion as function of the angular coordinate in the beam cross-section. Taking into account only the first term (average over cross-section) leads to the common displacement coupling constraint~(see, e.g., \cite{steinbrecher2019c}), which neglects the connection of the beam directors to the matrix material. Consideration of the second term allows to transfer dilatation,  shear and coupled stresses on the beam-matrix interface.  It thus provides an additional constraint,  which couples the beam directors with the deformation gradient,  such that bending and torsion are transferred to the matrix as well.  This also yields an additional tensor-valued Lagrange multiplier enforcing the matrix deformation to take into account the incompressibility of the beam cross-section.  The constraints typically involve the circular mean over the beam mantle (see, e.g.,  \cite{d2008coupling,koppl2018mathematical}).  Assuming enough regularity \review{of Galerkin approximation},  we truncate the Taylor expansion of the matrix displacement field at the beam centerline up to \review{quadratic} term with respect to the beam radius~$r$.  Hence, the constraints are collapsed to the beam centerline,  yielding a modeling error of order~\review{$\O\left(r^3\right)$}.
%A subsequent condensation of the beam balance equations
\review{Together with condensation of the beam balance equations, this} leads to a second gradient formulation along the beam\BW{. Thus, $C^1$-regular isogeometric analysis techniques are a natural candidate for discretization.} 
%,  which allows us to extract the so-called coupled forces. 

% Organisation of the paper
The paper is structured as follows. In Section \ref{sec:preliminaries} we present the fundamental formulations for the first and second gradient continuum as well as \review{for} the beam.  In Section~\ref{sec:mainresults},  we derive \review{our main results: the coupling terms and the final variational system.}  The spatial discretization is given in Section \ref{sec:discretization},  followed by a series of representative examples in Section \ref{sec:numerics}.  Eventually,  conclusions are summarized in Section \ref{sec:conclusions}.

\subsection{Definitions and notations.}

This section gives a brief summary on the used notation.  A single contraction of two vectors will be understood as $[\vec{a}\cdot\vec{b}] = a_i\,b_i$,  where the Einstein summation convention on repeated indices is used.  For two second order tensors it holds $[\vec{A}\,\vec{B}]_{ij} = A_{ik}\,B_{kj}$ and the double contraction reads  $[\vec{A}:\vec{B}] = A_{ij}\, B_{ij}$.  Next, we define the gradient with respect to the reference \(\nabla(\bullet)\) of a vector field \(\vec{a}\) and of a second-order tensor field \(\vec{A}\) as
\begin{equation}
[\nabla\vec{a}]_{iJ}=\frac{\partial[\vec{a}]_i}{\partial[\vec{X}]_J}\quad\text{and}\quad[\nabla\vec{A}]_{iJK}=\frac{\partial[\vec{A}]_{iJ}}{\partial[\vec{X}]_K}.
\end{equation}
For the divergence operator it follows
\begin{equation}
[\nabla\cdot\vec{A}]_{i}=\frac{\partial[\vec{A}]_{iJ}}{\partial[\vec{X}]_J}\quad\text{and}\quad[\nabla\cdot\mathfrak{A}]_{iJ}=\frac{\partial[\mathfrak{A}]_{iJK}}{\partial[\vec{X}]_K}.
\end{equation}
The triple contraction for two third order tensors is given via $[\mathfrak{A}\,\vdots\,\mathfrak{B}] = \mathfrak{A}_{ijk}\,\mathfrak{B}_{ijk}$. For the double contraction we define
\begin{equation}
[\mathfrak{A}:\vec{A}]_{i}=\mathfrak{A}_{iJK}\,A_{JK} \quad\text{and}\quad [\mathfrak{A}\tp:\vec{A}]_{K}=\mathfrak{A}_{KiJ}\,A_{iJ}.
\end{equation}
The axial vector of a skew symmetric $3\times 3$ matrix is defined via $[\axl \ten{A}]_i = -\frac{1}{2}\,\epsilon_{ijk}\,[\vec{A}]_{jk}$ and the spin of a 3-dimensional vector $[[\vec{a}]_{\times}]_{ij} = -\epsilon_{ijk}\,[\vec{a}]_{k}$ such that $\vec{a}\times\vec{b} = [\vec{a}]_{\times}\,\vec{b}$, i.e.\
\begin{equation}\label{axialVector}
	\axl \ten{A} = \frac{1}{2}
	\begin{bmatrix}
		A_{32} - A_{23} \\
		A_{13} - A_{31} \\
		A_{21} - A_{12}
	\end{bmatrix}\quad\text{and}\quad
\left[\vec{a}\right]_{\times}
=
\begin{bmatrix}
0& -a_3 & \phantom{-}a_2 \\
\phantom{-}a_3 & 0& -a_1 \\
-a_2 & \phantom{-}a_1 &0
\end{bmatrix}.
\end{equation}
Here, we have made use of the Levi-Civita permutation tensor $\epsilon_{ijk}$ with $\epsilon_{ijk}\,\epsilon_{ijl} = 2\,\delta_{kl}$ and the Kronecker symbol $\delta_{kl} = [\ten{I}]_{kl}$,  such that $\axl \left[\vec{a}\right]_{\times} = \vec{a}$.  Hence,  $\left[\vec{a}\right]_{\times}$ belongs to the vector space of skew-symmetric matrices
\begin{equation}
\mathrm{so}(3) = \{\vec{A}\in\mathbb{R}^{3\times 3}\,|\,\vec{A} + \vec{A}\tp = \vec{0}\},
\end{equation}
The image of the matrix exponential map of $\mathrm{so}(3)$ is the special orthogonal group $\mathrm{SO}(3)$,  i.e.\ $\mathrm{exp}:\mathrm{so}(3)\rightarrow\mathrm{SO}(3)$ and 
\begin{equation}\label{eq:so3}
\mathrm{SO}(3) = \{\vec{R}\in\mathrm{GL}(3)\,|\,\vec{R}\tp\,\vec{R} = \vec{R}\,\vec{R}\tp = \ten{I},\,\mathrm{det}(\vec{R}) = 1\},
\end{equation}
where $\mathrm{GL}(3)$ is the general linear group of $3\times 3$ matrices.  Eventually, we make use of the tensor cross product operation $\vec{\times}$,  defined as  $[\vec{A}\vec{\times}\vec{B}]_{iJ} = \epsilon_{imn}\,\epsilon_{JPQ}\,[\vec{A}]_{mP}\,[\vec{B}]_{nQ}$ for the two-point second order tensors $\vec{A}$ and $\vec{B}$.

Moreover, for any function $f(s)$, $s\in\{0,L\}$,  we introduce the following notation:
\begin{equation}\label{endpoint_int}
\EC{f} := f(L) - f(0).
\end{equation}

\section{Preliminaries and problem description}\label{sec:preliminaries}
We start with a short summary of non-linear continuum mechanics. Therefore, consider a \review{bounded Lipschitz} domain $\Omega_0\subset\mathbb{R}^3$ in its reference configuration with boundary $\partial\Omega_0$ and outward unit normal $\vec{N}$. The actual configuration $\Omega\subset\mathbb{R}^3$ with boundary $\partial\Omega$ and outward unit normal $\vec{n}$ is related to the reference configuration by a deformation mapping $\vec{\varphi}:\Omega_0\rightarrow\mathbb{R}^3$, such that $\Omega = \vec{\varphi}(\Omega_0)$. Material points are labelled by $\vec{X}\in\Omega_0$ with corresponding actual position \(\vec{x} = \vec{\varphi}(\vec{X})\). 

\subsection{Solid mechanics}
\paragraph{Kinematics.}
For the matrix material, we introduce the deformation gradient as second order tensor field $\vec{F}:\Omega_0\rightarrow\mathbb{R}^{3\times 3}$, such that
\begin{equation}\label{eq:deformationgradient}
\vec{F} = \nabla\vec{\varphi}(\vec{X}), 
\end{equation}
which maps the infinitesimal vector $\text{d}\vec{X}$ at $\vec{X}\in\Omega_0$ to the infinitesimal vector $\text{d}\vec{x}$ in the actual configuration. Next, the second order tensor field $\vec{H}:\Omega_0\rightarrow\mathbb{R}^{3\times 3}$ is introduced, where $\vec{H}=\cof{\ten{F}}$ denotes the cofactor of $\ten{F}$, defined as follows
\begin{equation}
\vec{H} = \frac{1}{2}\,\vec{F}\vec{\times}\vec{F},
\end{equation}
which maps the infinitesimal oriented area element $\text{d}\vec{A} = \vec{N}(\vec{X})\,\text{d}A$ to the infinitesimal oriented area element $\text{d}\vec{a} = \vec{n}(\vec{X})\,\mathrm{d}a$ in the actual configuration. Eventually, the scalar field $J:\Omega_0\rightarrow\mathbb{R}$ is introduced, where $J = \det{\ten{F}}$ denotes the determinant, defined as
\begin{equation}
J = \frac{1}{6}\,(\vec{F}\vec{\times}\vec{F}):\vec{F},
\end{equation}
which relates the infinitesimal volume element $\text{d}V$ to the corresponding infinitesimal volume element $\text{d}v$ in the actual configuration. This last equation completes the set of kinematic relations, to be used to define the strain energy of \review{a hyperelastic material in the large strain regime}. Note that we omit in the following the time dependency of the equations and postulate a zero density, as we focus here on static problems.

\paragraph{Large strain elasticity.}
We assume the existence of a strain energy function of the form
\begin{equation}
\Psi := \Psi(\vec{F},\cof{\vec{F}},\det{\vec{F}}),
\end{equation}
%\begin{equation}
%\Psi := \Psi(\vec{F},\review{\vec{H}, J}),
%\end{equation}
where $\Psi:\mathbb{R}^{3\times 3}\times\mathbb{R}^{3\times 3}\times\mathbb{R}\rightarrow\mathbb{R}$. \BW{To ensure frame invariance, i.e., objectivity, we require that $\Psi$ must be independent of rotational components of the deformation gradient and the cofactor.} This allows us to formulate the total stored energy in the material configuration \BW{as
%along with the external contributions \review{respectively as}
\begin{equation}\label{eq:energies}
\Pi^{int} = \int\limits_{\Omega_0}\Psi(\vec{F},\cof{\vec{F}},\det{\vec{F}})\,\mathrm{d}V.
%\quad\Pi^{ext} = 
%-\int\limits_{\Omega_0}\vec{B}_{ext}\cdot\vec{\varphi}\,\mathrm{d}V - \int
%\limits_{\Gamma^{\sigma}}\vec{T}_{ext}\cdot\vec{\varphi}\,\mathrm{d}A.
\end{equation}}
\review{
The virtual work of the internal contributions writes
\begin{equation}\label{eq:va_int}
\delta\Pi^{int}(\vec{\varphi}) = \int\limits_{\Omega_0}\vec{P}(\vec{\varphi}):\nabla\delta\vec{\varphi}\d V,
\end{equation}
using the first Piola-Kirchhoff stress tensor given by
\begin{equation}
\vec{P} = \frac{\partial\Psi}{\partial \vec{F}} + \frac{\partial\Psi}{\partial\vec{H}}\vec{\times}\vec{F} + \frac{\partial\Psi}{\partial J}\,\vec{H},
\end{equation}
see Bonet et al.\ \cite{bonet2015a} and references therein on polyconxity.
}

\review{
	\begin{remark}
		The well-known Mooney-Rivlin constitutive model
		\begin{equation}\label{eq:mooney}
		\Psi_{MR}(\vec{F},\vec{H},J) = \alpha\,\vec{F}:\vec{F} + \beta\,\vec{H}:\vec{H} + f(J),
		\end{equation}
		where $f(J) = -2\,\alpha\, \operatorname{ln}(J) - 4\,\beta \,J + \lambda/2\,(J-1)^2-3\,(\alpha+\beta)$ and $\alpha,\,\beta,\,\lambda$ are non-negative material parameters, is a suitable polyconvex constitutive model. In particular, $\Psi_{MR}$ is convex with respect to its 19 variables, i.e.\ the components of $\vec{F}$, $\vec{H}$ and $J$, which can be seen by using the Hessian of $\Psi_{MR}$, see \cite{bonet2015a}.
	\end{remark}
}
\BW{
	Let the boundary $\partial\Omega_0 = \Gamma^{\varphi}\cup\Gamma^{\sigma}$, such that $\Gamma^{\varphi}\cap\Gamma^{\sigma} = \emptyset$, be decomposed in Dirichlet boundary $\Gamma^{\varphi}$ with prescribed deformation~$\vec{\varphi}_{\Gamma}$ and Neumann boundary $\Gamma^{\sigma}$ with prescribed external surface stresses $\vec{T}_{ext}$.
	Denoting the external body forces by $\vec{B}_{ext}$, the external energy contributions are given by
	\begin{equation}\label{eq:ext_energies}
	\Pi^{ext} = 
	-\int\limits_{\Omega_0}\vec{B}_{ext}\cdot\vec{\varphi}\,\mathrm{d}V - \int
	\limits_{\Gamma^{\sigma}}\vec{T}_{ext}\cdot\vec{\varphi}\,\mathrm{d}A.
	\end{equation}	 
	Then, the principle of virtual work reads
\begin{equation}\label{eq:pdva}
	\delta\Pi^{int}(\vec{\varphi}) + \delta\Pi^{ext}(\vec{\varphi})=0.
\end{equation}
}

%\BW{Denoting the body forces by $\vec{B}_{ext}$ and the surface stresses on the Neumann boundary by $\vec{T}_{ext}$, the external contributions are given by
%\begin{equation}\label{eq:ext_energies}
%\Pi^{ext} = 
%-\int\limits_{\Omega_0}\vec{B}_{ext}\cdot\vec{\varphi}\,\mathrm{d}V - \int
%\limits_{\Gamma^{\sigma}}\vec{T}_{ext}\cdot\vec{\varphi}\,\mathrm{d}A.
%\end{equation}
%The} boundary $\partial\Omega_0$ is decomposed in Dirichlet boundary $\Gamma^{\varphi}$ with prescribed displacements $\bar{\vec{\varphi}}$ and Neumann boundary $\Gamma^{\sigma}$ with the standard relationship $\partial\Omega_0 = \Gamma^{\varphi}\cup\Gamma^{\sigma}$ and \review{$\Gamma^{\varphi}\cap\Gamma^{\sigma} = \emptyset$. 
%The} principle of virtual work reads
%\begin{equation}\label{eq:pdva}
%\review{\delta\Pi^{int}(\vec{\varphi}) + \delta\Pi^{ext}(\vec{\varphi})=0.}
%\end{equation}
%%\begin{equation}\label{eq:pdva}
%%\review{\delta\Pi^{int}(\vec{\varphi}) + \delta\Pi^{ext}(\vec{\varphi})=\,}
%%\frac{\text{d}}{\text{d}\epsilon} (\Pi^{int}(\vec{\varphi}_{\epsilon}) + \Pi^{ext}(\vec{\varphi}_{\epsilon}))\biggl|_{\epsilon = 0} = 0,
%%\end{equation}
%%where $\vec{\varphi}_{\epsilon}(\vec{X}) = \vec{\varphi}(\vec{X}) + \epsilon\,\delta\vec{\varphi}(\vec{X})$.

\review{
	Let $\TrialSpace$ be a suitable configuration function space, satisfying $J > 0$ in $\Omega_0$ and the Dirichlet boundary conditions $\vec{\varphi}|_{\Gamma^{\varphi}} = \vec{\varphi}_{\Gamma}$,
	and $\TestSpace$ be a suitable function space of kinematically admissible variations vanishing on~$\Gamma^{\varphi}$
	(cf.  Simo et al.\ \cite{simo1992h} and Marsden \& Hughes \cite{marsden1983}).
	Then, according to~\eqref{eq:pdva} and \eqref{eq:va_int}, the problem reads: find $\vec{\varphi}\in\TrialSpace$ such that for all $\delta\vec{\varphi}\in\TestSpace$ it holds
\begin{equation}\label{eq:pdvA}
\int\limits_{\Omega_0}\vec{P}:\nabla\delta\vec{\varphi} \d V - \int\limits_{\Omega_0}\vec{B}_{ext}\cdot\delta\vec{\varphi}\d V -
\int\limits_{\Gamma^{\sigma}}\vec{T}_{ext}\cdot\delta\vec{\varphi}\d A = 0.
\end{equation}
%where the virtual work of the internal contributions reads
%\begin{equation}
%\delta\Pi^{int}(\vec{\varphi}) = \int\limits_{\Omega_0}\vec{P}:\nabla\delta\vec{\varphi}\d V,
%\end{equation}
%using the first Piola-Kirchhoff stress tensor, given by
%\begin{equation}
%\vec{P} = \frac{\partial\Psi}{\partial \vec{F}} + \frac{\partial\Psi}{\partial\vec{H}}\vec{\times}\vec{F} + \frac{\partial\Psi}{\partial J}\,\vec{H},
%\end{equation}
%see Bonet et al.\ \cite{bonet2015a} and references therein on polyconxity.
%Here $\TrialSpace$ is the configuration function space (cf.  Simo et al.\ \cite{simo1992h} for the basic notation and Marsden \& Hughes \cite{marsden1983} for further information), satisfying $J > 0$ in $\Omega_0$ and the Dirichlet boundary conditions $\vec{\varphi}|_{\Gamma^{\varphi}} = \vec{\varphi}_{\Gamma}$.
%Moreover, $\TestSpace$ is a suitable function space of kinematically admissible variations, satisfying homogeneous Dirichlet boundary conditions on $\Gamma^{\varphi}$. 
}

\paragraph{Second gradient material.}
Next we assume, that the second gradient of the deformation $\nabla\vec{F}$ with respect to the reference configuration can be taken into account as well, see, among many others, \cite{steinmann2013,mindlin1965a}. Here, we start with a most general definition of a second gradient strain energy function, defined by
\begin{equation}
\Psi := \Psi(\nabla\vec{F},\vec{F},\cof{\vec{F}},\det{\vec{F}}),
\end{equation}
such that the internal virtual work reads
\begin{equation}
\delta\Pi^{int}(\vec{\varphi}) = \int\limits_{\Omega_0}\vec{P}:\nabla\delta\vec{\varphi} + \mathfrak{P}\,\vdots\,\nabla^2\delta\vec{\varphi}\d V,
\end{equation}
where $\mathfrak{P}$ is a third order stress tensor, conjugated to $\nabla^2\delta\vec{\varphi}$.  To identify and collect the corresponding boundary conditions, we apply twice integration by parts and the divergence theorem yielding
\begin{equation}
\delta\Pi^{int}(\vec{\varphi}) = \int\limits_{\Omega_0}\nabla\cdot(\nabla\cdot\mathfrak{P}-\vec{P})\cdot\delta\vec{\varphi} \d V +
\int\limits_{\partial\Omega_0}\delta\vec{\varphi}\cdot(\vec{P}-\nabla\cdot\mathfrak{P})\,\vec{N} + \nabla\delta\vec{\varphi}:(\mathfrak{P}\cdot\vec{N})\d A.
\end{equation}
Using the orthogonal decomposition $\nabla_{\bot}\cdot(\bullet) = \nabla(\bullet):(\vec{N}\otimes\vec{N})$ and $\nabla_{\|}\cdot(\bullet) = \nabla(\bullet):(\ten{I}-\vec{N}\otimes\vec{N})$, we obtain after some further technical steps
\begin{equation}\label{eq:gradBoundary}
\begin{aligned}
\delta\Pi^{int}(\vec{\varphi}) =& \,\int\limits_{\Omega_0}\nabla\cdot(\nabla\cdot\mathfrak{P}-\vec{P})\cdot\delta\vec{\varphi} \d V +
\int\limits_{\partial\Omega_0}\delta\vec{\varphi}\cdot(\vec{P}-\nabla\cdot\mathfrak{P})\,\vec{N} \d A \\
&-\int\limits_{\partial\Omega_0}\left[\delta\vec{\varphi}\cdot(K\,(\mathfrak{P}\,\vec{N})\,\vec{N}+\nabla_{\|}\cdot(\mathfrak{P}\,\vec{N}))  - \nabla_{\bot}\delta\vec{\varphi}:\left(\mathfrak{P}\,\vec{N}\right)\right]\d A \\ 
&+\int\limits_{\partial^2\Omega_0}\delta\vec{\varphi}\cdot(\mathfrak{P}:(\hat{\vec{N}}\otimes\vec{N}))\d S,
\end{aligned}
\end{equation}
\review{for a sufficiently smooth $\Omega_0$,}
where $\hat{\vec{N}}$ is the normal to $\partial^2\Omega_0$ and the tangent to $\partial\Omega_0$.  \review{Note that $\partial^2\Omega_0$ is defined by the union of the boundary curves of the boundary surface patches and thus, $\hat{\vec{N}}$ can be defined differently from both adjacent surfaces, see Javili et al. \cite{steinmann2013} and the citations therein for details.}
Moreover, $K = -\nabla_{\|}\cdot\vec{N}$ is the curvature of the surface. \review{Equilibrating this result with the external contributions defined in \eqref{eq:energies} yields 
\begin{equation}\label{eq:secondGradient}
\begin{aligned}
\int\limits_{\Omega_0}\vec{P}:\nabla\delta\vec{\varphi}+\mathfrak{P}\,\vdots\,\nabla^2\delta\vec{\varphi}\d V - \int\limits_{\Omega_0}\vec{B}_{ext}\cdot\delta\vec{\varphi}\d V -
\int\limits_{\Gamma^{\sigma}}\vec{T}_{ext}\cdot\delta\vec{\varphi}\d A = 0.
\end{aligned}
\end{equation}
Note that additional external contributions related to the gradient on the boundary can be applied as well. In general, there are only a few examples for higher-order strain energy formulations known, we refer to dell'Isola et al. \cite{delIsola2015} where a Piola homogenization procedure is used to derive a suitable formulation.  We will show the specific construction of this type of gradient material in the context of fiber reinforcements in the further course of this paper, see Section~\ref{sec:condensation} for the final formulation and Remarks \ref{remark_6} for a detailed discussion on the arising third order tensor $\mathfrak{P}$. }

\subsection{Continuum degenerate beam formulation}
In this section, we degenerate the general continuum mechanical framework as introduced above to a beam formulation. As we intend to embed fibers with a length-to-diameter ratio of 20 as standard for e.g. fiber reinforced polymers, it is reasonable that we restrict the kinematics of the 3-dimensional continuum along the fiber direction to a beam-like kinematic. In particular, we use the theory of geometrically exact beams, also known as Cosserat beam, introduced in \cite{cosserat1909}, see also \cite{rubin2010}.

\paragraph{Beam kinematic.}
Let us consider a beam as a 3-dimensional body, occupying $\tilde{\Omega}_0\subset\mathbb{R}^3$ with the following kinematical ansatz for the position field in the reference configuration
\begin{equation}\label{eq:fib:position_ref}
\tilde{\vec{X}}(\theta^{\alpha},s) = \tilde{\vec{\varphi}}_0(s) + \theta^{\alpha}\,\vec{D}_{\alpha}(s),
\end{equation}
parametrized in terms of $s\in [0,\,L]$ along the center line of the beam with length $L$. Moreover, an orthonormal triad $[\vec{D}_1,\,\vec{D}_2,\,\vec{D}_3]$ with convective coordinates $(\theta^{1},\theta^{2},s)\in\tilde{\Omega}_0$, is introduced, where $\vec{D}_{\alpha}$, $\alpha=1,2$, span the cross-section plane of the beam. Note that we assume a straight initial beam throughout the paper for the sake of a clear presentation, i.e.,\ the orthonormal triad is constant along the beam axis.

The motion of the geometrically exact beam is the restricted position field
\begin{equation}\label{eq:fib:position}
\tilde{\vec{x}}(\theta^{\alpha},s) = \tilde{\vec{\varphi}}(s) + \theta^{\alpha}\,\vec{d}_{\alpha}(s).
\end{equation}
Here, the orthonormal triad $\vec{d}_i$ is related to the reference triad via the rotation tensor $\tilde{\vec{R}}\in SO(3)$, i.e., $\tilde{\vec{R}} = \vec{d}_i\otimes\vec{D}_i$. 
The deformation gradient now reads \cite{ortigosa2016}
\begin{equation}\label{def_grad_beam}
\tilde{\vec{F}} = \tilde{\vec{\varphi}}^{\prime}\otimes\vec{D}_3 + \vec{d}_{\alpha}\otimes\vec{D}_{\alpha} + \theta^{\alpha}\,\vec{d}_{\alpha}^{\prime}\otimes\vec{D}_3,
\end{equation}
where $(\bullet)^{\prime}$ represents the derivative with respect to $s$.  
\UK{In case of a straight initial beam, when $\vec{D}_{\alpha}^{\prime}\equiv 0$,} the last equation can be rewritten as follows
\begin{equation}\label{def_grad_beam_2}
\tilde{\vec{F}} = \tilde{\vec{R}}\,\left(\vec{\Gamma}\otimes\vec{D}_3 + \left[\vec{K}\right]_\times\,\theta^{\alpha}\,\vec{D}_{\alpha}\otimes\vec{D}_3+\vec{I}\right),
\end{equation}
using the strain measure (\review{cf.}\ \cite{betsch2002b})
\begin{equation}\label{eq:def_Gamma_K}
\vec{\Gamma} = \tilde{\vec{R}}\tp\,\tilde{\vec{\varphi}}^{\prime}-\vec{D}_3,\quad\left[\vec{K}\right]_\times = \tilde{\vec{R}}\tp\,\tilde{\vec{R}}^{\prime}.
\end{equation}
The strain measure $\vec{\Gamma}$ is known as the axial-shear strain vector, whereas the curvature represented by $\vec{K}$ is called the torsional-bending strain vector. We refer to \cite{eugster2014} for a detailed analysis of the contravariant components of the effective curvature. For the strain energy to be defined subsequently, it is useful to introduce the polar decomposition of the deformation gradient via
\begin{equation}\label{eq_polar}
\tilde{\vec{F}} = \tilde{\vec{R}}\,\tilde{\vec{U}}, \quad\tilde{\vec{U}} = \vec{I}+\vec{B}\otimes\vec{D}_3, \quad \vec{B} = \vec{\Gamma} + (\vec{K}\times\theta^{\alpha}\,\vec{D}_{\alpha}). 
\end{equation}
With regard to \eqref{eq:so3}, the rotation tensor can be defined in terms of a rotation vector $\vec{\phi}\in\mathbb{R}^3$ via the exponential map $\tilde{\vec{R}}(\vec{\phi}) = \exp{\left[\vec{\phi}\right]_\times}$. A closed form expression is given by Rodrigues formula
\begin{equation}
\tilde{\vec{R}}(\vec{\phi}) = \vec{I} + \frac{\operatorname{sin}\|\vec{\phi}\|}{\|\vec{\phi}\|}\,\left[\vec{\phi}\right]_\times + \frac{1}{2}\,\left(\frac{\operatorname{sin}\left(\|\vec{\phi}\|/2\right)}{\|\vec{\phi}\|/2}\right)^2\,\left[\vec{\phi}\right]_\times^2.
\end{equation}
According to~\cite{meier2019geometrically}, the relationship between the spatial variation of $\delta\vec{\phi}$ and $\delta\fib{\ten{R}}$ is given by $\delta\fib{\ten{R}} = \left[\delta\vec{\phi}\right]_\times\fib{\ten{R}}$. 

\UK{
Alternatively, the rotation tensor~$\tilde{\vec{R}}$ can be also parameterized in terms of unit quaternions $\mathfrak{q}\in \mathcal{S}^3 = \{(q_0,\vec{q})\;|\; q_0\in\mathbb{R},\;\vec{q}\in\mathbb{R}^3, \;q_0^2+\vec{q}\cdot\vec{q} = 1\}$, by the formula 
\begin{equation}\label{eq:EulerRodrigues}
\tilde{\vec{R}}(\mathfrak{q}) = (2q_0^2-1)\,\vec{I} + 2q_0\left[\vec{q}\right]_\times + 2\,\vec{q}\otimes\vec{q},
\end{equation}
called Euler-Rodrigues parametrization, see \cite{betsch2009c,marsden2003} for details. 
}

\paragraph{Strain energy for large deformation beam.}
With regard to \eqref{eq_polar}, we can write our strain energy function as $\Psi(\vec{F},\vec{H},J)\, \hat{=}\, \Psi(\vec{B}(\vec{\Gamma},\vec{K}))$. In \cite{ortigosa2016}, we have shown that the Mooney-Rivlin material model presented in \eqref{eq:mooney} can be rewritten after some technical calculations in terms of the beam strain measures as
\begin{equation}\label{eq:beamMR}
\begin{aligned}
\tilde{\Psi}_{MR}(\vec{B}(\vec{\Gamma},\vec{K})) = \int\limits_{A(s)}&\left[\alpha\,(\vec{B}\cdot\vec{B}+2\,\vec{B}\cdot\vec{D}_3 + 3)\right.\, \\ 
&+ \beta\,(2\,\vec{B}\cdot\vec{B}-(\vec{D}_3\times\vec{B})\cdot(\vec{D}_3\times\vec{B}) + 4\,\vec{B}\cdot\vec{D}_3 + 3)\\
&+\left.f(\vec{B}\cdot\vec{D}_3 + 1)\right]\d^2\theta,
\end{aligned}
\end{equation}
where $A(s)$ denotes the local cross section of the beam, such that the internal energy can be written as a line integral
\begin{equation}
\tilde{\Pi}^{int} = \int\limits_{\mathfrak{C}_0}\tilde{\Psi}_{MR}(\vec{B}(\vec{\Gamma},\vec{K})) \d s,
\end{equation}
\review{where $\mathfrak{C}_0$ denotes the centerline of the beam.}
\review{A common model} is written directly in terms of $\vec{\Gamma}$ and $\vec{K}$:
\begin{equation}\label{eq:simpleMat}
\tilde{\Psi}(\vec{\Gamma},\vec{K}) = \frac{1}{2}\,\vec{\Gamma}\cdot\left(\mathbb{K}_1\,\vec{\Gamma}\right) + \frac{1}{2}\,\vec{K}\cdot\left(\mathbb{K}_2\,\vec{K}\right),
\end{equation}
\review{where $\mathbb{K}_1$ and $\mathbb{K}_2$ are diagonal $3\times3$ matrices given in terms of the material constants and the cross-section geometrical parameters~\cite{steinbrecher2019c,Weeger2017}.}
%$\mathbb{K}_1 = \op{Diag}[GA_1,\,GA_2,\,EA]$ and $\mathbb{K}_2 = \op{Diag}[EI_1,\,EI_2,\,GJ]$.
The corresponding constitutive relations in the reference and actual configuration \review{read} 
\begin{equation}
\tilde{\vec{N}} = \frac{\partial\tilde{\Psi}}{\partial\vec{\Gamma}},\quad\tilde{\vec{M}} = \frac{\partial\tilde{\Psi}}{\partial\vec{K}},\quad
\fib{\vec{n}} = \tilde{\vec{R}}\,\fib{\vec{N}},\quad
\fib{\vec{m}} = \tilde{\vec{R}}\,\fib{\vec{M}}.
\end{equation}
%The principle of virtual work of the beam is given in Appendix \ref{app:virt_work_beam}.
\review{
	The virtual work of the internal energy contributions then writes:
	\begin{equation}
	\delta\tilde{\Pi}^{int} = \int\limits_{\mathfrak{C}_0}\left[\fib{\vec{n}}\cdot\delta\tilde{\vec{\varphi}}^{\prime} 
	- \left(\fib{\vec{\varphi}}^\prime \times \fib{\vec{n}}\right)\cdot\delta\vct{\phi}
	+ \fib{\vec{m}}\cdot\delta\vct{\phi}^\prime\right]\d s.
	\end{equation}
} 

Introducing external \review{distributed} forces $\bar{\tilde{\vec{n}}}$ and couples $\bar{\tilde{\vec{m}}}$ \review{as well as external endpoint forces~$\vec{n}_{ext}^e$ and couples~$\vec{m}_{ext}^e$}, where we assume the existence of an external energy
%\rev{external forces and couples ?}
\begin{equation}
\tilde{\Pi}^{ext} = -\int\limits_{\mathfrak{C}_0}\left(\bar{\tilde{\vec{n}}}\cdot\tilde{\vec{\varphi}} + \bar{\tilde{\vec{m}}}\cdot\vct{\phi}\right)\d s
- \review{\EC{\vec{n}_{ext}^e\cdot\tilde{\vec{\varphi}} + \vec{m}_{ext}^e\cdot\vct{\phi}}}
,
\end{equation}
we obtain from integration by parts 
%the local form
\review{the classical balance equations in the actual configuration, \review{cf.}\ Antman \cite{antman1995},}
\begin{equation}\label{eq:pdvaBeam}
\begin{aligned}
-\int\limits_{\mathfrak{C}_0}\delta\tilde{\vec{\varphi}}\cdot\left(\tilde{\vec{n}}^{\prime} + \bar{\tilde{\vec{n}}}\right) \d s&= 0,\\
-\int\limits_{\mathfrak{C}_0}\delta\vct{\phi}\cdot\left(\tilde{\vec{m}}^{\prime} + \tilde{\vec{\varphi}}^{\prime}\times\tilde{\vec{n}} + \bar{\tilde{\vec{m}}}\right) \d s&= 0,
\end{aligned}
\end{equation}
\review{for all $(\delta\tilde{\vec{\varphi}}, \delta\vct{\phi})\in \tilde{\TestSpace}$,  a suitable functional space of kinematically admissible variations,
%	vanishing at the centerline endpoints,
with boundary conditions $\tilde{\vec{n}}|_{\Gamma^{0,L}}=\vec{n}_{ext}^e$ and $\tilde{\vec{m}}|_{\Gamma^{0,L}}=\vec{m}_{ext}^e$.
}
%\begin{equation}\label{eq:test_space2}
%\tilde{\TestSpace} = \{(\delta\tilde{\vec{\varphi}}, \delta\vct{\phi}):\, \mathfrak{C}_0 \rightarrow \mathbb{R}^3\times\mathbb{R}^3\;|\;\;\delta\tilde{\vec{\varphi}}|_{\Gamma^{0,L}} = \vec{0}, \; \delta\vec{\phi}|%_{\Gamma^{0,L}} = \vec{0}\},
%\end{equation}
% with homogeneous Dirichlet boundary conditions at the centerline endpoints. 
%Equations \eqref{eq:pdvaBeam} are the classical balance equations in the actual configuration, \review{cf.}\ Antman \cite{antman1995}. Additionally, we obtain information on the Neumann boundaries of the beam, i.e.,\ on the endpoints as 
%\begin{equation}\label{eq:BCBeam}
%\EC{\delta\tilde{\vec{\varphi}}^{e}\cdot\left(\tilde{\vec{n}}-\vec{n}_{ext}^e\right)} = 0\qquad\text{and}\qquad\EC{\delta\vct{\phi}^e\cdot\left(\tilde{\vec{m}}-\vec{m}_{ext}^e\right)} = 0,
%\end{equation}
%where $\vec{n}_{ext}^e$ and $\vec{m}_{ext}^e$ are appropriate external forces and couples at the endpoints.

\paragraph{The method of weighted residual.}
To avoid shear locking, a mixed, Hu-Washizu type method is employed. Therefore, we introduce additional, independent fields for the resultant spatial contact force vector ${\mathfrak{n}}$ and contact couple vector ${\mathfrak{m}}$ as well as the spatial axial-shear strain vector ${\mathfrak{g}}$ and the torsional-bending strain vector ${\mathfrak{k}}$. Now we assume a strain energy function of the form
\begin{equation}
\tilde{\Psi}_{HW} = \review{\tilde{\Psi}(\tilde{\vec{R}}\tp{\mathfrak{g}},\tilde{\vec{R}}\tp{\mathfrak{k}})} + {\mathfrak{n}}\cdot\left(\vec{\gamma}-\mathfrak{g}\right) + {\mathfrak{m}}\cdot\left(\vec{k}-\mathfrak{k}\right)
\end{equation}
where 
$\vec{\gamma} = \tilde{\vec{R}}\,\vec{\Gamma}$ and $\vec{k} = \tilde{\vec{R}}\,\vec{K}$,
%\review{and $\tilde{\Psi}_a({\mathfrak{g}},{\mathfrak{k}})=\tilde{\Psi}(\tilde{\vec{R}}\tp{\mathfrak{g}},\tilde{\vec{R}}\tp{\mathfrak{k}})$ is the strain energy in the actual configuration.}
The principle of virtual work yields now
\begin{equation}\label{eq:pdvaHW}
\begin{aligned}
\int\limits_{\mathfrak{C}_0}&\left[-\,\delta\tilde{\vec{\varphi}}\cdot\left({\mathfrak{n}}^{\prime} + \bar{\tilde{\vec{n}}}\right) -
\delta\vec{\phi}\cdot\left({\mathfrak{m}}^{\prime} + \tilde{\vec{\varphi}}^{\prime}\times{\mathfrak{n}} + \bar{\tilde{\vec{m}}}\right) +
\delta{\mathfrak{g}}\cdot\left( \review{\frac{\partial\tilde{\Psi}(\tilde{\vec{R}}\tp{\mathfrak{g}},\tilde{\vec{R}}\tp{\mathfrak{k}})}{\partial{\mathfrak{g}}}} - {\mathfrak{n}} \right) \right.  \\
&\left.+\,\delta{\mathfrak{k}}\cdot\left(\review{\frac{\partial\tilde{\Psi}(\tilde{\vec{R}}\tp{\mathfrak{g}},\tilde{\vec{R}}\tp{\mathfrak{k}})}{\partial{\mathfrak{k}}}}-{\mathfrak{m}} \right) +
\delta{\mathfrak{n}}\cdot\left({\vec{\gamma}}-{\mathfrak{g}}\right) + \delta{\mathfrak{m}}\cdot\left(\vec{k}-{\mathfrak{k}}\right) \right]\d s = 0.
%\quad\review{\text{\sout{$\forall\,\delta\tilde{\vec{\varphi}},\delta\vec{\phi},\delta{\mathfrak{g}},\delta{\mathfrak{k}},\delta{\mathfrak{n}},\delta{\mathfrak{m}}.$}}}
\end{aligned}
\end{equation}
%\review{
%\begin{equation}\label{eq:pdvaHW}
%\begin{aligned}
%\int\limits_{\mathfrak{C}_0}&\left[-\,\delta\tilde{\vec{\varphi}}\cdot\left({\mathfrak{n}}^{\prime} + \bar{\tilde{\vec{n}}}\right) -
%\delta\vec{\phi}\cdot\left({\mathfrak{m}}^{\prime} + \tilde{\vec{\varphi}}^{\prime}\times{\mathfrak{n}} + \bar{\tilde{\vec{m}}}\right) +
%\delta{\mathfrak{g}}\cdot\left( \frac{\partial\tilde{\Psi}({\mathfrak{g}},{\mathfrak{k}})}{\partial{\mathfrak{g}}}-\tilde{\vec{R}}\tp{\mathfrak{n}} \right) \right.  \\
%&\left.+\,\delta{\mathfrak{k}}\cdot\left(\frac{\partial\tilde{\Psi}({\mathfrak{g}},{\mathfrak{k}})}{\partial{\mathfrak{k}}}-\tilde{\vec{R}}\tp{\mathfrak{m}} \right) +
%\delta{\mathfrak{n}}\cdot\left({\vec{\Gamma}}-{\mathfrak{g}}\right) + \delta{\mathfrak{m}}\cdot\left(\vec{K}-{\mathfrak{k}}\right) \right]\d s = 0.
%%\quad\review{\text{\sout{$\forall\,\delta\tilde{\vec{\varphi}},\delta\vec{\phi},\delta{\mathfrak{g}},\delta{\mathfrak{k}},\delta{\mathfrak{n}},\delta{\mathfrak{m}}.$}}}
%\end{aligned}
%\end{equation}
%}
Assuming that the last two variations with respect to $\delta{\mathfrak{n}}$ and $\delta{\mathfrak{m}}$ in \eqref{eq:pdvaHW} are fulfilled locally, we can rewrite the mixed formulation as follows
\begin{equation}\label{eq:pdvaHW2}
\begin{aligned}
\int\limits_{\mathfrak{C}_0}\biggl[-\,\delta\tilde{\vec{\varphi}}\cdot\left({\mathfrak{n}}^{\prime} + \bar{\tilde{\vec{n}}}\right) - 
\delta\vec{\phi}\cdot\left({\mathfrak{m}}^{\prime} + \tilde{\vec{\varphi}}^{\prime}\times{\mathfrak{n}} + \bar{\tilde{\vec{m}}}\right)\phantom{\d s}&\\
+\,\delta{\mathfrak{g}}\cdot\left( \tilde{\vec{R}}\,\frac{\partial\tilde{\Psi}(\vec{\Gamma},\vec{K})}{\partial\vec{\Gamma}}-{\mathfrak{n}} \right) + 
\delta{\mathfrak{k}}\cdot\left(\tilde{\vec{R}}\,\frac{\partial\tilde{\Psi}(\vec{\Gamma},\vec{K})}{\partial\vec{K}}-{\mathfrak{m}} \right)\biggl]\d s &= 0,
%\quad\forall\,\delta\tilde{\vec{\varphi}},\delta\vec{\phi},\delta{\mathfrak{g}},\delta{\mathfrak{k}},
\end{aligned}
\end{equation}
\review{where $\vec{\Gamma}$ and $\vec{K}$ are defined in \eqref{eq:def_Gamma_K}.

% and we make use of the spaces of admissible strains $(\delta{\mathfrak{g}},\delta{\mathfrak{k}})\in 
%\hat{\TestSpace}=\{(\delta{\mathfrak{g}},\delta{\mathfrak{k}}):\, \mathfrak{C}_0 \rightarrow \mathbb{R}^3\times\mathbb{R}^3\}$.
% with
%\begin{equation}\label{eq:test_space3}
%\hat{\TestSpace} = \{(\delta{\mathfrak{g}},\delta{\mathfrak{k}}):\, \mathfrak{C}_0 \rightarrow \mathbb{R}^3\times\mathbb{R}^3\;|\;\;\delta{\mathfrak{g}}|_{\Gamma^{0,L}} = \vec{0}, \; \delta{\mathfrak{k}}|_{\Gamma^{0,L}} = \vec{0}\}.
%\end{equation}
%Note that \eqref{eq:test_space3} ensures that \eqref{eq:BCBeam} is valid.
}

\section{Main results: multidimensional coupling model}\label{sec:mainresults}

In this work, starting from a surface-to-volume coupling formulation for the matrix/beam system, we derive a reduced surrogate model where the new coupling \BW{constraints 
%defined on the beam centerline involve the deformation gradient and 
allow to transfer both the linear forces and the moments of the beam to the matrix.} 
The reduced model is based on the following assumptions:
\begin{itemize}
\item[A1. \textit{Overlapping domains}.] 
The matrix continuum is extended inside the beam domain.
This yields a modeling error due to additional stiffness, which depends on the fiber radius~$r$ and the ratio of the matrix stiffness to the fiber stiffness.

\item[A2. \textit{Form of the coupling force}.]
We make an assumption on the form of the coupling force (Lagrange multiplier), such that, applied to a beam cross-section, it presents the mean force, couple stresses, shear and dilatation forces.

%\review{
%	\item[A3. \textit{Thin beam}.]
%We also assume that the beam is thin enough to neglect the contribution of the end faces.
%}

%\item[A3. \textit{Sufficient regularity}.] 	
%We assume additional regularity of the solution, such that the traces of deformation and of its gradient are defined on the beam centerline. We will then truncate Taylor expansion of the deformation in the cross-section with respect to the radius~$r$. Considering all terms up to the linear part yields a modeling error of order~$\O\left(r^2\right)$.
\end{itemize}

\BW{	Estimation of the error due to the overlapping domains can be found, e.g., in~\cite{koppl2018mathematical} for the case of diffusion problem. Further information in the context of fluid-structure interaction problems can be obtained in, e.g., Liu et al. \cite{liu2007}.  In general,  it is possible to remove the additional matrix contributions using the chosen continuum degenerated beam formulation, as this allows to calculate the same type of stresses for the beam as given for the matrix material.  However, the kinematical assumption of a restricted cross sectional area within the beam will remain. From the application point of view, we are mainly interested in fiber reinforced polymers where the difference in the stiffness between fibers and matrix is quite high and the radius of the fiber is quite small, thus the modeling error in the assumption A1 can be neglected. }

% !TeX root = ../../MAIN.tex

%\def\cameraangle{105}
%\tdplotsetmaincoords{66}{\cameraangle} % orientation of camera
%\pgfmathsetmacro{\beginangle}{\cameraangle}
%\pgfmathsetmacro{\endangle}{\cameraangle - 180}
%\pgfmathsetmacro{\beginangle}{\cameraangle}
%\pgfmathsetmacro{\endangle}{\cameraangle - 180}

\def\rodlength{4}
\def\rodradius{0.2}
\def\rodcolor{none}
\def\a{0.7}
\def\b{1}
\def\deformangle{45}
\def\mylinewidth{0.8pt}

\tikzset{pics/system/.style={code={
%					
%			% draw potatoid

			\draw [line width=\mylinewidth] plot [smooth cycle, tension=.7] 
			coordinates {(-1,0.5) (0,-1.5) (\rodlength/2,-1.2) (\rodlength,-1.5) (\rodlength+1.5,-0.5) (\rodlength+1,1.5) (\rodlength-1,2.5) (1,1.5)};

%			% draw rod

			\draw [line width=\mylinewidth, color=black,fill=\rodcolor]
			plot[domain=90:270] 
			({\a*\rodradius*cos(\x)}, {\b*\rodradius*sin(\x)}) ;
			
			\draw [line width=\mylinewidth, dashed, color=black,fill=\rodcolor]
			plot[domain=-90:90] 
			({\a*\rodradius*cos(\x)}, {\b*\rodradius*sin(\x)}) ;
			\draw [line width=\mylinewidth, color=black,fill=\rodcolor]
			plot[domain=0:360] 
			({\rodlength + \a*\rodradius*cos(\x)}, {\b*\rodradius*sin(\x)}); 
			\draw [line width=\mylinewidth, color=black,fill=\rodcolor]
			( 0, \rodradius ) --  ( \rodlength, \rodradius ) 	;		
			
			\draw [line width=\mylinewidth, color=black,fill=\rodcolor]
			( 0, -\rodradius ) --  ( \rodlength, -\rodradius ) 	;

			\draw [line width=0.8pt, color=red,fill=\rodcolor]
			( 0, 0 ) --  ( \rodlength, 0 ) ;
%		
			
%			-- plot[domain=0:360] 
%			({\a*\rodradius*cos(\x)}, {\b*\rodradius*sin(\x)},\rodheight) --cycle;
%			\draw [line width=0.8pt, color=black,fill=\rodcolor]
%			plot[domain=0:360]  
%			({\rodradius*cos(\x)}, {\rodradius*sin(\x)},\rodheight) ;
}}}

%\tikzset{pics/rodaxis/.style={code={
%			\draw [line width=0.5pt, blue, line cap=round, dash pattern=on 12pt off 2pt on \the\pgflinewidth off 2pt]
%			(0,0,0.4pt) -- ( 0, 0, \rodheight+0.4pt);
%}}}

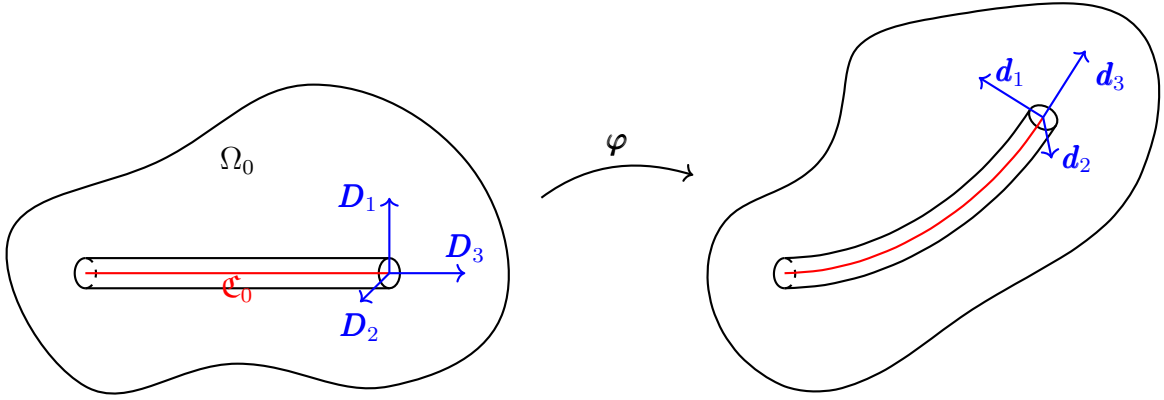
\begin{figure}[ht!]
	\centering\noindent%
\def\mycale{1}
\begin{tikzpicture}
[
scale=\mycale, every node/.style={scale=\mycale},
axis/.style={->,blue,thick},
vector/.style={-stealth,red,very thick},
vector guide/.style={dashed,red,thick}]
%\node (Cyl) [cylinder, shape border rotate=0, draw, minimum height=120mm, minimum width=12mm] {};
%\draw [<->] ([yshift=-5pt]Cyl.before bottom) -- ([yshift=-5pt]Cyl.after top) node [midway, below] {$L$};
%\draw [<->] (Cyl.center) -- (Cyl.south) node [midway, left] {$r$};
%\draw [->, very thick] (Cyl.center) -- ([xshift=30pt]Cyl.center) node [midway, above,yshift=-2pt] {$\vec{D}_3$};
%\draw[] ([xshift=-15pt]Cyl.bottom) node {$A_0$};
%\draw[] ([xshift=15pt]Cyl.top) node {$A_L$};
%\draw[] ([yshift=8pt]Cyl.north) node {$\Gamma_C$};

%%% initial system
\pic (initial) {system} ;

\node[red] (c0) at (\rodlength/2,-\rodradius) {$\mathfrak{C}_0$};
\node (Omega) at (\rodlength/2,1.5) {$\Omega_0$};

%draw axes
\draw[axis] (\rodlength,0,0) -- (\rodlength+1,0,0) node[anchor=south]{$\vec{D}_3$};
\draw[axis] (\rodlength,0,0) -- (\rodlength,1,0) node[anchor=east]{$\vec{D}_1$};
\draw[axis] (\rodlength,0,0) -- (\rodlength,0,1) node[anchor=north]{$\vec{D}_2$};

%%% map
\draw [line width=\mylinewidth, ->] (\rodlength+\rodlength/2,1) .. controls  (\rodlength+4/6*\rodlength,1.5) and (\rodlength+5/6*\rodlength,1.5) .. (\rodlength+\rodlength,1.3) node [midway, above] {$\vec{\varphi}$};

%%% Deformed system
\def\myshift{2.3*\rodlength}
\def\myh{0.01}
\scoped{
	\pgfsetcurvilinearbeziercurve
	{\pgfpointxy{\myshift}{0}}
%	{\pgfpointxy{\myshift+\rodlength}{0}}
	{\pgfpointxy{\myshift+\rodlength*0.5}{0}}
	{\pgfpointxy{\myshift+\rodlength}{\rodlength*0.5}}
	{\pgfpointxy{\myshift+\rodlength}{\rodlength}}
	\pgftransformnonlinear{\pgfgetlastxy\x\y%
		\pgfpointcurvilinearbezierorthogonal{\x}{\y}}%
	\pic (deformed) {system};
	\coordinate (C) at (\rodlength, 0) {};
	\coordinate (C3) at (\rodlength-\myh, 0) {};
	\coordinate (C1) at (\rodlength, -\myh) {};
	\coordinate (C2) at (\rodlength, 0, -\myh) {};
	
}
%\draw [line width=\mylinewidth] (\myshift,0) .. controls  (\myshift+\rodlength, 0) and (\myshift+\rodlength,\rodlength) .. (\myshift+\rodlength,\rodlength);
%\draw [line width=\mylinewidth] (\myshift,0) .. controls  (\myshift+\rodlength,0) .. (\myshift+\rodlength,\rodlength);

\coordinate (d3) at ($ (C) - (C3)$);
\coordinate (d1) at ($ (C) - (C1)$);
\coordinate (d2) at ($ (C) - (C2)$);

%draw axes
\draw[axis] (C) -- ($ (C) + 1.0/\myh*(d3)$) node[anchor=north west]{$\vec{d}_3$};
\draw[axis] (C) -- ($ (C) + 1.0/\myh*(d1)$) node[midway, anchor=south]{$\vec{d}_1$};
\draw[axis] (C) -- ($ (C) + 1.0/\myh*(d2)$) node[anchor=west]{$\vec{d}_2$};

\end{tikzpicture}

\caption{
	Deformation of the matrix/beam coupled system.
}
\label{fig:matrix-fiber}
\end{figure}

Let us consider the matrix/beam system illustrated in Figure \ref{fig:matrix-fiber}.
The coupling conditions are defined via Lagrange multipliers on the beam mantle~$\Gamma_C$ and at the end faces $A_0$, $A_L$ (see Figure~\ref{fig:fiber3D_ref}).
The work associated to the coupling forces over the whole matrix/beam interface reads
\begin{equation}\label{eq:CouplingTerm}
\Pi_{\Gamma} = \int\limits_{\review{\partial\tilde{\Omega}_0}}\vec{\mu}\cdot(\vec{\varphi}-\fib{\vec{x}})\d A =\Pi_{C} + \Pi_{A},
\end{equation}
where
\begin{equation}\label{eq:coupling_term_init}
\Pi_{C} = \int\limits_{\mathfrak{C}_0} \int\limits_{C(s)}\vec{\mu}\cdot(\vec{\varphi}-\fib{\vec{x}})\d C\d s,
\qquad
\Pi_{A} = 
%\sum_{s=0,L} 
\EC{\int\limits_{A(s)}\vec{\mu}\cdot(\vec{\varphi}-\fib{\vec{x}})\d A}.
\end{equation}
Above, $A(s)$ and $C(s)$ are the beam cross-section and its boundary corresponding to the arc length $s$\review{, respectively.}
The area and the circumference of the cross-section are denoted as $\abs{A} := \pi\, r^2$ and $\abs{C} := 2\,\pi\, r$\review{, respectively.}
Here and further, for the sake of shortness, we use the following abuse of notations: for any function $f(\X)$ defined in~\review{$\Omega_0$},
$f(\theta^\alpha, s):=f(\fib{\X}(\theta^\alpha, s))$ means the restriction of $f$ onto the beam domain~\review{$\fib{\Omega}_0$} in convective coordinates; and $f(s):=f(\fib{\vec{\varphi}}_0(s))$, $s\in\mathfrak{C}_0$, is understood in the sense of the trace $f|_{\mathfrak{C}_0}$ of $f$ on the beam centerline.

% !TeX root = ../../MAIN.tex

%\def\cameraangle{105}
%\tdplotsetmaincoords{66}{\cameraangle} % orientation of camera
%\pgfmathsetmacro{\beginangle}{\cameraangle}
%\pgfmathsetmacro{\endangle}{\cameraangle - 180}
%\pgfmathsetmacro{\beginangle}{\cameraangle}
%\pgfmathsetmacro{\endangle}{\cameraangle - 180}

\def\rodlength{120mm}
\def\rodradius{6mm}
\def\rodcolor{none}
\def\a{0.7}
\def\b{1}
\def\deformangle{45}
\def\mylinewidth{0.8pt}

%\tikzset{pics/system/.style={code={
%%					
%%			% draw potatoid
%
%			\draw [line width=\mylinewidth] plot [smooth cycle, tension=.7] 
%			coordinates {(-1,0.5) (0,-1.5) (\rodlength/2,-1.2) (\rodlength,-1.5) (\rodlength+1.5,-0.5) (\rodlength+1,1.5) (\rodlength-1,2.5) (1,1.5)};
%
%
%%			% draw rod
%
%			\draw [line width=\mylinewidth, color=black,fill=\rodcolor]
%			plot[domain=90:270] 
%			({\a*\rodradius*cos(\x)}, {\b*\rodradius*sin(\x)}) ;
%			
%			\draw [line width=0.8pt, dashed, color=black,fill=\rodcolor]
%			plot[domain=-90:90] 
%			({\a*\rodradius*cos(\x)}, {\b*\rodradius*sin(\x)}) ;
%%
%			\draw [line width=0.8pt, color=black,fill=\rodcolor]
%			plot[domain=0:360] 
%			({\rodlength + \a*\rodradius*cos(\x)}, {\b*\rodradius*sin(\x)}); 
%%			
%			\draw [line width=0.8pt, color=black,fill=\rodcolor]
%			( 0, \rodradius ) --  ( \rodlength, \rodradius ) 	;		
%			
%			\draw [line width=0.8pt, color=black,fill=\rodcolor]
%			( 0, -\rodradius ) --  ( \rodlength, -\rodradius ) 	;
%			
%			
			
%		
			
%			-- plot[domain=0:360] 
%			({\a*\rodradius*cos(\x)}, {\b*\rodradius*sin(\x)},\rodheight) --cycle;
%			\draw [line width=0.8pt, color=black,fill=\rodcolor]
%			plot[domain=0:360]  
%			({\rodradius*cos(\x)}, {\rodradius*sin(\x)},\rodheight) ;
%}}}

%\tikzset{pics/rodaxis/.style={code={
%			\draw [line width=0.5pt, blue, line cap=round, dash pattern=on 12pt off 2pt on \the\pgflinewidth off 2pt]
%			(0,0,0.4pt) -- ( 0, 0, \rodheight+0.4pt);
%}}}

\begin{figure}[ht!]
	\centering\noindent%
%
%
%%% Initial 
\def\mycale{0.9}
\begin{tikzpicture}
[
scale=\mycale, every node/.style={scale=\mycale},
axis/.style={->,blue,thick},
vector/.style={-stealth,red,very thick},
vector guide/.style={dashed,red,thick}]
\node (Cyl) [cylinder, aspect=2, shape border rotate=0, draw, minimum height=\rodlength, minimum width=2*\rodradius] {};
\draw [<->] ([yshift=-5pt]Cyl.before bottom) -- ([yshift=-5pt]Cyl.after top) node [midway, below] {$L$};

\draw[] ([xshift=-15pt]Cyl.bottom) node {$A_0$};
\draw[] ([xshift=15pt]Cyl.top) node {$A_L$};
\draw[] ([yshift=8pt]Cyl.north east) node {$\Gamma_C$};

\draw[dashed]
let \p1 = ($ (Cyl.after bottom) - (Cyl.before bottom) $),
\n1 = {0.5*veclen(\x1,\y1)-\pgflinewidth},
\p2 = ($ (Cyl.bottom) - (Cyl.after bottom)!.5!(Cyl.before bottom) $),
\n2 = {veclen(\x2,\y2)-\pgflinewidth}
in
([yshift=\pgflinewidth] Cyl.before bottom) arc [start angle=-90, end angle=90,
x radius=\n2, y radius=\n1];

\coordinate (C) at ([yshift=\pgflinewidth, xshift=-20mm] Cyl.center);

\draw[dotted] 
let \p1 = ($ (Cyl.after bottom) - (Cyl.before bottom) $),
\n1 = {0.5*veclen(\x1,\y1)-\pgflinewidth},
\p2 = ($ (Cyl.bottom) - (Cyl.after bottom)!.5!(Cyl.before bottom) $),
\n2 = {veclen(\x2,\y2)-\pgflinewidth}
in
([yshift=-\rodradius]C) arc [start angle=-90, end angle=270,
x radius=\n2, y radius=\n1]  ;
\node[above] at ([yshift=\rodradius]C)  {$C(s)$};;

\draw [<->] (C) -- ([yshift=-\rodradius]C) node [midway, left] {$r$};
\draw [->, very thick] (C) -- ([xshift=30pt]C) node [midway, above,yshift=-2pt] {$\vec{D}_3$};

\end{tikzpicture}

\caption{
	3D beam reference configuration.
}
\label{fig:fiber3D_ref}
\end{figure}
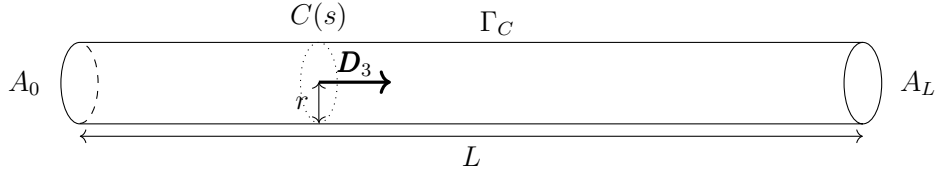

\subsection{Coupling term model}
The Lagrange multiplier~$\vec{\mu}$, physically interpreted as the interface load, is defined on the beam mantle $\Gamma_C = \lbrace (r\,\cos\theta, r\,\sin\theta, s) \; | \, s\in[0,L], \theta\in[0, 2\,\pi]\rbrace$, where $r$ is the beam radius and $\theta=\arctan(\theta^2/\theta^1)$ denotes the angle along the beam cross-section circumference.
Thus, $\vec{\mu}=\vec{\mu}(\theta, s)$ can be decomposed in Fourier series as periodic function of the angle~$\theta\in[0, 2\,\pi]$.
Here, we base our reduced model on truncation of this series after the first two terms.
That is, we assume the interface load in the form:
\begin{equation}\label{eq:LM_mu}
\vec{\mu}(\theta, s) = \bar{\vec{\mu}}(s) + \tilde{\vec{\mu}}_1(s)\,\cos\theta + \tilde{\vec{\mu}}_2(s)\,\sin\theta,
\end{equation}
where $\bar{\vec{\mu}}(s) = \frac{1}{2\,\pi}\int_{0}^{2\pi}\vec{\mu}(\theta, s)\d\theta$ is the mean cross-sectional load, and $\tilde{\vec{\mu}}_{\alpha}$, $\alpha=1,2$, are the first cosine and sine Fourier coefficients of $\vec{\mu}$\review{, respectively.}
These two terms compose the load fluctuation
\begin{equation}\label{eq:new:mu_tilde}
	\tilde{\vec{\mu}}(\theta,s)
	:= \ten{\Sigma}(s)\,\vec{N}(\theta)
	= \tilde{\vec{\mu}}_1(s)\,\cos\theta + \tilde{\vec{\mu}}_2(s)\,\sin\theta,
\end{equation}
where $\ten{\Sigma}=\tilde{\vec{\mu}}_{\alpha}\otimes\vec{D}_{\alpha}$ denotes the interface stress tensor, and $\vec{N}=\vec{D}_1\,\cos\theta + \vec{D}_2\,\sin\theta$ is the unit outer normal to the interface in the reference configuration. Note that $\int_{0}^{2\pi}\tilde{\vec{\mu}}(\theta,s)\d\theta = 0$.

While $\bar{\vec{\mu}}$ corresponds to the resulting force in the cross-section and, therefore, couples position of the beam center-line to the matrix, $\tilde{\vec{\mu}}$ contains the coupled stresses, shears and hydrostatic pressure, and is thus responsible for transition of the beam bending and torsion to the matrix.
In the following remark, we discuss in detail the structure of $\tilde{\vec{\mu}}$ and $\ten{\Sigma}$.

\begin{remark}\label{th:Sigma_structure}
	Due to the definition~\eqref{eq:new:mu_tilde}, the stress tensor~$\ten{\Sigma}$ has only six degrees of freedom.
	In particular, denoting $\ten{\mathcal{D}}:=\vec{D}_j\otimes\vec{e}_j$, we write $\ten{\Sigma}$ in the following form, involving coupled stresses, shears and dilatation:
	\begin{equation}\label{eq:new:LM_matrix}
	\ten{\Sigma} =
	\fib{\vec{R}}\,\ten{\mathcal{D}}
	\begin{bmatrix}
	\tilde{\mu}_{n,1} & \tilde{\mu}_{n,3} - \tilde{\mu}_{\tau,3} & 0\\
	\tilde{\mu}_{n,3} + \tilde{\mu}_{\tau,3} & \tilde{\mu}_{n,2} & 0\\
	-2\,\tilde{\mu}_{\tau,2} &  2\,\tilde{\mu}_{\tau,1} & 0
	\end{bmatrix}
	\ten{\mathcal{D}}\tp,
	\end{equation}
where $\tilde{\mu}_{n,1}$, $\tilde{\mu}_{n,2}$ and $\tilde{\mu}_{n,3}$ are associated to the dilatation and the shear of the cross-section, $\tilde{\mu}_{\tau,1}$ and $\tilde{\mu}_{\tau,2}$ to the bending, and $\tilde{\mu}_{\tau,3}$ to the torsion.	
Let $\tilde{\vec{\mu}}_n = \tilde{\mu}_{n,j}\,\vec{d}_j$
and $\tilde{\vec{\mu}}_\tau = \tilde{\mu}_{\tau,j}\,\vec{d}_j = \axl\left(\ten{\Sigma}\,\fib{\vec{R}}\tp\right)$.	
Then, we have
\begin{equation}\label{eq:PQ_decomp}
\ten{\Sigma} = 
\left(
	\ten{\mathcal{P}}_{\alpha}\,\tilde{\vec{\mu}}_\tau 
+ 	\ten{\mathcal{Q}}_{\alpha}\,\tilde{\vec{\mu}}_n
\right)\otimes\vec{D}_\alpha,
\end{equation}
where
\begin{equation}\label{eq:shortcuts}
\begin{aligned}
\ten{\mathcal{P}}_1 &= \phantom{+}\vec{d}_2\otimes\vec{d}_3 - 2\,\vec{d}_3\otimes\vec{d}_2, 
&\qquad
\ten{\mathcal{Q}}_1 &= \vec{d}_1\otimes\vec{d}_1 + \vec{d}_2\otimes\vec{d}_3,
\\ 
\ten{\mathcal{P}}_2 &= -\vec{d}_1\otimes\vec{d}_3 + 2\,\vec{d}_3\otimes\vec{d}_1,
&\qquad
\ten{\mathcal{Q}}_2 &= \vec{d}_2\otimes\vec{d}_2 + \vec{d}_1\otimes\vec{d}_3.
\end{aligned}
\end{equation}
Moreover, the tensor 
\begin{equation}\label{mu_sym}
	[\tilde{\vec{\mu}}_n]_s := \fib{\vec{R}}\tp\,\ten{\mathcal{Q}}_{\alpha}\,\tilde{\vec{\mu}}_n\otimes\vec{D}_\alpha = \tilde{\mu}_{n,1}\,\vec{D}_1\otimes\vec{D}_1 + \tilde{\mu}_{n,2}\,\vec{D}_2\otimes\vec{D}_2 + \tilde{\mu}_{n,3}\,\left(\vec{D}_1\otimes\vec{D}_2 + \vec{D}_2\otimes\vec{D}_1\right)
\end{equation}
is symmetric. 
\end{remark}

\review{
	Let us introduce the circular means
\begin{equation}\label{eq:phi_bar}
\BW{\vec{\varphi}_{c}} := \frac{1}{\abs{C}}\int\limits_{C(s)}\vec{\varphi}\d C
\qquad\text{and}\qquad
\BW{\ten{F}_{c}} := \BW{\frac{1}{\abs{A}}}\int\limits_{C(s)} \vec{\varphi} \otimes \vec{N}\d C,
\end{equation}
and Fourier cosine and sine coefficients 
\begin{equation}
\hat{\vec{\varphi}}_1 := \BW{\frac{1}{\abs{A}}}\int\limits_{C(s)}\vec{\varphi}\,\cos\theta\d C
\qquad\text{and}\qquad
\hat{\vec{\varphi}}_2 := \BW{\frac{1}{\abs{A}}}\int\limits_{C(s)}\vec{\varphi}\,\sin\theta\d C,
\end{equation}
respectively. Then, the following lemma gives the representation of the energy contribution due to the coupling on the beam mantle~$\Gamma_C$.
\begin{lemma}\label{th:CouplingTerm}
	If $\vec{\mu}$ is in form~\eqref{eq:LM_mu}, the coupling energy term~\eqref{eq:coupling_term_init}$_1$ writes
\begin{equation}\label{eq:coupling_term_lemma}
	\begin{aligned}
	\Pi_{C}
	&=
	\int\limits_{\mathfrak{C}_0} \vec{\mu}\cdot\left(\vec{\varphi}-\tilde{\vec{x}}\right) \abs{C}\d s
	\\
	&=
	\int\limits_{\mathfrak{C}_0} \bar{\vec{\mu}}\cdot\left(\BW{\vec{\varphi}_{c}}-\tilde{\vec{\varphi}}\right) \abs{C}\d s
	+
	\int\limits_{\mathfrak{C}_0} \ten{\Sigma}:\left(\BW{\ten{F}_{c}}-\fib{\vec{R}}\right) \BW{\abs{A}}\d s
	\\
	&=	
	\int\limits_{\mathfrak{C}_0} \bar{\vec{\mu}}\cdot\left(\BW{\vec{\varphi}_{c}}-\tilde{\vec{\varphi}}\right) \abs{C}\d s
	+
	\int\limits_{\mathfrak{C}_0} \tilde{\vec{\mu}}_{\alpha}\cdot\left(\hat{\vec{\varphi}}_{\alpha}-\vec{d}_{\alpha}\right) \BW{\abs{A}}\d s.
	\end{aligned}
	\end{equation}	
\end{lemma}
}

\begin{proof}
	Substitution of \eqref{eq:LM_mu} and \eqref{eq:new:mu_tilde} to \eqref{eq:coupling_term_init}$_1$ yields
	\begin{equation}\label{eq:coupling_term_proof}
	\Pi_{C}
	=
	\int\limits_{\mathfrak{C}_0}\int\limits_{C(s)}\underbrace{\left(\bar{\vec{\mu}} +  \ten{\Sigma}\,\vec{N}\right)}_{\vec{\mu}}
	\cdot\biggl(\vec{\varphi}-\underbrace{\left(\tilde{\vec{\varphi}}+\fib{\vec{R}}\,\vec{N}\,r\right)}_{\tilde{\vec{x}}}\biggr) \d C\d s
	.
	\end{equation}
	Given $\vec{N} = \vec{D}_{1}\cos\theta + \vec{D}_{2}\sin\theta$, we have $\int_{C(s)}\vec{N}\d C = \vec{0}$ and $\int_{C(s)}\vec{N}\otimes\vec{N}\d C = \pi\, r\,\vec{D}_{\alpha}\otimes\vec{D}_{\alpha}$.
	In particular, \BW{$\int_{C}\vec{N}\otimes\vec{N}\d C = r\int_{0}^{2\pi} \vec{N}\otimes\vec{N}\,\d\theta$ and}
	\begin{equation}
	\begin{split}	
	\int\limits_{0}^{2\pi} \vec{N}\otimes\vec{N}\,\d\theta
	=\,&\vec{D}_1\otimes\vec{D}_1\,\underbrace{\int\limits_{0}^{2\pi}\cos^2\theta \,\d\theta}_{=\pi}
	+\,
	\vec{D}_2\otimes\vec{D}_2\,\underbrace{\int\limits_{0}^{2\pi}\sin^2\theta\d\theta}_{=\pi} \\
	&+ \left[\vec{D}_1\otimes\vec{D}_2+\vec{D}_2\otimes\vec{D}_1\right]\,\underbrace{\int\limits_{0}^{2\pi}\sin\theta\,\cos\theta\,\d\theta}_{=0}
	= \pi\,\vec{D}_\alpha\otimes\vec{D}_\alpha.
	\end{split}
	\end{equation}
%	\begin{equation}\label{eq:int_N}
%		\int\limits_{C(s)}\vec{N}\d C = \vec{0}
%		\qquad\text{and}\qquad
%		\int\limits_{C(s)}\vec{N}\otimes\vec{N}\d C = \pi\,\ten{P}_{12}.
%	\end{equation}
	Thus, taking into account that $\ten{\Sigma}\,\vec{D}_{\alpha}\otimes\vec{D}_{\alpha} = \ten{\Sigma}$, \eqref{eq:coupling_term_proof} writes
	\begin{equation}
	\begin{aligned}
	\Pi_{C}
	&=
	\int\limits_{\mathfrak{C}_0}\bar{\vec{\mu}}\cdot\int\limits_{C(s)}\left(\vec{\varphi}-\tilde{\vec{\varphi}}\right)\d C\d s
	+ 
	\int\limits_{\mathfrak{C}_0}\ten{\Sigma}:\int\limits_{C(s)}\left(\vec{\varphi}-\fib{\vec{R}}\,\vec{N}\,r\right)\otimes\vec{N} \d C\d s
	\\
	&=
	\int\limits_{\mathfrak{C}_0} \bar{\vec{\mu}}\cdot\left[\int\limits_{C(s)}\vec{\varphi}\d C-\tilde{\vec{\varphi}}\abs{C}\right] \d s
	+
	\int\limits_{\mathfrak{C}_0} \ten{\Sigma}:\left[\int\limits_{C(s)}\vec{\varphi}\otimes\vec{N}\d C - \pi \,r^2\,\fib{\vec{R}}\right]\d s
	.
	\end{aligned}	
	\end{equation}
	Hence if follows~\eqref{eq:coupling_term_lemma}$_2$.
	Finally, \eqref{eq:coupling_term_lemma}$_3$ is obtained by substitution of $\ten{\Sigma} = \tilde{\vec{\mu}}_{\alpha}\otimes\vec{D}_{\alpha}$ into \eqref{eq:coupling_term_lemma}$_2$.
\end{proof}

\BW{
	Thus, the virtual work associated to the kinematic constraints is obtained as variation of~\eqref{eq:coupling_term_lemma} in the following form:
	\begin{equation}\label{eq:VirtWork_coupling}
	\begin{aligned}
%	\delta\Pi_\Gamma
%	\approx
	\delta\Pi_C
	=&\,
	\int\limits_{\mathfrak{C}_0}\bigl[\delta\bar{\vec{\mu}}\cdot(\vec{\varphi}_{c} -\tilde{\vec{\varphi}})
	+
	\bar{\vec{\mu}}\cdot(\delta\vec{\varphi}_{c} -\delta\tilde{\vec{\varphi}})\bigr]\abs{C}\d s 
	\\
	&+\,
	\int\limits_{\mathfrak{C}_0}\left[
	\ten{\Sigma}:\big(\delta\ten{F}_{c}-\left[\delta\vec{\phi}\right]_\times\fib{\vec{R}}\big)
	+
	\delta\ten{\Sigma}:\big(\ten{F}_{c}-\fib{\vec{R}}\big)
	\right]\abs{A}\d s,
	\end{aligned}
	\end{equation}
	where we note that $\ten{\Sigma}:\left(\VPTen{\delta\vec{\phi}}\,\fib{\ten{R}}\right) = 2\,\tilde{\vec{\mu}}_\tau\cdot\delta\vec{\phi}$.
}

%%%%%%%%%%%%%%%%%%%%%%%%%%%%%%%%%%%%%%%%%%%%%%%%%%%%%%%%%%%%%%%%%%%%%%%%%%%%%%%%%%%%%%%%%%%%%%%%%%%%%%%%%
% End faces
%%%%%%%%%%%%%%%%%%%%%%%%%%%%%%%%%%%%%%%%%%%%%%%%%%%%%%%%%%%%%%%%%%%%%%%%%%%%%%%%%%%%%%%%%%%%%%%%%%%%%%%%%

\BW{	
	Eventually, at the end faces, the associated coupling energy term can be written in the following form:
	\begin{equation}\label{eq:PiA_init}
	\Pi_{A}
	=
	\EC{\int_{A_s} \vec{\mu}\cdot\left(\vec{\varphi}-\tilde{\vec{x}}\right)\d A}
	=
	\EC{\bar{\vec{\mu}}_e\cdot\left(\vec{\varphi}_e-\tilde{\vec{\varphi}}\right)\abs{A_e} + \O(r^3)},
	\end{equation}
	where we denoted
	\begin{equation}\label{eq:face_avg}
	\bar{\vec{\mu}}_e := \frac{1}{\abs{A}}\int_{A_e}\vec{\mu}\d A,
	\qquad
	\vec{\varphi}_e := \frac{1}{\abs{A}}\int_{A_e}\vec{\varphi}\d A.
	\end{equation}
	We formalize the area factor $\abs{A_e}$ in~\eqref{eq:PiA_init} in such a way that it takes value of $\abs{A}=\pi\, r^2$ only in case when the associated face~$A_e$ is embedded into the matrix, otherwise $\abs{A_e}=0$.
	Under the assumption that the fiber is thin, we neglect the $\O(r^3)$-term.
	Then, the associated virtual work reads
	\begin{equation}\label{eq:VirtWork_coupling_A}
	\delta\Pi_A
	=\,
	\EC{\,\delta\bar{\vec{\mu}}_e\cdot(\vec{\varphi}_e -\tilde{\vec{\varphi}})\abs{A_e}
	+
	\bar{\vec{\mu}}_e\cdot(\delta\vec{\varphi}_e -\delta\tilde{\vec{\varphi}})\abs{A_e}\,}.
	\end{equation}
	In case when the face~$A_e$ is embedded into the matrix, this corresponds to Dirichlet boundary condition $\tilde{\vec{\varphi}} = \vec{\varphi}_e$ at the beam endpoint, imposed via Lagrange multiplier in $\mathfrak{n}|_{\Gamma^{0,L}} = \vec{n}_{ext}^{e} + \bar{\vec{\mu}}_e\abs{A_e}$.
	If the face~$A_e$ is not embedded into the matrix, we have only Neumann boundary condition $\mathfrak{n}|_{\Gamma^{0,L}}=\vec{n}_{ext}^{e}$ instead.	
	Note that taking into account higher order terms in~$r$ in~\eqref{eq:PiA_init}, one can consider additional Dirichlet boundary conditions for the beam directors.
}

\subsection{Total coupled system}
Finally, we obtain the principle of virtual work for the  coupled matrix/beam system
\begin{equation}
%\delta\Pi^{int} + \delta\Pi^{ext} + \delta\fib{\Pi}^{int} + \delta\fib{\Pi}^{ext} + \delta\Pi_C + \delta\Pi_A = 0,
\review{
	\delta\Pi^{total}(\vec{\varphi},\fib{\vec{\varphi}},\fib{\ten{R}},\mathfrak{n},\mathfrak{m},\bar{\vec{\mu}},\bar{\vec{\mu}}_e,\tilde{\vec{\mu}}_\tau,\tilde{\vec{\mu}}_n)
}
=
\delta\Pi^{int} + \delta\Pi^{ext} + \delta\fib{\Pi}^{int} + \delta\fib{\Pi}^{ext} + \review{\delta\Pi_C + \delta\Pi_A} = 0,
\end{equation}
in terms of \eqref{eq:pdvA} and \eqref{eq:pdvaHW2} along with the virtual work of the coupling forces in \eqref{eq:VirtWork_coupling}. Regarding the matrix material, we obtain
\begin{multline}\label{eq:sys:matrix}
\int\limits_{\mtx{\Omega}_0} \left[\drv{\mtx{\Psi}}{\ten{\mtx{F}}}:\nabla\delta\vec{\mtx{\varphi}} - \vec{B}_{ext}\cdot\delta\vec{\varphi}\right] \d V	
-
\int\limits_{\Gamma^{\sigma}} \vec{T}_{ext}\cdot\delta\vec{\mtx{\varphi}} \d A
\\
+\BW{
	\EC{\,\bar{\vec{\mu}}_e\cdot\delta\vec{\varphi}_e\abs{A_e}\,}}
+
\int\limits_{\mathfrak{C}_0}\left[\vec{\bar{\mu}} \cdot \BW{\delta\vec{\mtx{\varphi}}_{c}\abs{C}}
+
\ten{\Sigma}\review{(\tilde{\vec{\mu}}_\tau,\tilde{\vec{\mu}}_n)}:\BW{\delta\ten{F}_{c}\abs{A}} \right] \d s
= 0,
\end{multline}
%\review{where $\delta\vec{\mtx{\varphi}}\in\TestSpace$}.
\review{with boundary conditions $\vec{\varphi}|_{\Gamma^{\varphi}} = \vec{\varphi}_{\Gamma}$ and $\delta\vec{\varphi}|_{\Gamma^{\varphi}} = \vec{0}$, and $\ten{\Sigma}(\tilde{\vec{\mu}}_\tau,\tilde{\vec{\mu}}_n)$ given in the form~\eqref{eq:new:LM_matrix}}.
\BW{The circular means $\delta\vec{\varphi}_{c}$ and $\delta\ten{F}_{c}$ are defined as in~\eqref{eq:phi_bar}, and $\delta\vec{\varphi}_e$ as in~\eqref{eq:face_avg}.}
The beam contributions  yield
\begin{align}
&-\int\limits_{\mathfrak{C}_0}\left(
{\mathfrak{n}}^{\prime} + \bar{\tilde{\vec{n}}}
+
\bar{\vec{\mu}}\abs{C}\right)\cdot\delta\fib{\vec{\varphi}}\d s = 0,  
\label{eq:sys:loads}
\\	
&-\int\limits_{\mathfrak{C}_0}\left({\mathfrak{m}}^{\prime} + \tilde{\vec{\varphi}}^{\prime}\times{\mathfrak{n}}
+
\bar{\tilde{\vec{m}}}
+
\BW{2}\,\tilde{\vec{\mu}}_\tau\,\BW{\abs{A}}\right)\cdot\delta\vec{\phi} \d s = 0,
\label{eq:sys:moments}
\end{align}
%\review{where $(\delta\tilde{\vec{\varphi}}, \delta\vct{\phi})\in \tilde{\TestSpace}$},
\review{ with endpoint conditions
\BW{$\abs{A_e}\left(\tilde{\vec{\varphi}}-\vec{\varphi}_e\right)|_{\Gamma^{0,L}}=0$ ($\abs{A_e}=0$ if the face is not embedded), $\mathfrak{n}|_{\Gamma^{0,L}}=\vec{n}_{ext}^{e}+\bar{\vec{\mu}}_e\abs{A_e}$ }
and $\mathfrak{m}|_{\Gamma^{0,L}}=\vec{m}_{ext}^{e}$,}
%	 and $\delta\tilde{\vec{\varphi}}|_{\Gamma^{0,L}}=\delta\vct{\phi}|_{\Gamma^{0,L}}=\vec{0}$,}
% $\EC{\delta\tilde{\vec{\varphi}}\cdot\left(\mathfrak{n}-\vec{n}_{ext}^{e}\right)} = 0$ and $\EC{\delta\vct{\phi}\cdot\left(\mathfrak{m}-\vec{m}_{ext}^{e}\right)} = 0$, 
and supplemented by the weak constitutive equations
\begin{align}
&\int\limits_{\mathfrak{C}_0}\delta{\mathfrak{g}}\cdot\left(\fib{\ten{R}}\,\frac{\partial\tilde{\Psi}(\vec{\Gamma},\vec{K})}{\partial\vec{\Gamma}}-{\mathfrak{n}} \right)\d s = 0,
\label{eq:constitutive_n}
\\
&\int\limits_{\mathfrak{C}_0}\delta{\mathfrak{k}}\cdot\left( \fib{\ten{R}}\,\frac{\partial\tilde{\Psi}(\vec{\Gamma},\vec{K})}{\partial\vec{K}}-{\mathfrak{m}} \right)\,\text{d}s = 0,
\label{eq:constitutive_m}
\end{align}
as well as the coupling constraints
\begin{align}
\int\limits_{\mathfrak{C}_0}\delta\vec{\bar{\mu}} \cdot\left(\vec{\mtx{\varphi}}_c-\fib{\vec{\varphi}}\right)\abs{C}\d s &= 0,
%&\EC{\delta\bar{\vec{\mu}}_e\cdot(\vec{\varphi}_{c} -\tilde{\vec{\varphi}})}\abs{A} &= 0
\label{eq:sys:pi1}
\\
\int\limits_{\mathfrak{C}_0}\delta\ten{\Sigma}:\big(\vec{F}_{c}-\fib{\vec{R}}\big)\BW{\abs{A}}\d s &= 0.
%&\,\,\frac{1}{2}\,\EC{\delta\ten{\Sigma}^{e}:(\vec{F}_{c}-\fib{\vec{R}})}\abs{A} &= 0.
\label{eq:sys:pi2}
\end{align}

\review{
\begin{remark}\label{re:shear}
	Note that the variation of $\ten{\Sigma} = \vec{\mu}_{\alpha}\otimes\vec{D}_{\alpha}$ has the form $\delta\ten{\Sigma} = \delta\vec{\mu}_{\alpha}\otimes\vec{D}_{\alpha}$.
	Thus, the constraints~\eqref{eq:sys:pi2} can be rewritten as follows:
	\begin{equation}\label{key6}
	\int\limits_{\mathfrak{C}_0}
	\big(\vec{F}_{c}\,\vec{D}_{1}-\vec{d}_{1}\big)\cdot\delta\tilde{\vec{\mu}}_1\BW{\abs{A}}\d s = 0,
	\qquad
	\int\limits_{\mathfrak{C}_0}
	\big(\vec{F}_{c}\,\vec{D}_{2}-\vec{d}_{2}\big)\cdot\delta\tilde{\vec{\mu}}_2\BW{\abs{A}}\d s = 0.
	\end{equation}
	Alternatively, in line with the representation~\eqref{eq:PQ_decomp}, let $\delta\ten{\Sigma} = \left(\ten{\mathcal{P}}_{\alpha}\,\delta\tilde{\vec{\mu}}_\tau + \ten{\mathcal{Q}}_{\alpha}\,\delta\tilde{\vec{\mu}}_n\right)\otimes\vec{D}_\alpha$.
	Then, the constraints~\eqref{eq:sys:pi2} rewrite as 
	\begin{align}
	\int\limits_{\mathfrak{C}_0}\left(
	\ten{\mathcal{P}}_{\alpha}\tp\,\left[\vec{F}_{c}\,\vec{D}_{\alpha}-\vec{d}_{\alpha} \right]\right)
	\cdot\delta\tilde{\vec{\mu}}_{\tau}\BW{\abs{A}}\d s& = 0,
	\label{eq:coup_cond1}
	\\
	\label{eq:coup_cond2}
	\int\limits_{\mathfrak{C}_0} \left(
	\ten{\mathcal{Q}}_{\alpha}\tp\,\left[\vec{F}_{c}\,\vec{D}_{\alpha}-\vec{d}_{\alpha} \right]\right)
	\cdot\delta\tilde{\vec{\mu}}_{n}\BW{\abs{A}}\d s &= 0,
	\end{align}
	\UK{such that the bending and torsion constraints~\eqref{eq:coup_cond1} are separated from the dilatation and shear constraints~\eqref{eq:coup_cond2}.}
	\UK{
		Let us remark that owing to the directors orthogonality, $\ten{\mathcal{P}}_{\alpha}\tp\vec{d}_{\alpha}=0$ in~\eqref{eq:coup_cond1}.
	Moreover, it turns out that
}
%	In fact, 
	the last condition does not involve the beam at all, restraining only the stretches and the shears of the matrix, which will be shown in Lemma~\ref{th:area_cond_F} below.
	
\end{remark}
}

\subsection{Static condensation}\label{sec:condensation}
The above system is large, as we have to deal with the degrees of freedom of the matrix material (three per node in the 3-dimensional continuum), and the 21 unknowns including 9 Lagrange multipliers per node along the beam center line.  Thus, we aim at a two step static condensation procedure: first, we condense ~\eqref{eq:sys:loads}-\eqref{eq:sys:moments} and eliminate the corresponding Lagrange multipliers in the continuous system. The remaining equations for the constraints and the constitutive laws for the beam are condensed in the discrete setting, such that finally only the matrix degrees of freedom remain and the beam is fully condensed.

\BW{
\begin{proposition}[\textbf{Condensed system}]\label{th:condensed_sys1}
%Let $V=\{v\in\H{1}(\Omega_0),\text{s.t. } \vec{D}_\alpha\cdot \nabla v|_{\mathfrak{C}_0}\in\L{2}(\mathfrak{C}_0) \}$.
The system~\eqref{eq:sys:matrix}-\eqref{eq:sys:pi2} can be formally reduced to the following system of the unknowns $\vec{\varphi},\fib{\vec{\varphi}},\fib{\ten{R}},\mathfrak{n},\mathfrak{m}, \bar{\vec{\mu}}_e$ and $\tilde{\vec{\mu}}_n$:
\begin{equation}\label{eq:condensedCompact1}
\begin{aligned}
\int\limits_{\mtx{\Omega}_0}\bigl(\vec{P}:\nabla\delta\vec{\mtx{\varphi}}\, -\, \vec{B}_{ext}\cdot\delta\vec{\varphi}\bigl) \d V 
-\int\limits_{\Gamma^{\sigma}} \vec{T}_{ext}\cdot\delta\vec{\mtx{\varphi}} \d A 
+ \EC{\bar{\vec{\mu}}_e\cdot\delta\vec{\varphi}_e\abs{A_e}}
&\\
+\,\int\limits_{\mathfrak{C}_0}\left[
	\left(-\frac{1}{2}\,\ten{\mathcal{P}}_{\alpha}\left({\mathfrak{m}}^{\prime} + \tilde{\vec{\varphi}}^{\prime}\times{\mathfrak{n}}+\bar{\tilde{\vec{m}}}\right)\otimes\vec{D}_{\alpha} + \ten{F}_{c}\,[\tilde{\vec{\mu}}_n]_s\,\abs{A}
\right):\delta\ten{F}_{c}
- (\mathfrak{n}^\prime + \bar{\tilde{\vec{n}}}) \cdot \delta\mtx{\vec{\varphi}}_{c}
\right]\d s
&= 0,
\\[4mm]
\int\limits_{\mathfrak{C}_0}\delta{\mathfrak{g}}\cdot\left(\fib{\ten{R}}\,\frac{\partial\tilde{\Psi}(\vec{\Gamma},\vec{K})}{\partial\vec{\Gamma}}-{\mathfrak{n}} \right)\d s+\int\limits_{\mathfrak{C}_0}\delta{\mathfrak{k}}\cdot\left( \fib{\ten{R}}\,\frac{\partial\tilde{\Psi}(\vec{\Gamma},\vec{K})}{\partial\vec{K}}-{\mathfrak{m}} \right)\,\d s &=0,
\\[4mm]
\int\limits_{\mathfrak{C}_0}\delta\vec{\bar{\mu}} \cdot\left(\vec{\mtx{\varphi}}_c-\fib{\vec{\varphi}}\right)\abs{C}\,\d s
+
\int\limits_{\mathfrak{C}_0}
\UK{\ten{\mathcal{P}}_{\alpha}\tp\,\vec{F}_{c}\,\vec{D}_{\alpha}} \cdot\delta\tilde{\vec{\mu}}_\tau\abs{A}\, \d s 
%&
%\\
+\int\limits_{\mathfrak{C}_0}\left(\vec{F}_{c}\tp\,\vec{F}_{c}-\ten{I}\right):[\delta\tilde{\vec{\mu}}_n]_s\abs{A}\, \d s
&= 0,
\end{aligned}
\end{equation}
with boundary conditions $\vec{\varphi}|_{\Gamma^{\varphi}} = \vec{\varphi}_{\Gamma}$ and  $\delta\vec{\varphi}|_{\Gamma^{\varphi}} = \vec{0}$, and beam endpoint conditions $\abs{A_e}\left(\tilde{\vec{\varphi}}-\vec{\varphi}_e\right)|_{\Gamma^{0,L}}=0$ (where $\abs{A_e}=0$ if the face is not embedded), $\mathfrak{n}|_{\Gamma^{0,L}}=\vec{n}_{ext}^{e}+\bar{\vec{\mu}}_e\abs{A_e}$
%$\mathfrak{n}|_{\Gamma^{0,L}}=\vec{n}_{ext}^{e}$ 
and $\mathfrak{m}|_{\Gamma^{0,L}}=\vec{m}_{ext}^{e}$.
Recall that the operators~$\ten{\mathcal{P}}_{\alpha}$ and $[\;\cdot\;]_s$ are defined in Remark~\ref{th:Sigma_structure} \UK{and depend on $\fib{\ten{R}}$}.
\end{proposition}
}

%%%%%%%%%%%%%%%%%%%%%%%%%%%%%%%%%%%%%%%%%%%%%%%%%%%%%%%%%%%%%%%%%%%%%%%%%%%%

For the proof of Proposition~\ref{th:condensed_sys1}, let us first consider
the two following lemmas.

\begin{lemma}\label{th:area_cond_F}
	Condition \eqref{eq:coup_cond2} can be rewritten as
	\begin{equation}\label{eq:area_cond_F}
	\int\limits_{\mathfrak{C}_0} \left(\vec{F}_{c}\tp\,\vec{F}_{c}-\ten{I}\right):[\delta\tilde{\vec{\mu}}_{n}]_s\, \BW{\abs{A}}\,\d s = 0,
	\end{equation}
	with
	$
	[\delta\tilde{\vec{\mu}}_{n}]_s =\delta\tilde{\mu}_{n,1}\,\vec{D}_1\otimes\vec{D}_1 + \delta\tilde{\mu}_{n,2}\,\vec{D}_2\otimes\vec{D}_2
	+ \delta\tilde{\mu}_{n,3}\,(\vec{D}_1\otimes\vec{D}_2 + \vec{D}_2\otimes\vec{D}_1)
	$.
\end{lemma}

\begin{proof}
	Integrating in~\eqref{eq:sys:pi2} with $\delta\ten{\Sigma} = \fib{\ten{R}}\,[\delta\tilde{\vec{\mu}}_{n}]_s$ and $\delta\ten{\Sigma} = \ten{F}_{c}\,[\delta\tilde{\vec{\mu}}_{n}]_s$, and using symmetry of $[\delta\tilde{\vec{\mu}}_{n}]_s$, we obtain 
	\begin{equation}\label{eq:area_cond_F_proof}
	\int\limits_{\mathfrak{C}_0} \fib{\ten{R}}\tp\,(\ten{F}_{c}-\fib{\ten{R}}):[\delta\tilde{\vec{\mu}}_{n}]_s\,
	\BW{\abs{A}}\,\d s = 0,
	\qquad
	\int\limits_{\mathfrak{C}_0} (\ten{F}_{c}-\fib{\ten{R}})\tp\,\ten{F}_{c}:[\delta\tilde{\vec{\mu}}_{n}]_s\,
	\BW{\abs{A}}\,\d s = 0,
	\end{equation}
	respectively. Given $	\ten{F}_{c}\tp\,\ten{F}_{c}-\ten{I} = \fib{\ten{R}}\tp\,(\ten{F}_{c}-\fib{\ten{R}}) + (\ten{F}_{c}-\fib{\ten{R}})\tp\,\ten{F}_{c}$, summing up the integrals in~\eqref{eq:area_cond_F_proof} yields the statement.
\end{proof}

\begin{lemma}\label{th:R2F}
	The following equality holds:
	\begin{equation}\label{eq:R2F}
	\int\limits_{\mathfrak{C}_0}\ten{\mathcal{Q}}_{\alpha}\,\tilde{\vec{\mu}}_n
	\otimes\vec{D}_{\alpha}:\BW{\delta\ten{F}_{c}\abs{A}}\,\d s 
	=
	\int\limits_{\mathfrak{C}_0}\ten{F}_{c}\,[\tilde{\vec{\mu}}_{n}]_s:\BW{\delta\ten{F}_{c}\abs{A}}\,\d s 
	\end{equation}
\end{lemma}

\begin{proof}
	\review{Integration} in~\eqref{eq:sys:pi2} with $\delta\ten{\Sigma} = \BW{\delta\ten{F}_{c}}\,[\tilde{\vec{\mu}}_{n}]_s$, \review{using the definition~\eqref{mu_sym},} yields the statement.
\end{proof}

%%%%%%%%%%%%%%%%%%%%%%%%%%%%%%%%%%%%%%%%%%%%%%%%%%%%%%%%%%%%%%%%%%%%%%%%%%%

\BW{
\begin{proof}[Proof of Proposition~\ref{th:condensed_sys1}]	
Integrating~\eqref{eq:sys:loads} with the test functions~\BW{$\delta\fib{\vec{\varphi}} = \delta\vec{\varphi}_{c}$}, we sum it up with~\eqref{eq:sys:matrix} to obtain
	\begin{multline}\label{eq:new:sys:matrix_condense}
	\int\limits_{\mtx{\Omega}_0} \left(\drv{\mtx{\Psi}}{\ten{\mtx{F}}}:\nabla\delta\vec{\mtx{\varphi}}- \vec{B}_{ext}\cdot\delta\vec{\varphi}\right) \d V
	-
	\int\limits_{\Gamma^{\sigma}} \vec{T}_{ext}\cdot\delta\vec{\mtx{\varphi}} \d A
	\\
	+\EC{\bar{\vec{\mu}}_e\cdot\delta\vec{\varphi}_e\abs{A_e}}
	-\int\limits_{\mathfrak{C}_0}\left(\mathfrak{n}^\prime +
	\bar{\tilde{\vec{n}}} \right)\cdot \delta\mtx{\vec{\varphi}}_{c}\d s
	+
	\int\limits_{\mathfrak{C}_0}
	\ten{\Sigma}:\delta\ten{F}_{c}\abs{A}\d s
	= 0.
	\end{multline}
Thus, we have eliminated the set of Lagrange multipliers $\bar{\vec{\mu}}$. For the second set of Lagrange multipliers~$\ten{\Sigma}$, we integrate in~\eqref{eq:sys:moments} with the test functions~$\delta\vec{\phi}= \frac{1}{2}\,\ten{\mathcal{P}}_{\alpha}\tp\,\delta\ten{F}_{c} \,\vec{D}_{\alpha}$.
Taking the decomposition~\eqref{eq:PQ_decomp} and Lemma~\ref{th:R2F} into account,  it yields
	\begin{equation}\label{eq:expl_stress}
	-\int\limits_{\mathfrak{C}_0}\left[\frac{1}{2}\,\ten{\mathcal{P}}_{\alpha}
	\left({\mathfrak{m}}^{\prime} + \tilde{\vec{\varphi}}^{\prime}\times{\mathfrak{n}}+\bar{\tilde{\vec{m}}}\right)
	\otimes\vec{D}_{\alpha}
	+
	\left(\ten{\Sigma}
	-
	\ten{F}_{c}\,[\tilde{\vec{\mu}}_n]_s
	\right)	\abs{A}	
	\right]
	:
	\delta\ten{F}_{c}
	\d s
	= 0.
	\end{equation}
Then, summing it up with~\eqref{eq:new:sys:matrix_condense}, we obtain~\eqref{eq:condensedCompact1},
supplemented with constitutive equations~\eqref{eq:constitutive_n}-\eqref{eq:constitutive_m} and coupling constraints \eqref{eq:sys:pi1}, \eqref{eq:coup_cond1} and \eqref{eq:area_cond_F}.
\end{proof}

Let us remark that the existence of solution to the system~\eqref{eq:condensedCompact1} is an open question and beyond the scope of this work.
}

\subsection{Projection of the coupling constraints to the beam centerline}\label{sec:TaylorTruncation}

\BW{
	We obtained the coupled system of the 1D Cosserat beam and the 3D material matrix using the circular mean of the deformation~$\vec{\varphi}$ on the beam mantle $\Gamma_C$.
	Such approach is common for 1D-3D coupling models (see, e.g., \cite{d2008coupling,koppl2018mathematical}).
	Note that the trace of~$\vec{\varphi}\in\H{1}(\Omega_0)^3$ on the beam centerline~$\mathfrak{C}_0$ is not well-defined, while the circular mean $\vec{\varphi}_{c}\in\L{2}(\mathfrak{C}_0)^3$ according to the trace theorem.
	We point out that Galerkin $\Cn{1}$-conforming approximations allow us to deal with the traces of $\vec{\varphi}$ and $\nabla\vec{\varphi}$ on the beam centerline.
	In what follows, we assume the necessary regularity and make use of the following lemma exploiting formal Taylor expansion of $\vec{\varphi}$ around the centerline to approximate~$\vec{\varphi}_{c}$.

\begin{lemma}\label{th:expansion_in_r}
	Let $\vec{\varphi}$ be regular enough, such that the traces $\vec{\varphi}|_{\mathfrak{C}_0}(s)=\vec{\varphi}(\tilde{\vec{\varphi}}_0(s))$ and $\vec{F}|_{\mathfrak{C}_0}(s)=\vec{F}(\tilde{\vec{\varphi}}_0(s))$ are in~$\L{2}(\mathfrak{C}_0)$.
	Then, the circular means in~\eqref{eq:phi_bar} can be approximated as follows:
	\begin{equation}\label{eq:centerline_projection}
	\vec{\varphi}_{c} = \vec{\varphi}|_{\mathfrak{C}_0}	+ \O(r^2),
	\qquad
	\ten{F}_{c} = \ten{F}|_{\mathfrak{C}_0}\,\vec{D}_{\alpha}\otimes\vec{D}_{\alpha} + \O(r^2).	
	\end{equation}
\end{lemma}

\begin{proof}
	Using Fourier representation of $\vec{\varphi}$ in the plane containing $C(s)$ with the origin on the centerline and the formal power series of the exponential, we can write
	\begin{equation}\label{eq:FormalTaylor}
			\vec{\varphi}\bigr|_{C(s)} 
			= 
			\int\limits_{\R^2}\hat{\vec{\varphi}}(\vec{\omega},s)\, \exp{\i\,\vec{\omega}\cdot\vec{N}(\theta)\,r} \d\vec{\omega}
			=
			\int\limits_{\R^2}\hat{\vec{\varphi}}(\vec{\omega},s)\, \sum_{k=0}^{\infty}\frac{\left(\i\,\vec{\omega}\cdot\vec{N}(\theta)\,r\right)^k}{k!} \d\vec{\omega},
	\end{equation}	
%	Let us consider the formal Taylor expansion of $\vec{\varphi}$ on the mantle around the centerline:
%	\begin{equation}\label{eq:FormalTaylor}
%	\vec{\varphi}|_{C(s)} 
%	=
%	\vec{\varphi}|_{\mathfrak{C}_0}(s) + \vec{F}|_{\mathfrak{C}_0}(s)\cdot\vec{N}(\theta)\,r + \frac{1}{2}\left[\nabla\vec{F}|_{\mathfrak{C}_0}(s)\cdot\vec{N}(\theta)\,r\right]\cdot\vec{N}(\theta)\,r + \O(r^3),
%	\end{equation}
	where $\vec{N}(\theta) = \vec{D}_{1}\cos\theta + \vec{D}_{2}\sin\theta$.
	Then, substituting this expansion to the formula~\eqref{eq:phi_bar}$_1$, we get rid of the term $k=1$ in the sum, owing to $\int_{C(s)}\vec{N}\d C = \vec{0}$, which leads to~\eqref{eq:centerline_projection}$_1$.
	
	Next, let us also remind that $\int_{C(s)}\vec{N}\otimes\vec{N}\d C = \pi\, r\,\vec{D}_{\alpha}\otimes\vec{D}_{\alpha}$.
	Hence, we have
	\begin{equation}\label{eq:Taylor_grad_k1}
	\begin{aligned}
		\int\limits_{C(s)}\!\left[\;\int\limits_{\R^2}\hat{\vec{\varphi}}(\vec{\omega},s)\, \,(\i\,\vec{\omega}\cdot\vec{N}\,r)\d\vec{\omega}\right] \otimes\vec{N} \d C
		&=
		\int\limits_{\R^2}\hat{\vec{\varphi}}(\vec{\omega},s)\otimes (\i\,\vec{\omega})\d\vec{\omega}\cdot r\int\limits_{C(s)}\vec{N}\otimes\vec{N}\d C 
		\\	
		&=
		\pi\, r^2\,\nabla\vec{\varphi}|_{\mathfrak{C}_0}\,\vec{D}_{\alpha}\otimes\vec{D}_{\alpha}.
	\end{aligned}		
	\end{equation}
	Moreover, it holds $
	\int_{0}^{2\pi}\cos^3\theta\d\theta = \int_{0}^{2\pi}(1-\sin^2\theta)\d\,(\sin\theta) = 0$ and $
	\int_{0}^{2\pi}\cos^2\theta\sin\theta\d\theta = -\int_{0}^{2\pi}\cos^2\theta\d\,(\cos\theta) = 0
	$.
	Analogously, $\int_{0}^{2\pi}\sin^3\theta\d\theta=\int_{0}^{2\pi}\sin^2\theta\cos\theta\d\theta=0$.
	Therefore, $\int_{C(s)}(\vec{\omega}\cdot\vec{N})^2\,\vec{N}\d C = \vec{0}$.
	Thus, substituting the expansion~\eqref{eq:FormalTaylor} to the formula~\eqref{eq:phi_bar}$_2$, we get rid of the terms $k=0$ and $k=2$ in the sum.
	And using~\eqref{eq:Taylor_grad_k1} for the term $k=1$, we obtain~\eqref{eq:centerline_projection}$_2$.
\end{proof}

Note that the $\O(r^2)$-terms in~\eqref{eq:centerline_projection} yield respectively $\O(r^3)$- and $\O(r^4)$-terms in~\eqref{eq:coupling_term_lemma}.
Under the assumption that the fiber is thin enough, such that the above terms are of order of the discretization error, we neglect them.
Thus, we have $\vec{\varphi}_{c} \approx \vec{\varphi}|_{\mathfrak{C}_0}$ and $\ten{\Sigma}:\ten{F}_{c} \approx \ten{\Sigma}:\ten{F}|_{\mathfrak{C}_0}$, owing to $\ten{\Sigma}\,\vec{D}_{\alpha}\otimes\vec{D}_{\alpha}=\ten{\Sigma}$.
Since $\vec{D}_3$-component of the gradient is never used in the constraints, in what follows, we directly replace $\vec{\varphi}_{c}$ and $\ten{F}_{c}$ with $\vec{\varphi}|_{\mathfrak{C}_0}$ and $\ten{F}|_{\mathfrak{C}_0}$, respectively in the system~\eqref{eq:condensedCompact1}, which results to a second gradient material model (see Proposition~\ref{th:condensed_sys} below).
Note that further consideration of higher order terms in the expansion~\eqref{eq:FormalTaylor} will lead to higher gradient models.
}

\begin{proposition}[Second gradient model]\label{th:condensed_sys}
	%Let $V=\{v\in\H{1}(\Omega_0),\text{s.t. } \vec{D}_\alpha\cdot \nabla v|_{\mathfrak{C}_0}\in\L{2}(\mathfrak{C}_0) \}$.
	\review{Assuming necessary regularity, let us neglect $\O(r^2)$-terms in~\eqref{eq:centerline_projection}. Then, integrating by parts on the centerline~$\mathfrak{C}_0$ with the endpoint conditions, the system~\eqref{eq:condensedCompact1} can be written as the following second gradient model:}
	\begin{equation}\label{eq:condensedCompact}
	\begin{aligned}
	\int\limits_{\mtx{\Omega}_0}\bigl(\vec{P}:\nabla\delta\vec{\mtx{\varphi}}\, -\, \vec{B}_{ext}\cdot\delta\vec{\varphi}\bigl) \d V 
	-\int\limits_{\Gamma^{\sigma}} \vec{T}_{ext}\cdot\delta\vec{\mtx{\varphi}} \d A &\\
	+\,\int\limits_{\mathfrak{C}_0}\biggl(\mathfrak{P}\,\vdots\,\nabla^2\delta\vec{\varphi} + \bigl(\vec{P}^{\mathfrak{n}} + \vec{P}^{\mathfrak{m}} + \vec{P}^{g} + \ten{F}\,[\tilde{\vec{\mu}}_n]_s\,\abs{A}
	\bigr):\nabla\delta\vec{\varphi} 
	- \bar{\tilde{\vec{n}}} \cdot \delta\vec{\varphi}
	\biggl)\d s&\\
	-\EC{\vec{n}_{ext}^{e}\cdot\delta\vec{\varphi}	+	\frac{1}{2}\left(\ten{\mathcal{P}}_{\alpha}\,\vec{m}_{ext}^{e} \otimes\,\vec{D}_{\alpha}\right):\nabla\delta\vec{\varphi}}&= 0,
	\\[4mm]
	\int\limits_{\mathfrak{C}_0}\delta{\mathfrak{g}}\cdot\left(\fib{\ten{R}}\,\frac{\partial\tilde{\Psi}(\vec{\Gamma},\vec{K})}{\partial\vec{\Gamma}}-{\mathfrak{n}} \right)\d s+\int\limits_{\mathfrak{C}_0}\delta{\mathfrak{k}}\cdot\left( \fib{\ten{R}}\,\frac{\partial\tilde{\Psi}(\vec{\Gamma},\vec{K})}{\partial\vec{K}}-{\mathfrak{m}} \right)\,\d s &=0,
	\\[4mm]
	\int\limits_{\mathfrak{C}_0}\delta\vec{\bar{\mu}} \cdot\left(\vec{\mtx{\varphi}}-\fib{\vec{\varphi}}\right)\,\abs{C}\,\d s
	+
	\int\limits_{\mathfrak{C}_0}
	\UK{\ten{\mathcal{P}}_{\alpha}\tp\vec{F}\,\vec{D}_{\alpha}} 
	\cdot\delta\tilde{\vec{\mu}}_\tau\,\abs{A}\, \d s
%	&
%	\\
	+\int\limits_{\mathfrak{C}_0} \left(\vec{F}\tp\vec{F}-\ten{I}\right):[\delta\tilde{\vec{\mu}}_n]_s\,\abs{A}\, \d s
	&= 0,
	\end{aligned}
	\end{equation}
	%\review{for all $\delta\vec{\varphi}\in \TestSpace$, $\delta{\mathfrak{g}},\delta{\mathfrak{k}},\delta\vec{\bar{\mu}},\delta\tilde{\vec{\mu}}_\tau,\delta\tilde{\vec{\mu}}_n$,}
	\review{with boundary conditions $\vec{\varphi}|_{\Gamma^{\varphi}} = \vec{\varphi}_{\Gamma}$ and  $\delta\vec{\varphi}|_{\Gamma^{\varphi}} = \vec{0}$,}
	where the corresponding stresses are defined as follows:
	\begin{align}
	\vec{P} &:= \drv{\mtx{\Psi}}{\ten{\mtx{F}}}, \\
	\vec{P}^{\mathfrak{n}} &:= \mathfrak{n}\otimes\vec{D}_3, \label{eq:Pn}\\
	\vec{P}^{\mathfrak{m}} &:= \frac{1}{2}\left[\ten{\mathcal{P}}_{\alpha}^\prime\,\mathfrak{m}		
	-\ten{\mathcal{P}}_{\alpha}\left(\tilde{\vec{\varphi}}^{\prime}\times{\mathfrak{n}} +\bar{\tilde{\vec{m}}}\right)
	%		+ \ten{\mathcal{Q}}_{\alpha}\,\tilde{\vec{\mu}}_n
	\right]\otimes\vec{D}_{\alpha}, \label{eq:Pm}\\
	\vec{P}^{g} &:= \frac{1}{2}\,\ten{\mathcal{P}}_{\alpha}\,\mathfrak{m}\otimes\vec{D}_\alpha^\prime,
	\label{eq:Pg}
	\\
	\mathfrak{P} &:= \frac{1}{2}\,\ten{\mathcal{P}}_{\alpha}\,\mathfrak{m}\otimes\vec{D}_\alpha\otimes\vec{D}_3.\label{eq:thirdOrderDef}
	\end{align}
%	Recall that the tensors~$\ten{\mathcal{P}}_{\alpha}$ and $[\delta\tilde{\vec{\mu}}_n]_s$ are defined in Remark~\ref{th:Sigma_structure}.
\end{proposition}

\UK{Note that after integration by parts, the end faces terms are gone owing to the projection on the centerline.}
%The existence of solution to the system~\eqref{eq:condensedCompact1} is an open question and beyond the scope of this work.
For the remainder of this section, a series of remarks with further details on technical issues as well as the interpretation of the arising coupled stresses are provided.

\begin{remark}
	Note that the term $\vec{P}^{g}$ in~\eqref{eq:Pg} is due to the initial curvature of the beam.	
	In the case of a straight initial configuration, i.e.,\ when $\vec{D}_{\alpha}^\prime\equiv 0$, this term vanishes.	
%	Moreover, according to Appendix \ref{app:derivativeP}, we have in this case
%	$\ten{\mathcal{P}}_{\alpha}^\prime = \VPTen{\vec{\kappa}}\ten{\mathcal{P}}_{\alpha} - \ten{\mathcal{P}}_{\alpha}\VPTen{\vec{\kappa}}$.
	\UK{
		Moreover, we have in this case
		$\ten{\mathcal{P}}_{\alpha}^\prime = \VPTen{\vec{\kappa}}\ten{\mathcal{P}}_{\alpha} - \ten{\mathcal{P}}_{\alpha}\VPTen{\vec{\kappa}}$ with $\vec{\kappa} = \axl(\tilde{\ten{R}}^{\prime}\tilde{\ten{R}}\tp)$.
	Indeed, since we have $\vec{d}_i^{\prime} = \tilde{\vec{R}}^{\prime}\vec{D}_i = \tilde{\vec{R}}^{\prime}\tilde{\vec{R}}^T\vec{d}_i = \VPTen{\vec{\kappa}}$, direct differentiation in~\eqref{eq:shortcuts} results in
	\begin{equation}
	\begin{aligned}
	\ten{\mathcal{P}}_1^{\prime} 
	&=
	\left(\vec{d}_2^{\prime}\otimes\vec{d}_3 + \vec{d}_2\otimes\vec{d}_3^{\prime}\right) -2\left(\vec{d}_3^{\prime}\otimes\vec{d}_2 + \vec{d}_3\otimes\vec{d}_2^{\prime}\right)
	\\
	&= \left(\VPTen{\vec{\kappa}}\vec{d}_2\otimes\vec{d}_3 - \vec{d}_2\otimes\vec{d}_3\VPTen{\vec{\kappa}}\right) - 
	2\left(\VPTen{\vec{\kappa}}\vec{d}_3\otimes\vec{d}_2 - \vec{d}_3\otimes\vec{d}_2\VPTen{\vec{\kappa}}\right),\\
	&= \VPTen{\vec{\kappa}}\ten{\mathcal{P}}_1 - \ten{\mathcal{P}}_1\VPTen{\vec{\kappa}}.
	\end{aligned}
	\end{equation}
	For the second term $\ten{\mathcal{P}}_2^{\prime}$, we proceed analogously.
	}
\end{remark}

\UK{
\begin{remark}
	Let us note that
	\begin{equation}\label{key}
%	\begin{aligned}
	\ten{\mathcal{P}}_{\alpha}\tp\vec{F}\,\vec{D}_{\alpha}
%	&= \left(\vec{d}_2\otimes\vec{d}_3 - 2\,\vec{d}_3\otimes\vec{d}_2\right)\tp\vec{F}\,\vec{D}_1
%	+
%	\left(-\vec{d}_2\otimes\vec{d}_3 + 2\,\vec{d}_3\otimes\vec{d}_2\right)\tp\vec{F}\,\vec{D}_2
%	\\
%	&=
	= 
	\left[\vec{D}_2\fib{\ten{R}}\tp\vec{F}\vec{D}_1
	-
	\vec{D}_1\fib{\ten{R}}\tp\vec{F}\vec{D}_2\right]\,\vec{d}_3 
	-
	2\left[\vec{D}_3\fib{\ten{R}}\tp\vec{F}\vec{D}_1\right]\vec{d}_2
	+
	2\left[\vec{D}_3\fib{\ten{R}}\tp\vec{F}\vec{D}_2\right]\vec{d}_1.
%	\end{aligned}	
	\end{equation}
	Thus, the constraints associated to~$\delta\tilde{\vec{\mu}}_\tau$ can be seen as a planar (restricted on the beam cross-section) version of the spherical (omni-directional) rotational constraint 
%	$\skewsym(\fib{\ten{R}}\tp\mtx{\ten{F}})=0$
	\begin{equation}
		\axl(\fib{\ten{R}}\tp\mtx{\ten{F}})=0,
	\end{equation}
	which can be found, e.g., in~\cite{simo1992h}.
\end{remark}
}

\begin{remark}\label{remark_6}
	We now obtain in~\eqref{eq:condensedCompact} an explicit representation for the third order stress tensor introduced in~\eqref{eq:secondGradient}, which emanates from the condensation of the moment equation of the beams. Note that this term directly depends on the resultant torque~$\mathfrak{m}$, derived from, e.g.,\ \eqref{eq:simpleMat}, which is written in terms of the second and the polar areal moment,  respectively. This third order stress tensor is restricted to the line $\mathfrak{C}_0$ at the scale of the fiber, which is, in general,  in terms of [\textmu m] for short fiber reinforced polymers. For a more general strain gradient framework at a macroscopic scale, suitable homogenization procedures have to be applied.
\end{remark}

\section{Spatial discretization}\label{sec:discretization}
Concerning the spatial approximation, suitable isogeometric discretization schemes are employed for the different fields. In particular, a standard IGA approach is applied to the matrix as well as to the beam. 

\paragraph{Matrix material.}
A standard displacement-based finite element approach employs approximations of the deformation field $\vec{\varphi}$ and its variation of the form
\begin{equation}\label{eq:approxGeo}
\vec{\varphi}^{\mathrm{h}}=\sum\limits_{A\in\mathcal{I}}R^{A}\,\vec{q}_{A} \quad\text{and}\quad\delta\vec{\varphi}^{\mathrm{h}} =\sum\limits_{A\in\mathcal{I}}R^{A}\,\delta\vec{q}_{A},  
\end{equation}
respectively, where $\vec{q}_{A}\in\mathbb{R}^3$ and $\delta\vec{q}_{A}\in\mathbb{R}^3$. Here, \(R^A:\mathcal{B}_0^{\mathrm{h}}\rightarrow\mathbb{R}\) are NURBS based shape functions associated of order $p$ with control points \(A\in\mathcal{I}=\{1,\,\hdots,\,\mathfrak{M}\}\), where \(\mathfrak{M}\) denote the overall number of control points. The discrete form of \eqref{eq:pdvA} reads now
\begin{equation}\label{eq:pdvADis}
\delta\vec{q}_A\cdot\left(\,\int\limits_{\Omega_0}\vec{P}^{\mathrm{h}}\,\nabla R^A\d V - \int\limits_{\Omega_0}R^A\,\vec{B}^{\mathrm{h}}_{ext}\d V -
\int\limits_{\Gamma^{\sigma}}R^A\,\vec{T}^{\mathrm{h}}_{ext}\,\text{d}A \right)= 0,\quad\forall\,\delta\vec{q}_A,
\end{equation}
where $\vec{P}^{\mathrm{h}}$, $\vec{B}^{\mathrm{h}}_{ext}$ and $\vec{T}^{\mathrm{h}}_{ext}$ are the appropriately approximated Piola-Kirchhoff stress and the external contributions.

\paragraph{\CH{Beam}.}
For the beam, we require in a first step the approximations of the centerline and the rotations,
where we make use of quaternions $\mathfrak{q}$ for the parametrization \UK{of $\tilde{\vec{R}}$ as in~\eqref{eq:EulerRodrigues}}.
%\UK{parameterized in terms of quaternions $\mathfrak{q}$ as in~\eqref{eq:EulerRodrigues}.}
The approximation along the centerline is given by
\begin{equation}
\tilde{\vec{\varphi}}^{\mathrm{h}}=\sum\limits_{A\in\mathcal{J}}\tilde{R}^{A}\,\tilde{\vec{q}}_{A} \quad\text{and}\quad\mathfrak{q}^{\mathrm{h}}=\sum\limits_{A\in\mathcal{J}}\tilde{R}^{A}\,\mathfrak{q}_{A}.
\end{equation}
Here, $\tilde{R}^{A}:\mathfrak{C}_0^{\mathrm{h}}\rightarrow\mathbb{R}$ are 1-dimensional, NURBS based shape functions of order \review{$p_1$} with associated control points $A\in\mathcal{J} = \{1,\hdots,\mathfrak{O}\}$. For the mixed formulation as introduced before, additional approximations are required for the resultants
\begin{equation}
\mathfrak{n}^{\mathrm{h}}=\sum\limits_{A\in\mathcal{L}}\tilde{M}^{A}\,\mathfrak{n}_{A}\quad\text{and}\quad\mathfrak{m}^{\mathrm{h}}=\sum\limits_{A\in\mathcal{L}}\tilde{M}^{A}\,\mathfrak{m}_{A},
\end{equation}
where $\tilde{M}^{A}:\mathfrak{C}_0^{\mathrm{h}}\rightarrow\mathbb{R}$ are 1-dimensional, NURBS based shape functions of order \review{$p_2$} with associated control points $A\in\mathcal{L} = \{1,\hdots,\mathfrak{O}-1\}$. The test functions $\delta\tilde{\vec{\varphi}}^{\mathrm{h}}$ and $\delta\vec{\phi}^{\mathrm{h}}$ and in the context of the applied mixed method, $\delta\mathfrak{g}^{\mathrm{h}}$ and $\delta\mathfrak{k}^{\mathrm{h}}$ are discretized via a classical Bubnov-Galerkin approach. Thus, we obtain
\begin{equation}
\delta\tilde{\vec{\varphi}}^{\mathrm{h}} = \sum\limits_{A\in\mathcal{J}} \tilde{R}^{A}\,\delta\tilde{\vec{q}}_A\quad,\quad
\delta\vec{\phi}^{\mathrm{h}} = \sum\limits_{A\in\mathcal{J}} \tilde{R}^{A}\,\delta\vec{\phi}_A,
\end{equation}
and
\begin{equation}
\delta\mathfrak{g}^{\mathrm{h}} = \sum\limits_{A\in\mathcal{L}} \tilde{M}^{A}\,\delta\mathfrak{g}_A\quad,\quad
\delta\mathfrak{k}^{\mathrm{h}} = \sum\limits_{A\in\mathcal{L}} \tilde{M}^{A}\,\delta\mathfrak{k}_A.
\end{equation}
As the balance equations are condensed within the matrix equation, we will show the discrete contributions of the mixed beam formulation subsequently within the coupled system.

\paragraph{Coupling conditions.}
Next, we introduce suitable approximations for the variations of the Lagrange multipliers as
\begin{equation}\label{eq:sigmaDis}
\delta\bar{\vec{\mu}}^{\mathrm{h}} = \sum\limits_{A\in\mathcal{J}} \tilde{N}^{A}\,\delta\bar{\vec{\mu}}_A,
\end{equation}
along with the 6 variations 
\begin{equation}\label{eq:sigmaDis2}
\delta\tilde{\vec{\mu}}_{\tau}^{\mathrm{h}} = \sum\limits_{A\in\mathcal{J}}\tilde{N}^{A}\,\delta\tilde{\vec{\mu}}_{\tau,A},\quad \delta\tilde{\vec{\mu}}_n^{\mathrm{h}} = \sum\limits_{A\in\mathcal{L}} \tilde{N}^{A}\,\delta\tilde{\vec{\mu}}_{n,A},
\end{equation}
using the discretization of the beam, along which we define the Lagrange multipliers. Here, $\tilde{N}^{A}:\mathfrak{C}_0^{\mathrm{h}}\rightarrow\mathbb{R}$ are 1-dimensional, NURBS based shape functions
%, this time 
of order \review{$p_3$} with associated control points $A\in\mathcal{L} = \{1,\hdots,\mathfrak{O}-2\}$. Note that $\bar{\vec{\mu}}$ and $\tilde{\vec{\mu}}_{\tau}$ are already condensed and \CH{$\tilde{\vec{\mu}}_n^{\mathrm{h}} = \sum\limits_{A\in\mathcal{L}}\tilde{N}^{A}\,\tilde{\vec{\mu}}_{n,A}$} follows immediately from \eqref{eq:sigmaDis2}$_2$ for a Bubnov-Galerkin approach.

\begin{remark}
We have to note that the coupling condition in the form~\eqref{eq:coup_cond2} explicitly includes the rotation tensor~$\fib{\ten{R}}$.	Due to its orthogonality, this condition can be written only in terms of $\ten{F}$ (see Lemma~\ref{th:area_cond_F}). However, the quaternion approach,  which is free of singularities,  produces an error in the orthogonality of $\fib{\ten{R}}^{\mathrm{h}}$, as the unity constraints provide their own approximation error. Thus,~\eqref{eq:coup_cond2} couples $\mathfrak{q}^{\mathrm{h}}$ and $\ten{F}^{\mathrm{h}}$ providing an additional numerical error. These issues are avoided when using the condition in the form~\eqref{eq:area_cond_F}.
\end{remark}

\paragraph{Discrete system.}
The discretized condensed system as proposed in \eqref{eq:condensedCompact} reads
\begin{equation}\label{eq:red2}
\begin{aligned}
\int\limits_{\Omega_0}\vec{P}^{\mathrm{h}}\,\nabla R^A\,\text{d}V - \int\limits_{\Omega_0}R^A\,\vec{B}^{\mathrm{h}}_{ext}\,\text{d}V -
\int\limits_{\Gamma^{\sigma}}R^A\,\vec{T}^{\mathrm{h}}_{ext}\,\text{d}A &\\
+\int\limits_{\mathfrak{C}_0}\left(\mathfrak{P}^{\mathrm{h}}\,:\,\nabla^2R^A + 
\left(\vec{P}^{\mathfrak{n},{\mathrm{h}}}+\vec{P}^{\mathfrak{m},{\mathrm{h}}} + \vec{P}^{g,{\mathrm{h}}} + \BW{\vec{F}^{\mathrm{h}}\,[\tilde{\vec{\mu}}_n^{\mathrm{h}}]_s\,\abs{A^\mathrm{h}}}\right)\,\nabla R^A - 
R^A\,\bar{\tilde{\vec{n}}}^{\mathrm{h}} \right)\d s&\\
-\left[R^A\,\vec{n}_{ext}^{\CH{e,\mathrm{h}}} + \CH{\frac{1}{2}}\,\left(\ten{\mathcal{P}}^{\mathrm{h}}_{\alpha}\,\vec{m}_{ext}^{\CH{e,\mathrm{h}}} \otimes\,\vec{D}^{\mathrm{h}}_{\alpha}\right)\,\nabla R^A\right]\biggl|_0^L &= \vec{0},
\end{aligned}
\end{equation}
evaluated at every node $A$. Additionally, we require
\begin{equation}\label{eq:beamDis2}
\begin{aligned}
\vec{0} &= \int\limits_{\mathfrak{C}_0}
\tilde{M}^{A}\left[\tilde{\vec{R}}^{\mathrm{h}}\,\frac{\partial\tilde{\Psi}}{\partial\vec{\Gamma}}\left(\vec{\Gamma}^{\mathrm{h}},\vec{K}^{\mathrm{h}}\right)-\mathfrak{n}^{\mathrm{h}}\right]\d s,\\
\vec{0} &= \int\limits_{\mathfrak{C}_0}
\tilde{M}^{A}\left[\tilde{\vec{R}}^{\mathrm{h}}\,\frac{\partial\tilde{\Psi}}{\partial\vec{K}}\left(\vec{\Gamma}^{\mathrm{h}},\vec{K}^{\mathrm{h}}\right)-\mathfrak{m}^{\mathrm{h}}\right]\d s ,\\
0 &= \int\limits_{\mathfrak{C}_0}
\tilde{R}^{A}\left[\mathfrak{q}^{\mathrm{h}}\cdot\mathfrak{q}^{\mathrm{h}} - 1\right]\d s\\
\end{aligned}
\end{equation}
along with 
\begin{equation}\label{eq:disc:constraints}
\begin{aligned}
\vec{0} &= \int\limits_{\mathfrak{C}_0}
\tilde{N}^{A}\left[\vec{\varphi}^{\mathrm{h}} - \tilde{\vec{\varphi}}^{\mathrm{h}}\right]\,\CH{\abs{C^\mathrm{h}}}\,\d s,\\
\vec{0} &= \int\limits_{\mathfrak{C}_0}
\tilde{N}^{A}\,\ten{\mathcal{P}}\CH{\tp}^{,\mathrm{h}}_{\alpha}
%\left[
\UK{\vec{F}^{\mathrm{h}}\,\vec{D}^{\mathrm{h}}_{\alpha}}
%-\tilde{\vec{R}}^{\mathrm{h}}\vec{D}^{\mathrm{h}}_{\alpha}\right]
\,\BW{\abs{A^\mathrm{h}}}\,\d s,\\
\vec{0} &= \int\limits_{\mathfrak{C}_0}
\tilde{N}^{A}\,\ten{\mathcal{Q}}_{\alpha}^{\mathrm{T,h}}\,\tilde{\vec{R}}^\mathrm{h}\left[\vec{F}^{\mathrm{T,h}}\,\vec{F}^{\mathrm{h}}-\ten{I}\right]\,\vec{D}^\mathrm{h}_{\alpha}\,\BW{\abs{A^\mathrm{h}}}\,\d s,
\end{aligned}
\end{equation}
Note that the discrete values of $\mathfrak{P}^{\mathrm{h}}$, $\vec{P}^{\mathfrak{n},{\mathrm{h}}}$ and $\vec{P}^{\mathfrak{m},{\mathrm{h}}}$, are given as the discrete counterparts of \eqref{eq:thirdOrderDef}, \eqref{eq:Pn} and \eqref{eq:Pm}. Moreover, the unity constraint of the quaternions is evaluated in an integral sense,  i.e.,\ this is not necessarily fulfilled at each Gauss point where we use them to construct the rotation matrix.

\paragraph{Discrete reduction scheme.}
Within the Newton-Raphson iteration, we have to solve at every step

\begin{equation}\label{eq:newton}
\begin{bmatrix}
\vec{K}_{\vec{\varphi}\vec{\varphi}} & \vec{K}_{\vec{\varphi}\tilde{\vec{\varphi}}} & 
\vec{K}_{\vec{\varphi}\mathfrak{n}} & \vec{K}_{\vec{\varphi}\mathfrak{m}} & 
\vec{K}_{\vec{\varphi}\mathfrak{q}} & \vec{K}_{\vec{\varphi}\tilde{\vec{\mu}}_n} \\
\vec{K}_{\tilde{\vec{\varphi}}\vec{\varphi}} & \ten{M}_{\tilde{\vec{\varphi}}\tilde{\vec{\varphi}}} & \vec{0} & \vec{0} & \vec{0} & \vec{0} \\
\vec{0} & \vec{K}_{\mathfrak{n}\tilde{\vec{\varphi}}} &
\ten{M}_{\mathfrak{n}\mathfrak{n}} & \vec{0} & \vec{K}_{\mathfrak{n}\mathfrak{q}} & \vec{0} \\
\vec{0} & \vec{0} & \vec{0} & \ten{M}_{\mathfrak{m}\mathfrak{m}} & \vec{K}_{\mathfrak{m}\mathfrak{q}} & \vec{0} \\
\vec{K}_{\mathfrak{q}\vec{\varphi}} & \vec{0} & \vec{0} & \vec{0} & \vec{K}_{\mathfrak{q}\mathfrak{q}} & \vec{0}\\
\vec{K}_{\tilde{\vec{\mu}}_n\vec{\varphi}} & \vec{0} & \vec{0} & \vec{0} &\vec{K}_{\tilde{\vec{\mu}}_n\mathfrak{q}} & \vec{0}
\end{bmatrix}
\begin{bmatrix}
\Delta\vec{\varphi}\\
\Delta\tilde{\vec{\varphi}}\\
\Delta\mathfrak{n}\\
\Delta\mathfrak{m}\\
\Delta\mathfrak{q}\\
\Delta\tilde{\vec{\mu}}_n
\end{bmatrix} =
\begin{bmatrix}
\vec{R}_{\vec{\varphi}}\\
\vec{R}_{\tilde{\vec{\varphi}}}\\
\vec{R}_{\mathfrak{n}}\\
\vec{R}_{\mathfrak{m}}\\
\vec{R}_{\mathfrak{q}}\\
\vec{R}_{\tilde{\vec{\mu}}_n}
\end{bmatrix},
\end{equation}
where we have made use of \eqref{eq:simpleMat}. Here, $\vec{R}_{\vec{\varphi}}$ correlates to \eqref{eq:red2}, $\vec{R}_{\tilde{\vec{\varphi}}}$ to \eqref{eq:disc:constraints}$_1$, $\vec{R}_{\mathfrak{n}}$ to \eqref{eq:beamDis2}$_1$, $\vec{R}_{\mathfrak{m}}$ to \eqref{eq:beamDis2}$_2$, $\vec{R}_{\mathfrak{q}}$ and $\vec{R}_{\mathfrak{p}}$ to \eqref{eq:beamDis2}$_3$, \eqref{eq:disc:constraints}$_2$ and \eqref{eq:disc:constraints}$_3$, respectively. The matrices on the diagonal with components
\begin{equation}
[\ten{M}_{\tilde{\vec{\varphi}}\tilde{\vec{\varphi}}}]^{AB} = - \int\limits_{\mathfrak{C}_0}\tilde{N}^{A}\tilde{R}^{B}\d s,\quad\text{and}\quad
[\ten{M}_{\mathfrak{n}\mathfrak{n}}]^{AB} = [\ten{M}_{\mathfrak{m}\mathfrak{m}}]^{AB} = - \int\limits_{\mathfrak{C}_0}\tilde{M}^{A}\tilde{M}^{B}\d s
\end{equation}
are invertible, as long as $\ten{M}_{\tilde{\vec{\varphi}}\tilde{\vec{\varphi}}}$ is quadratic and invertible, depending on the order of $\delta\bar{\vec{\mu}}$.  This is obvious for $\tilde{N}^A = \tilde{R}^A$, for all other cases the invertibility has to be shown for the specific case.  Different choices are compared in the subsequently following numerical examples. With this, we can solve the second, third and fourth line in \eqref{eq:newton} via
\begin{equation}\label{eq:reduction1}
\begin{aligned}
\Delta\tilde{\vec{\varphi}} &= \ten{M}_{\tilde{\vec{\varphi}}\tilde{\vec{\varphi}}}^{-1}\left(\vec{R}_{\tilde{\vec{\varphi}}}-\vec{K}_{\tilde{\vec{\varphi}}\vec{\varphi}}\,\Delta\vec{\varphi}\right),\\
\Delta\mathfrak{n} &= \ten{M}_{\mathfrak{n}\mathfrak{n}}^{-1}\left(\vec{R}_{\mathfrak{n}}-\vec{K}_{\mathfrak{n}\tilde{\vec{\varphi}}}\,\Delta\tilde{\vec{\varphi}}-\vec{K}_{\mathfrak{n}\mathfrak{q}}\,\Delta\mathfrak{q}\right),\\
\Delta\mathfrak{m} &= \ten{M}_{\mathfrak{m}\mathfrak{m}}^{-1}\left(\vec{R}_{\mathfrak{m}}-\vec{K}_{\mathfrak{m}\mathfrak{q}}\,\Delta\mathfrak{q}\right),
\end{aligned}
\end{equation}
and obtain
\begin{equation}\label{eq:newton3}
\begin{bmatrix}
\bar{\vec{K}}_{\vec{\varphi}\vec{\varphi}} & 
\bar{\vec{K}}_{\vec{\varphi}\mathfrak{q}} & 
\vec{K}_{\vec{\varphi}\tilde{\vec{\mu}}_n} \\
\vec{K}_{\mathfrak{q}\vec{\varphi}} & \vec{K}_{\mathfrak{q}\mathfrak{q}} & \vec{0} \\
\vec{K}_{\tilde{\vec{\mu}}_n\vec{\varphi}} &\vec{K}_{\tilde{\vec{\mu}}_n\mathfrak{q}}  & \vec{0}
\end{bmatrix}
\begin{bmatrix}
\Delta\vec{\varphi}\\
\Delta\mathfrak{q}\\
\Delta\tilde{\vec{\mu}}_n
\end{bmatrix} =
\begin{bmatrix}
\bar{\vec{R}}_{\vec{\varphi}}\\
\vec{R}_{\mathfrak{q}}\\
\vec{R}_{\tilde{\vec{\mu}}_n}
\end{bmatrix},
\end{equation}
where 
\begin{equation}
\begin{aligned}
\bar{\vec{K}}_{\vec{\varphi}\vec{\varphi}} &= \vec{K}_{\vec{\varphi}\vec{\varphi}} - (\vec{K}_{\vec{\varphi}\tilde{\vec{\varphi}}} -
\vec{K}_{\vec{\varphi}\mathfrak{n}}\,\ten{M}_{\mathfrak{n}\mathfrak{n}}^{-1}\,\vec{K}_{\mathfrak{n}\tilde{\vec{\varphi}}})\,
\ten{M}_{\tilde{\vec{\varphi}}\tilde{\vec{\varphi}}}^{-1}\,\vec{K}_{\tilde{\vec{\varphi}}\vec{\varphi}},\\
\bar{\vec{K}}_{\vec{\varphi}\mathfrak{q}} &= \vec{K}_{\vec{\varphi}\mathfrak{q}} -
\vec{K}_{\vec{\varphi}\mathfrak{n}}\,\ten{M}_{\mathfrak{n}\mathfrak{n}}^{-1}\,\vec{K}_{\mathfrak{n}\mathfrak{q}} -
\vec{K}_{\vec{\varphi}\mathfrak{m}}\,\ten{M}_{\mathfrak{m}\mathfrak{m}}^{-1}\,\vec{K}_{\mathfrak{m}\mathfrak{q}},
\end{aligned}
\end{equation}
and the modified residual vector 
\begin{equation}
\bar{\vec{R}}_{\vec{\varphi}} = \vec{R}_{\vec{\varphi}} -
(\vec{K}_{\vec{\varphi}\tilde{\vec{\varphi}}}-\vec{K}_{\vec{\varphi}\mathfrak{n}}\,\ten{M}_{\mathfrak{n}\mathfrak{n}}^{-1}\,\vec{K}_{\mathfrak{n}\tilde{\vec{\varphi}}})\,
\ten{M}_{\tilde{\vec{\varphi}}\tilde{\vec{\varphi}}}^{-1}\,\vec{R}_{\tilde{\vec{\varphi}}} -
\vec{K}_{\vec{\varphi}\mathfrak{n}}\,\ten{M}_{\mathfrak{n}\mathfrak{n}}^{-1}\,\vec{R}_{\mathfrak{n}}  -
\vec{K}_{\vec{\varphi}\mathfrak{m}}\,\ten{M}_{\mathfrak{m}\mathfrak{m}}^{-1}\,\vec{R}_{\mathfrak{m}}.
\end{equation}

\begin{remark}
Without further proof we remark here, that the quaternions can be removed from the system using \eqref{eq:disc:constraints}$_2$ with regard to \eqref{eq:disc:constraints}$_3$. Therefore, we evaluate $\vec{d}_{\alpha} = \vec{F}\,\vec{D}_{\alpha}$ and calculate $\vec{d}_3 = \vec{d}_1\times\vec{d}_2$ at the respective Gauss point along the beam center line, such that the rotation tensor is given by $\tilde{\vec{R}} =\vec{d}_i\otimes\vec{D}_i$. Eventually, we need $\tilde{\vec{R}}^{\prime} = \vec{d}_i^{\prime}\otimes\vec{D}_i$ assuming again straight initial fibers, and we obtain immediately
\begin{equation}
\vec{d}_{\alpha}^{\prime} = \nabla \vec{F}:\vec{D}_{\alpha}\otimes\vec{D}_3,\quad\text{and}\quad
\vec{d}^{\prime}_3 = (\nabla \vec{F}:\vec{D}_{1}\otimes\vec{D}_3)\times \vec{d}_2 + \vec{d}_1\times(\nabla \vec{F}:\vec{D}_{2}\otimes\vec{D}_3).
\end{equation}
This leads finally to a system to be solved with respect to the matrix degrees of freedom $\vec{\varphi}$ and the Lagrange multipliers $\tilde{\vec{\mu}}_n$. If required, suitable methods like augmented Lagrange can be considered to solve solely with respect to the matrix unknowns. For the numerical examples in the next section, quaternions are used.
\end{remark}

\section{Numerical experiment}\label{sec:numerics}
In this section, we investigate the accuracy of the proposed formulation. In particular, we consider a benchmark test with results for a surface coupling between beam and matrix instead of a multidimensional coupling from~\cite{steinbrecher2019c}. Afterwards we investigate a torsional test, where it is obvious that pure position constraints are not suitable. Eventually, we demonstrate the applicability towards larger representative volume elements (RVE) for the analysis of multiple embedded fibers.

\subsection{Bending of a beam}\label{num:Bending}

\begin{figure}[t]
	\centering
	\begin{minipage}[top]{0.46\textwidth}
	\psfrag{x}{\small{$x$}}
	\psfrag{y}{\small{$y$}}
	\psfrag{z}{\small{$z$}}
	\psfrag{M}[c][r]{\small{$M$}}
		\includegraphics[width=\textwidth]{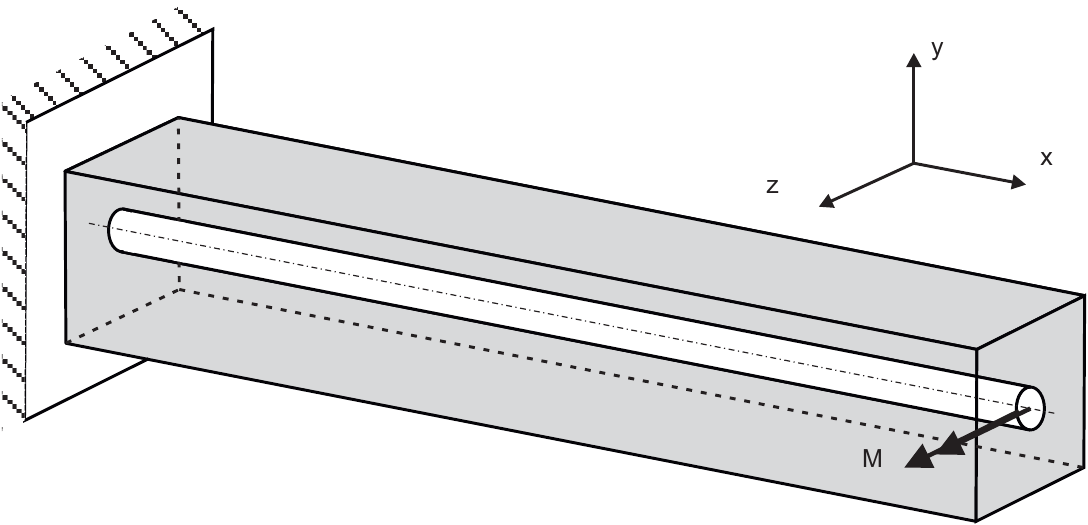}
	\end{minipage}
	\hfill
	\begin{minipage}[top]{0.46\textwidth}
	\psfrag{x}{\small{$x$}}
	\psfrag{y}{\small{$y$}}
	\psfrag{z}{\small{$z$}}
	\psfrag{M}{\small{$M$}}
	 \includegraphics[width=\textwidth]{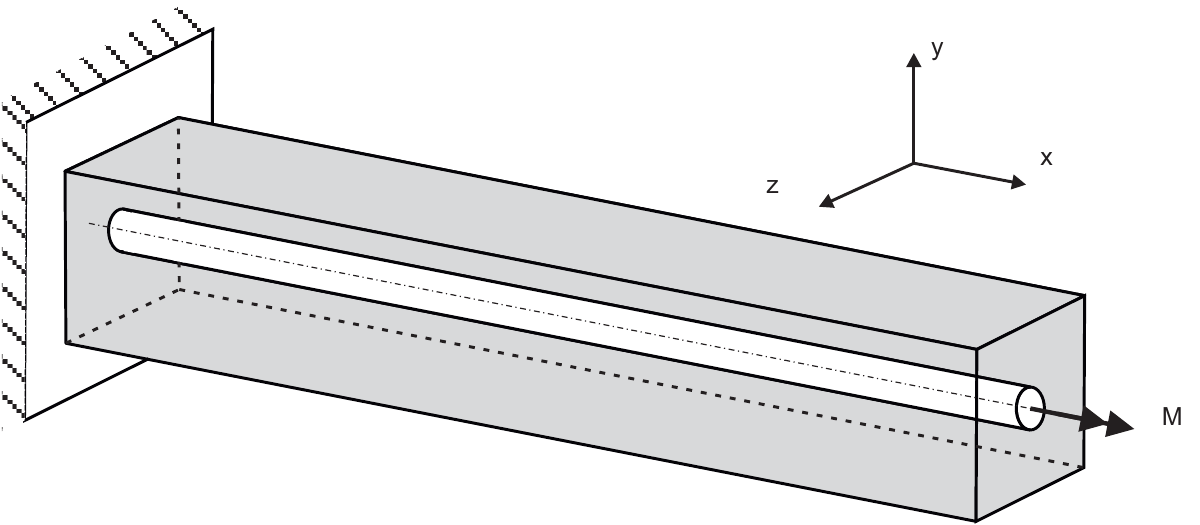}
	\end{minipage}
	\hfill
	\caption{Reference configuration of the bending test (left) and the torsion test (right).}
	\label{fig:RefBend}
\end{figure}

In this first numerical example, we consider a model problem from~\cite{steinbrecher2019c}, Section 4.2: A beam of length~$5\,\mathrm{m}$ and radius~$r=0.125\,\mathrm{m}$ with Young modulus~$4346\,\mathrm{N/m^2}$ is embedded into the $1\,\mathrm{m}\times1\,\mathrm{m}\times5\,\mathrm{m}$ matrix block of Saint-Venant-Kirchhoff material with Young modulus~$10\,\mathrm{N/m^2}$. Poisson ratio is zero for both materials. The geometry and material parameters for both, the matrix and the beam are chosen such that both systems have similar properties, i.e.,\ if the matrix material would be considered as beam, we would obtain the same moment of inertia.

The matrix and the beam are both fixed at $x=0\,\mathrm{m}$, and we apply a moment~$\vec{m}^{L}_{ext}=[0,\,0,\,0.025]\,\mathrm{Nm}$ to the beam tip~$x=5\,\mathrm{m}$ as a dead load, see Figure \ref{fig:RefBend}, left, for details. This benchmark test allows us to illustrate the stability of the proposed approach, its convergence and the model error, since the solution for a surface coupling is provided in \cite{steinbrecher2019c}. The matrix is discretized with B-Splines of order~$p=[p_x,p_y,p_z]$. We consider the same order in $y$- and $z$-direction, $p_y=p_z$, and the order in $x$-direction is the same as for the beam. The number~$n$ of elements in $y$- and $z$-direction is the same and by a factor of five smaller than in $x$-direction. For the beam, the same number of elements in $x$-direction is used.

\begin{figure}[t]
	\centering
	\includegraphics[width=\textwidth]{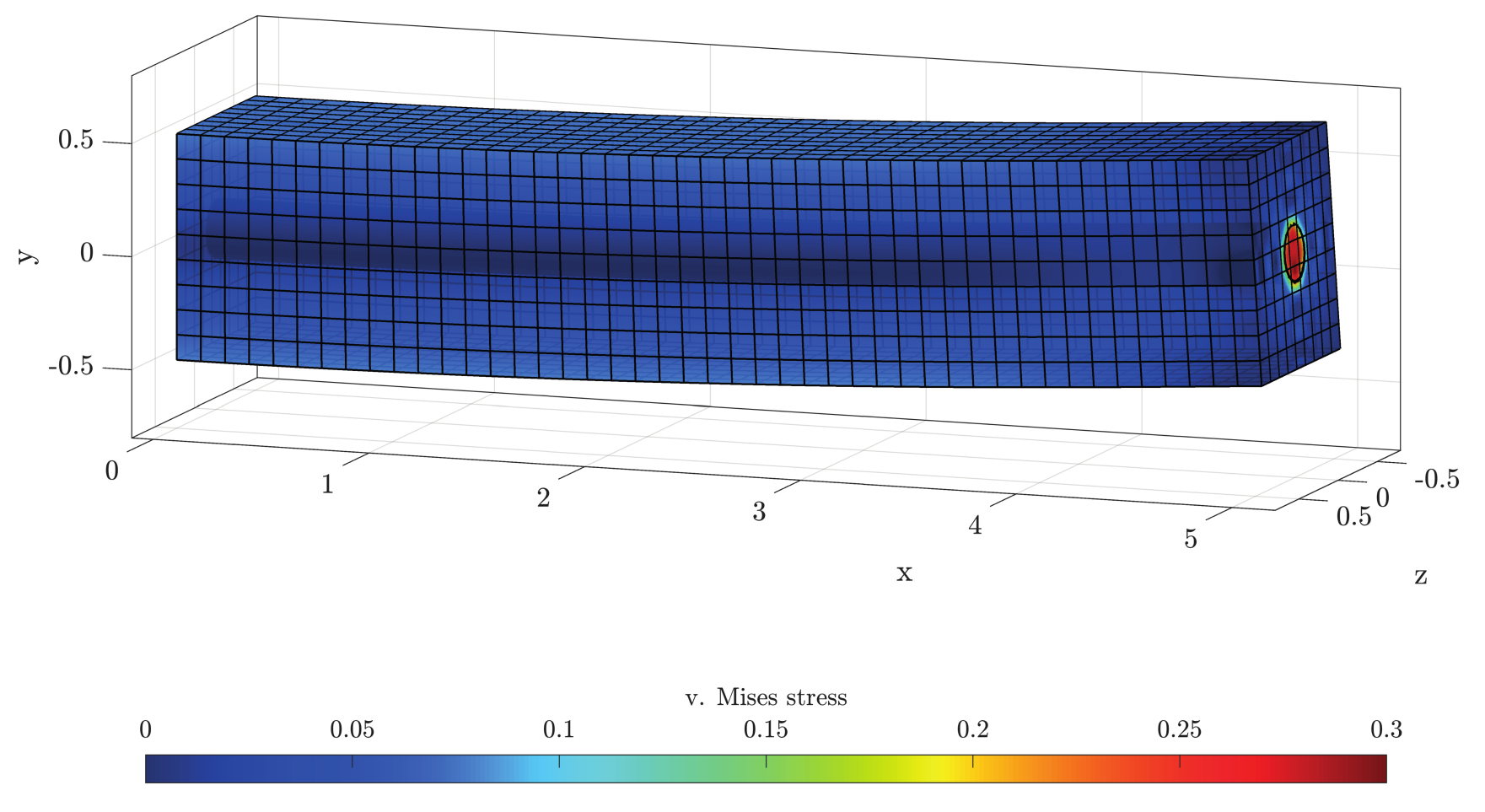}
	\caption{\textbf{Bending of a beam:} Von Mises stress distribution.}
	\label{fig:BeamMatrixBending}
\end{figure}

\begin{figure}[htp]
\centering
\includegraphics[width=0.6\textwidth]{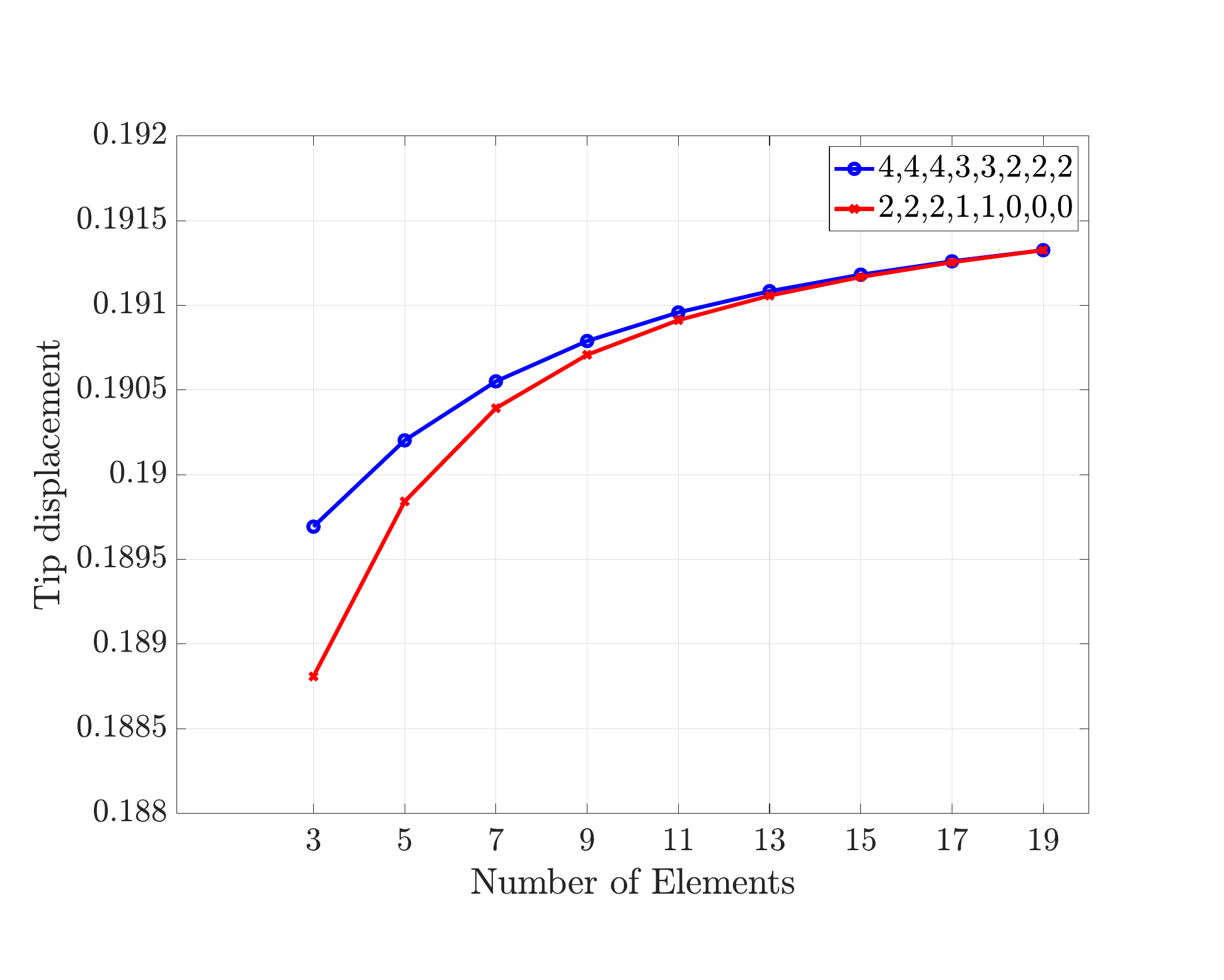}
\caption{\textbf{Bending of a beam:} Convergence of the tip displacement for different combinations of shape functions.}\label{fig:TipDispAbsErr}
\end{figure}

The deformation and the von Mises stress distribution for the bending problem is presented in Figure \ref{fig:BeamMatrixBending}. As expected, the stresses concentrate at the right hand side of the beam where the external moment is applied to. In Figure \ref{fig:TipDispAbsErr}, the displacement of the tip of the beam is presented for a second and fourth order example. The numbers in the legend denote the order of the B-splines for [$\vec{\varphi},\,\tilde{\vec{\varphi}},\,\mathfrak{q},\,\mathfrak{n},\,\mathfrak{m},\,\bar{\vec{\mu}},\,\tilde{\vec{\mu}}_{\tau},\,\tilde{\vec{\mu}}_{n}$]. As can be seen, both combinations of different orders converge to the same result.

\begin{figure}[t]
	\centering
	\begin{minipage}[top]{0.325\textwidth}
		\includegraphics[width=\textwidth]{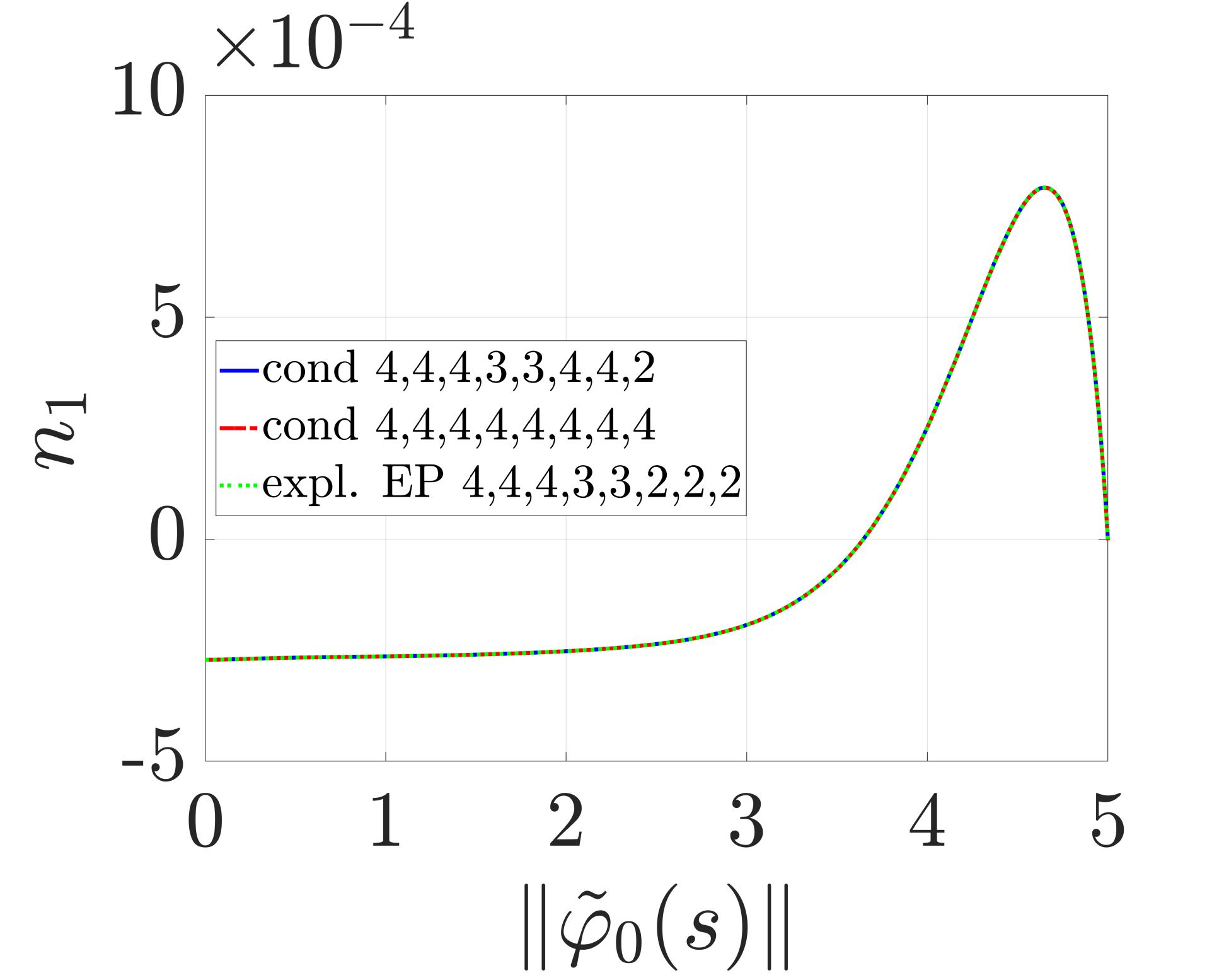}
	\end{minipage}
	\hfill
	\begin{minipage}[top]{0.325\textwidth}
		\includegraphics[width=\textwidth]{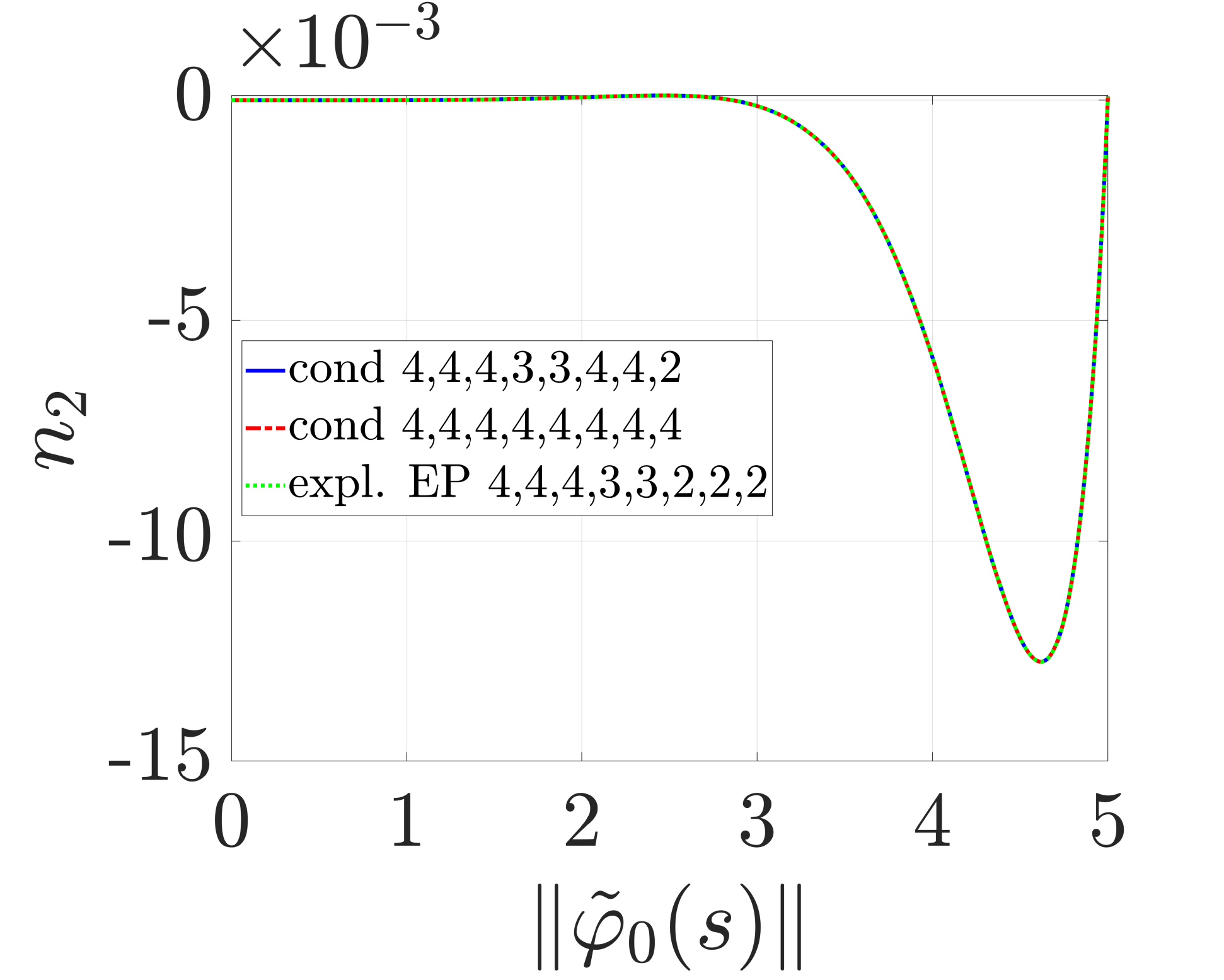}
	\end{minipage}
	\hfill
	\begin{minipage}[top]{0.325\textwidth}
		\includegraphics[width=\textwidth]{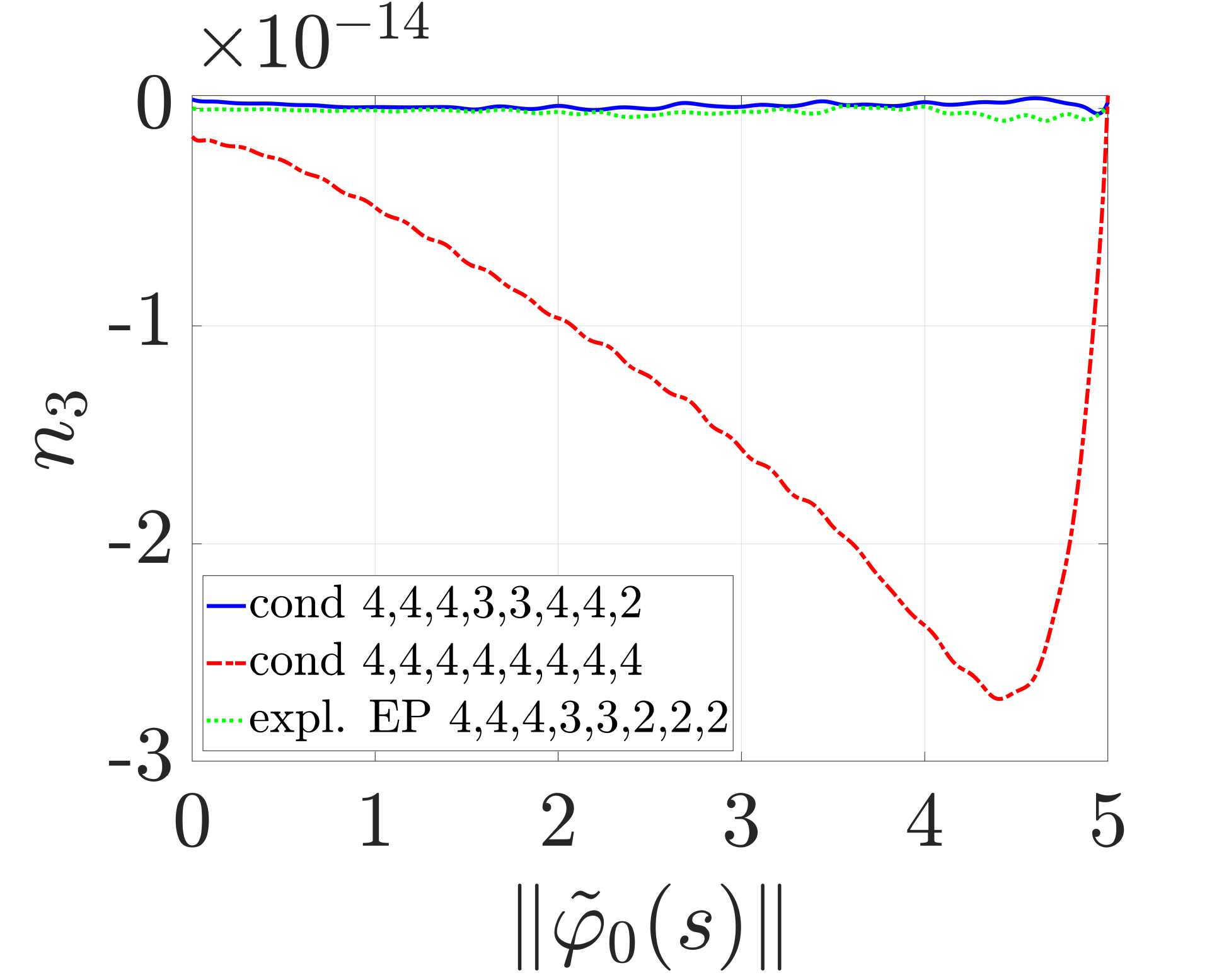}
	\end{minipage}
	\hfill\\
	\begin{minipage}[top]{0.325\textwidth}
		\includegraphics[width=\textwidth]{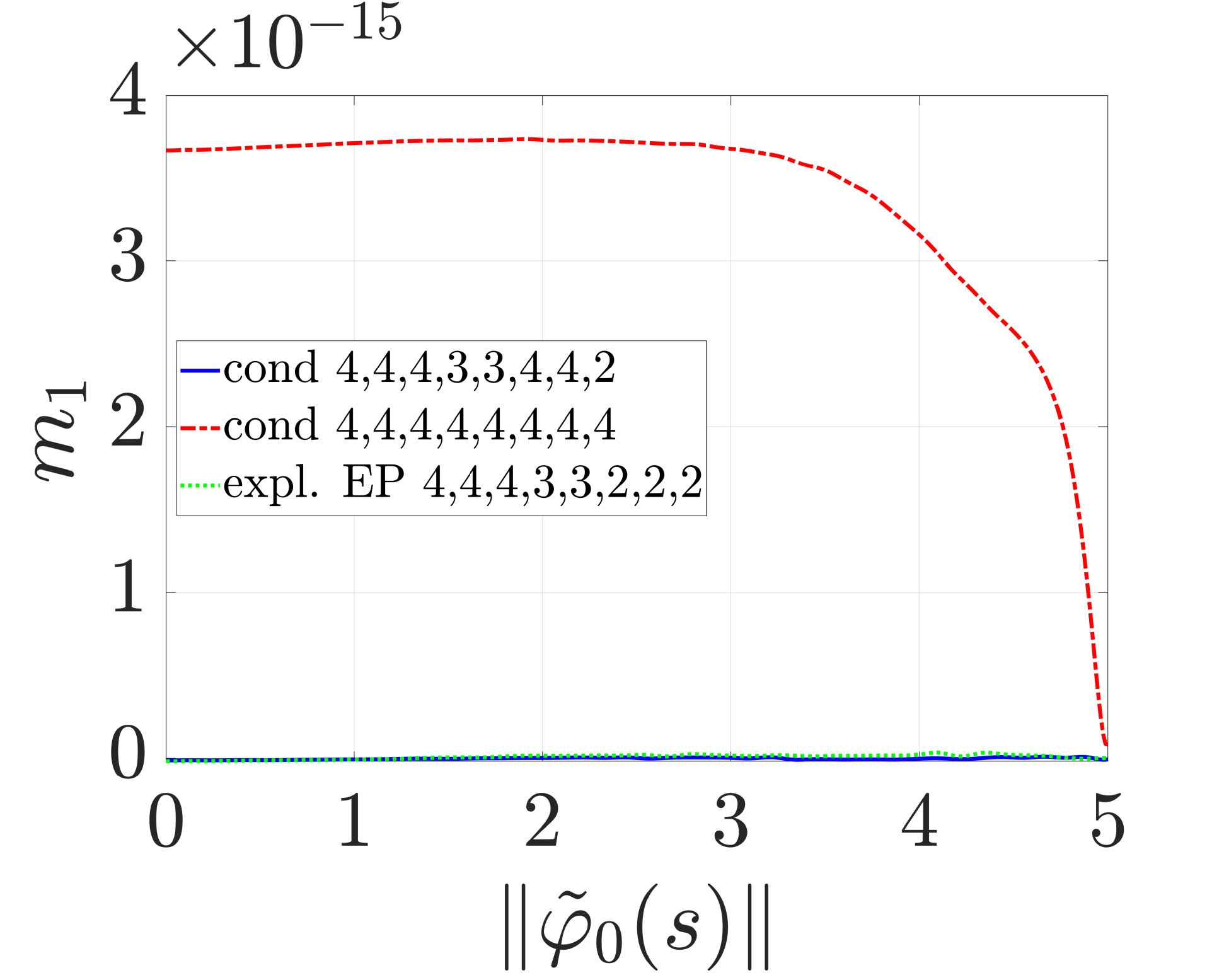}
	\end{minipage}
	\hfill
	\begin{minipage}[top]{0.325\textwidth}
		\includegraphics[width=\textwidth]{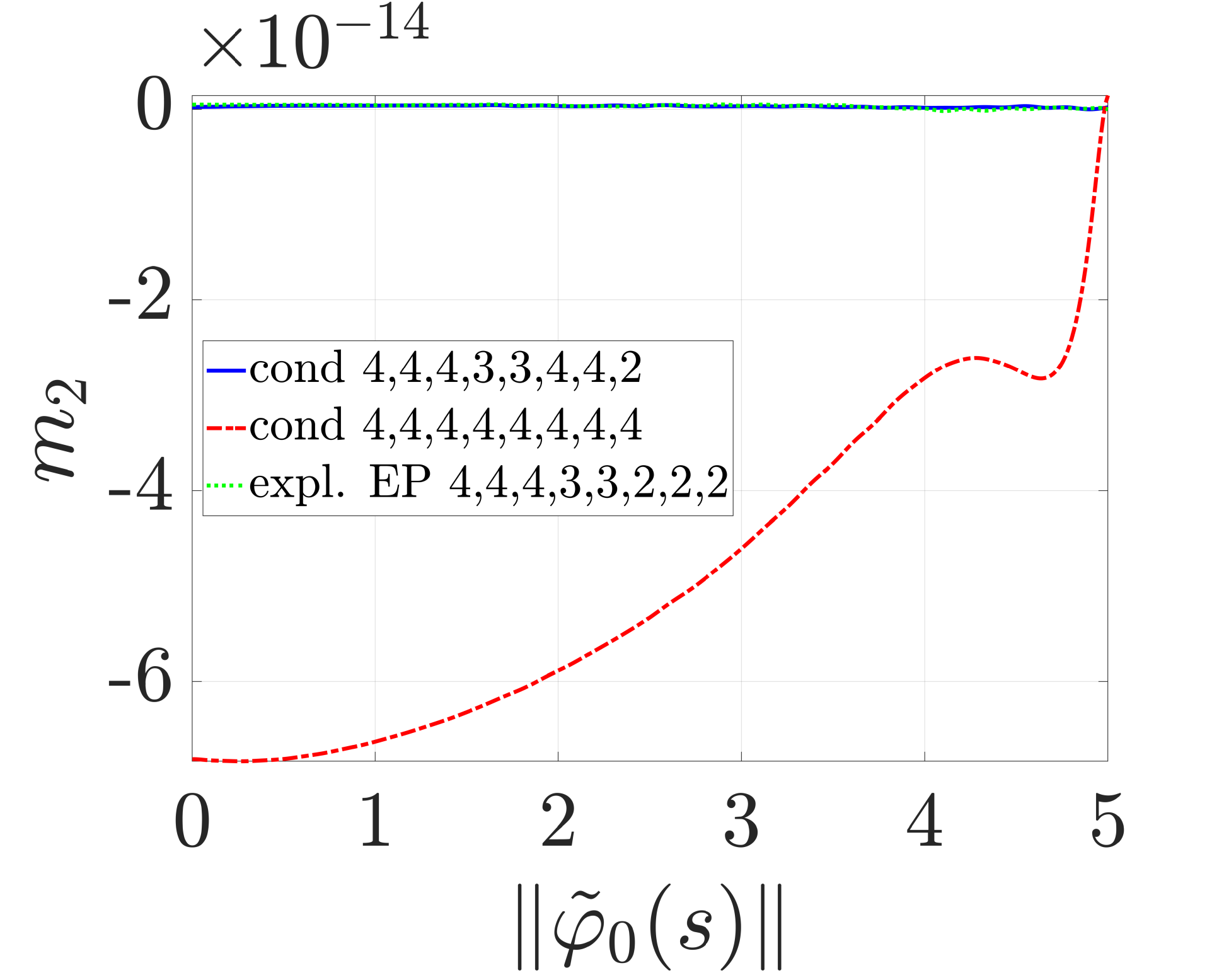}
	\end{minipage}
	\hfill
	\begin{minipage}[top]{0.325\textwidth}
		\includegraphics[width=\textwidth]{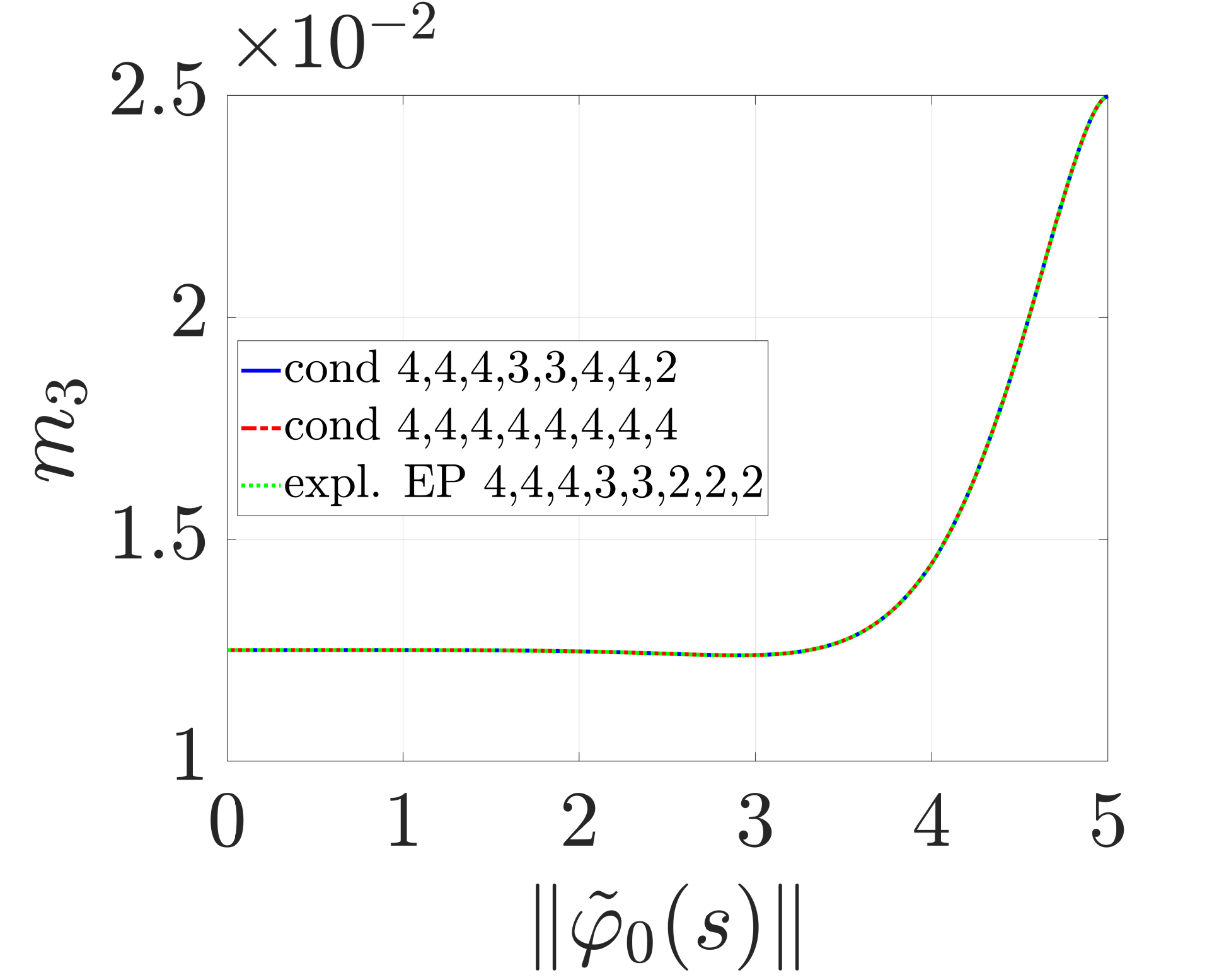}
	\end{minipage}
	\hfill\\
	\begin{minipage}[top]{0.325\textwidth}
		\includegraphics[width=\textwidth]{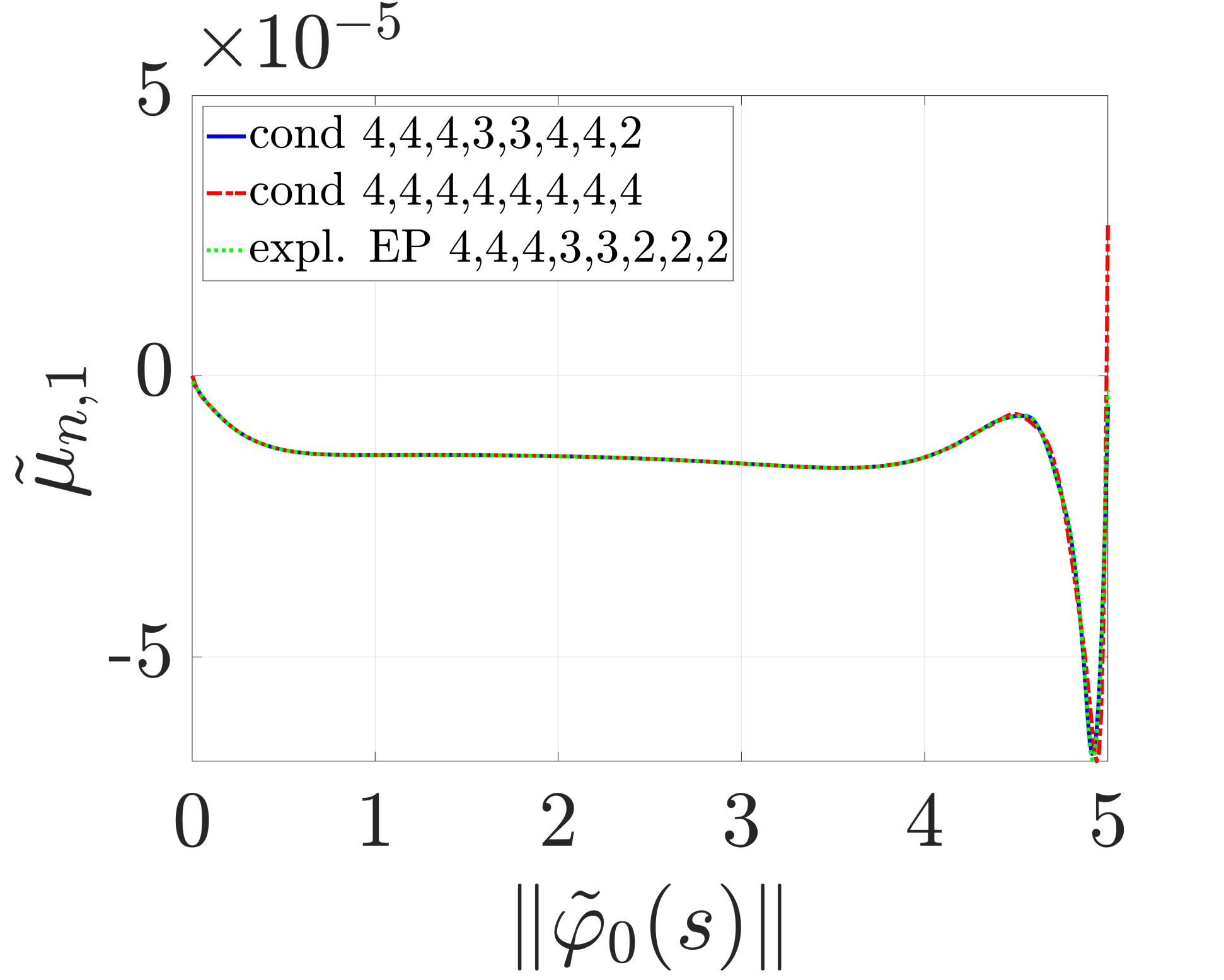}
	\end{minipage}
	\hfill
	\begin{minipage}[top]{0.325\textwidth}
		\includegraphics[width=\textwidth]{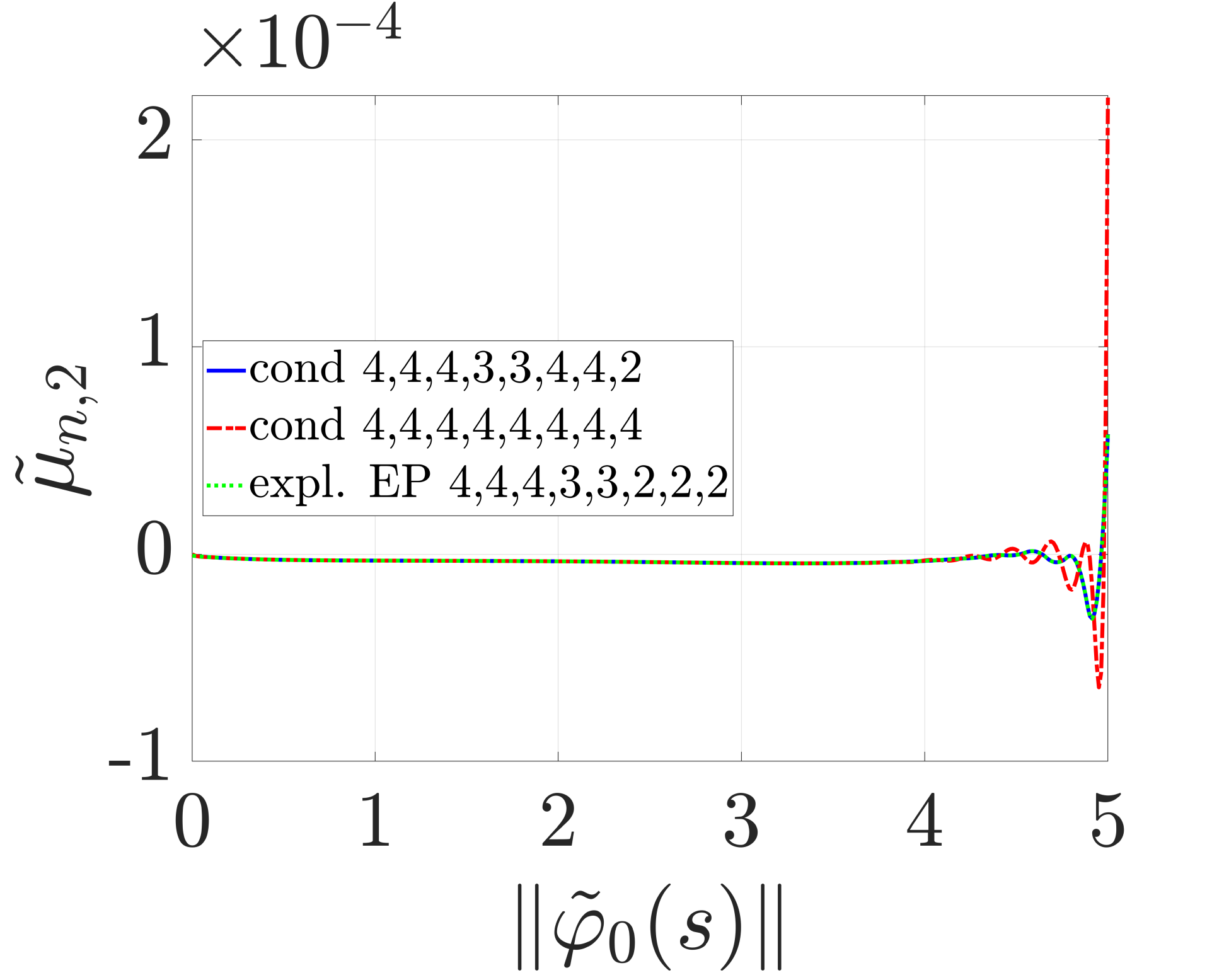}
	\end{minipage}
	\hfill
	\begin{minipage}[top]{0.325\textwidth}
		\includegraphics[width=\textwidth]{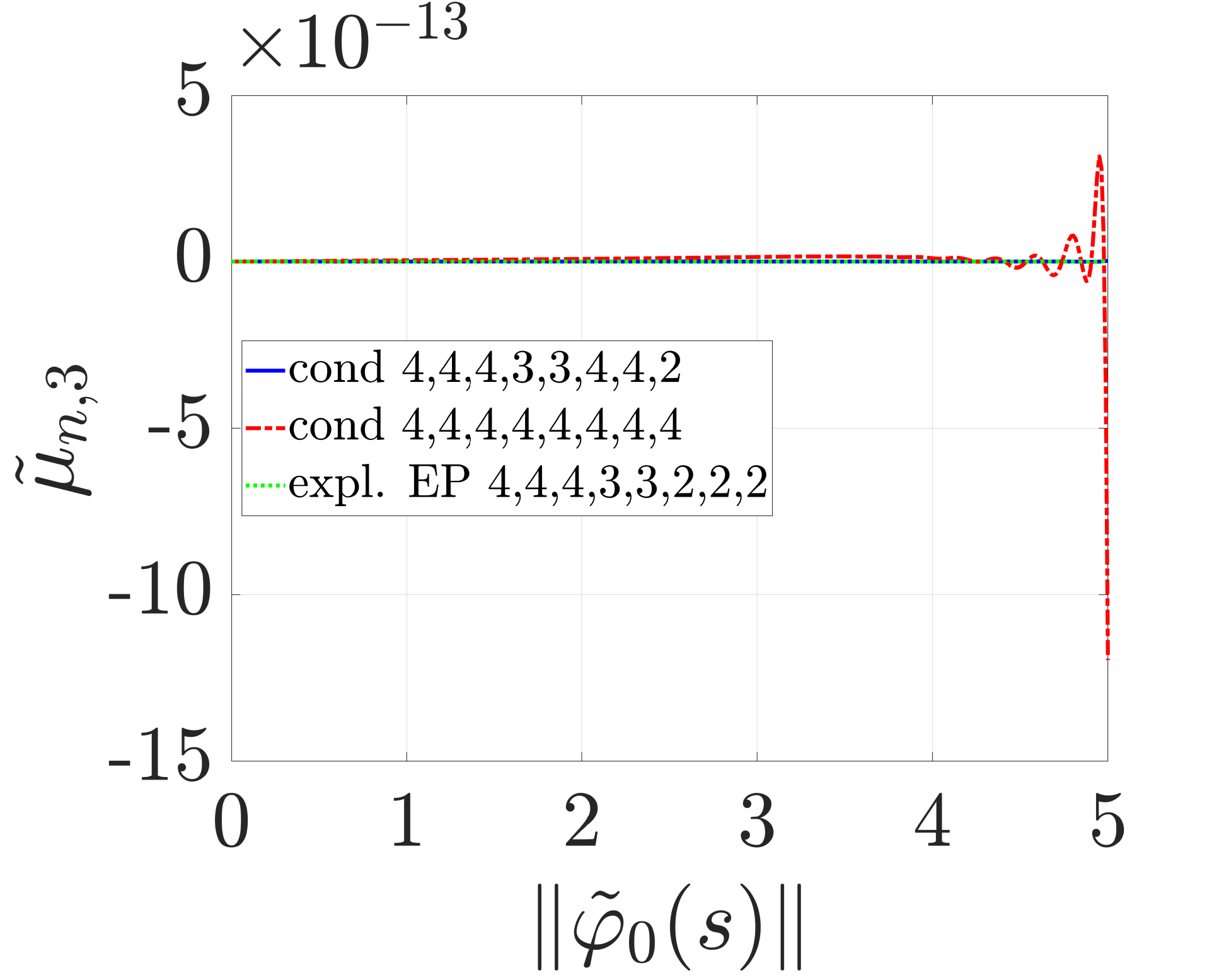}
	\end{minipage}
	\hfill
	\caption{\textbf{Bending of a beam:} From top to bottom, $\mathfrak{n},\,\mathfrak{m}$ and $\tilde{\vec{\mu}}_n$, from left to right $1$-,\,$2$- and $3$-components are displayed.}
	\label{fig:BeamMatrixBendingResults}
\end{figure}

In Figure \ref{fig:BeamMatrixBendingResults}, the results for $\mathfrak{n},\,\mathfrak{m}$ and $\tilde{\vec{\mu}}_n$ along the beam center line are displayed. The different order chosen for [$\vec{\varphi},\,\tilde{\vec{\varphi}},\,\mathfrak{q},\,\mathfrak{n},\,\mathfrak{m},\,\bar{\vec{\mu}},\,\tilde{\vec{\mu}}_{\tau},\,\tilde{\vec{\mu}}_{n}$] are displayed as well. Note that we fixed $\tilde{\vec{\mu}}_{n}$ at the first node equal to zero, to avoid spurious oscillations leading to a divergence within the Newton iteration. Moreover, the components $[\mathfrak{n}]_3$, $[\mathfrak{m}]_1$, $[\mathfrak{m}]_2$ and $[\tilde{\vec{\mu}}_n]_3$ are close to zero, as expected for this example. Finally, we have plotted the results for an explicit enforcement of \eqref{eq:sys:pi2} at both endpoints of the beam additionally to the integral enforcement of the constraints along $\mathcal{C}_0$. Note that no differences can be observed. 

\begin{center}
\begin{table}
\scriptsize
\centering\begin{tabular}{|c|c|c|c|c|c|c|c|c|}
\hline
Routine & & & expl. EP & expl. EP & expl. EP & expl. EP & cond & cond \\\hline
&q-unity && weak&  direct&  direct&  direct&  direct&  direct\\\hline
&&orders & 44433222 & 44433222 &44444333	&44444222	&44433442	&44433443	\\\hline
expl. EP&	weak	&44433222	&&	1.13E-4 \%	&1.32E-4 \%	& 1.13E-4\%	& 1.17E-4 \%	& 1.15E-4 \%\\\hline
expl. EP&	direct	&44433222	& 1.13E-4 \%	&&	1.90E-5 \%&	1.00E-7 \%	& 4.00E-6\%	& 2.00E-6 \%\\\hline
expl. EP&	direct	&44444333	& 1.32E-4 \%	&1.90E-5 \%	&& 1.90E-5 \%	& 1.50E-5 \%	& 1.70E-5 \%\\\hline
expl. EP&	direct	&44444222	& 1.13E-4 \%	&1.00E-7 \%& 1.90E-5 \%	&&	4.00E-6\% & 2.00E-6 \%\\\hline
cond	&direct	&44433442	& 1.17E-4 \%&	4.00E-6 \%	& 1.50E-5 \%	& 4.00E-6 \%	&& 2.00E-6 \%\\\hline
cond	&direct	&44433443	& 1.15E-4 \%	&2.00E-6 \%	& 1.70E-5 \%	& 2.00E-6 \%	& 2.00E-6 \%	&\\\hline
\end{tabular}
\caption{\textbf{Bending of a beam:} Deviation of the beam tip displacement in percentage for different orders and approaches for the quaternion unity constraint.}\label{tab:bending}
\end{table}
\end{center}

Eventually, we compare the deviation of the beam tip displacement in percentage for different approaches. To be more specific, we evaluate the absolute value of 
$|\frac{\|(\tilde{\vec{\varphi}}(L)-\tilde{\vec{\varphi}}_0(L)\|_a}{\|(\tilde{\vec{\varphi}}(L)-\tilde{\vec{\varphi}}_0(L)\|_b} -1|$ in percentage for approach $a$ compared to $b$, where the absolute tip displacement of the first approach displaced in Table \ref{tab:bending} is $0.19078898128\,\mathrm{m}$. First, we compare the additional, explicit enforcement of the endpoint constraints as discussed above. Moreover, the unity constraints of the quaternions is either enforced in an integral sense or directly, i.e.,\ pointwise at the nodes. This is evaluated for different combinations shape functions of different order for [$\vec{\varphi},\,\tilde{\vec{\varphi}},\,\mathfrak{q},\,\mathfrak{n},\,\mathfrak{m},\,\bar{\vec{\mu}},\,\tilde{\vec{\mu}}_{\tau},\,\tilde{\vec{\mu}}_{n}$], see Table \ref{tab:bending}. We remark here, that the relative deviation between $[4,\,4,\,4,\,3,\,3,\,2,\,2,\,2]$ compared to $[4,\,4,\,4,\,4,\,4,\,3,\,3,\,3]$ is in the order of $10^{-5}$ \%, i.e.\ negligible and already includes the better approximation of the higher order shape functions. However, the implementation simplifies for $[4,\,4,\,4,\,4,\,4,\,3,\,3,\,3]$ significant, as only two different order of shape functions have to be evaluated at each Gauss point. 

\subsection{Torsion test}\label{num:Torsion}

Next, we use the same setup as in Section \ref{num:Bending} using a Mooney-Rivlin material model 
\begin{equation}
\Psi(J, I_1,I_2) = c\,(J-1)^2 - d\, \op{ln}(J) + c_1\,(I_1-3) + c_2\,(I_2-3),
\end{equation}
as we expect to obtain large deformations. Here, $J = \op{det}(\vec{F})$, $I_1 = \op{tr}(\vec{F}\tp\,\vec{F}) = \vec{F}:\vec{F}$ and \review{$I_2 = \op{tr}(\op{cof}(\vec{F}\tp\,\vec{F})) = \vec{H}:\vec{H}$}. Moreover, $c = 2/3\,(c_1 + c_2)$, $d = 2\,(c_1 + 2\,c_2)$, $c_1 = 2$ and $c_2 = 1$. We apply an external moment $\vec{m}_{ext}^L = [0.9,\,0,\, 0]$ Nm, such that we obtain a pure torsion (see Figure \ref{fig:RefBend}, right). Note that using only position constraints is not possible for this example, as the immersed beam can rotate without influencing the matrix material. In Figure \ref{fig:BeamMatrixTorsion}, the von Mises stress distribution is shown. For visualization, we have plotted the (virtual) surface of the beam inside the matrix material. As can be observed, the matrix rotates within the virtual area of the beam geometry due to the incorporation of geometrical data in terms of the second moment of inertia within the beam material model in \eqref{eq:simpleMat}. In fact, the gradient around the 1-dimensional beam is restricted; a full reconstruction of the beam geometry would require all higher-order terms.

\begin{figure}[t]
	\centering
	\begin{minipage}[top]{0.9\textwidth}
	\vspace{5ex}
		\includegraphics[width=\textwidth]{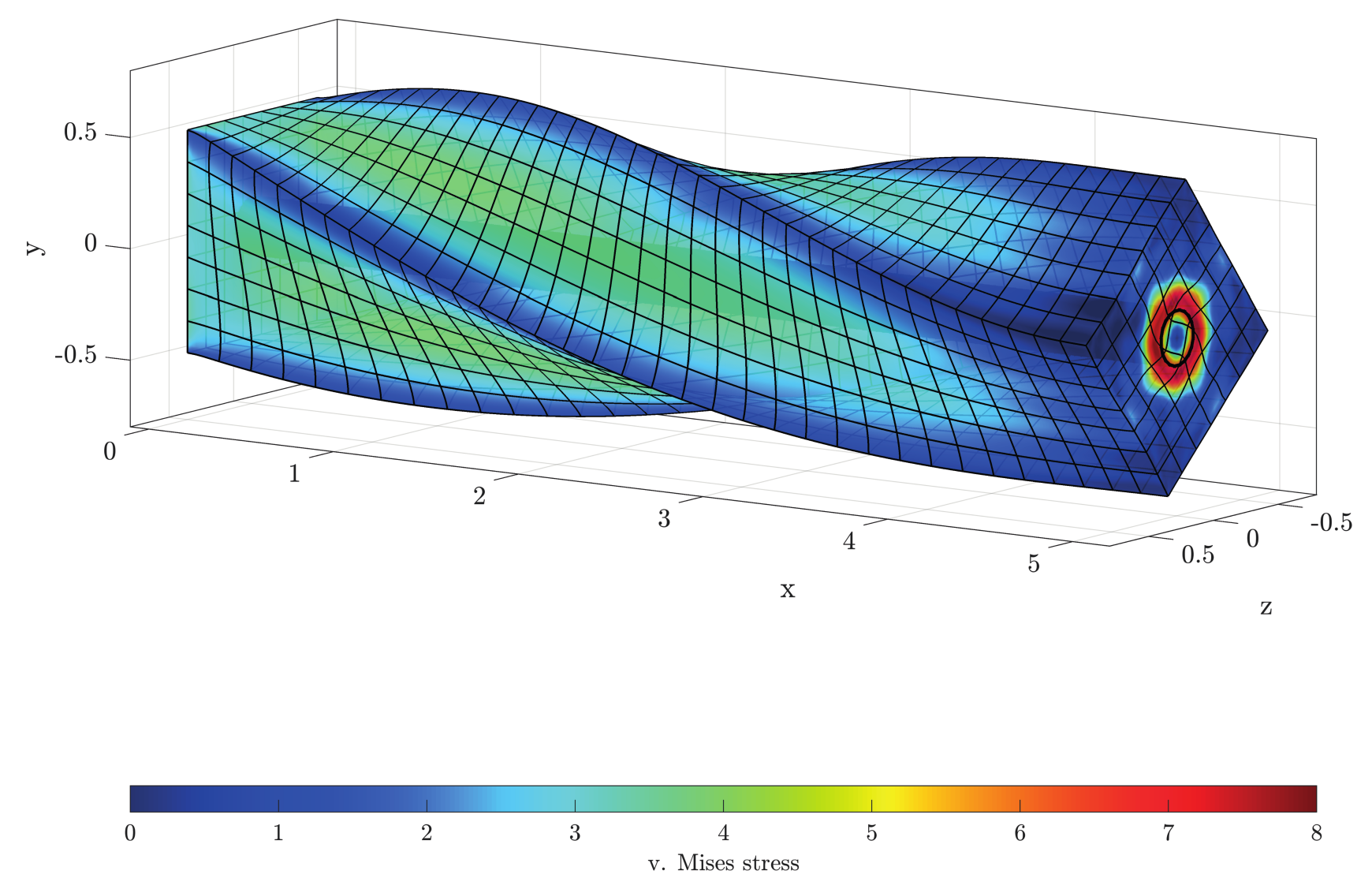}
	\end{minipage}
	\hfill
	\caption{\textbf{Torsion Test:} Von Mises stress distribution.}
	\label{fig:BeamMatrixTorsion}
\end{figure}

In Figure \ref{fig:BeamMatrixTorsionResults}, the results for $\mathfrak{n},\,\mathfrak{m}$ and $\tilde{\vec{\mu}}_n$ along the beam center line are displayed. Note that the $\mathfrak{n}_2$, $\mathfrak{n}_3$, $\mathfrak{m}_2$ and $\mathfrak{m}_3$ are again near the numerical limit, which corresponds to the physics of this example. The chosen order for [$\vec{\varphi},\,\tilde{\vec{\varphi}},\,\mathfrak{q},\,\mathfrak{n},\,\mathfrak{m},\,\bar{\vec{\mu}},\,\tilde{\vec{\mu}}_{\tau},\,\tilde{\vec{\mu}}_{n}$] is $[4,\,4,\,4,\,3,\,3,\,2,\,2,\,2]$. As the Mooney-Rivlin material model does not provide a Poisson ration equal zero, we obtain a slight elongation in $x-$ direction, resulting in the displayed forces obtained in $\mathfrak{n}_1(L)$. A comparative study using the Saint-Venant Kirchhoff model of the sub-section before results in an exact zero elongation and zero forces in $\mathfrak{n}_1(L)$. As the results are not physically meaningful due to the obtained large deformation, only the results for the Mooney-Rivlin material are plotted.

\begin{figure}[t]
	\centering
	\begin{minipage}[top]{0.325\textwidth}
		\includegraphics[width=\textwidth]{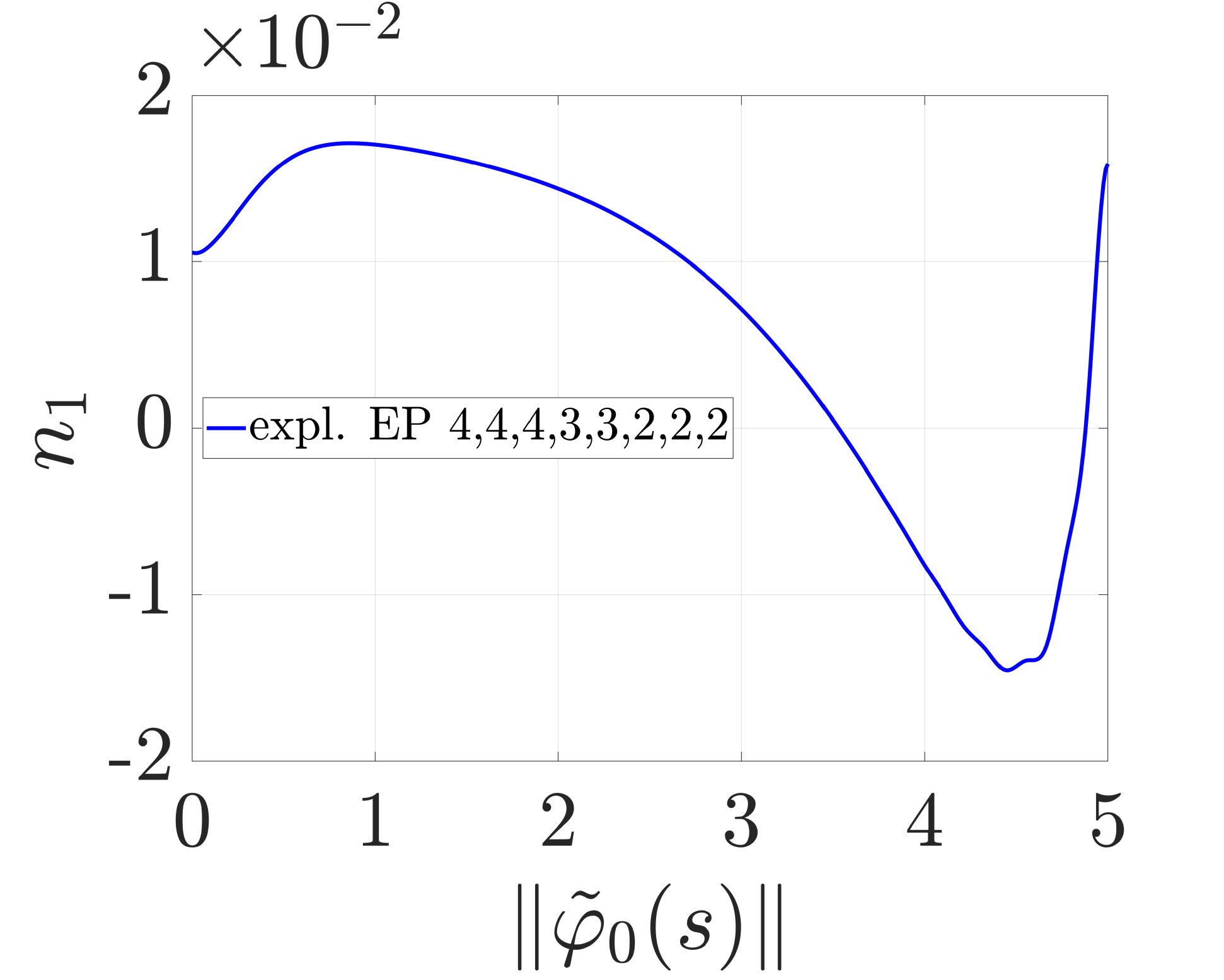}
	\end{minipage}
	\hfill
	\begin{minipage}[top]{0.325\textwidth}
		\includegraphics[width=\textwidth]{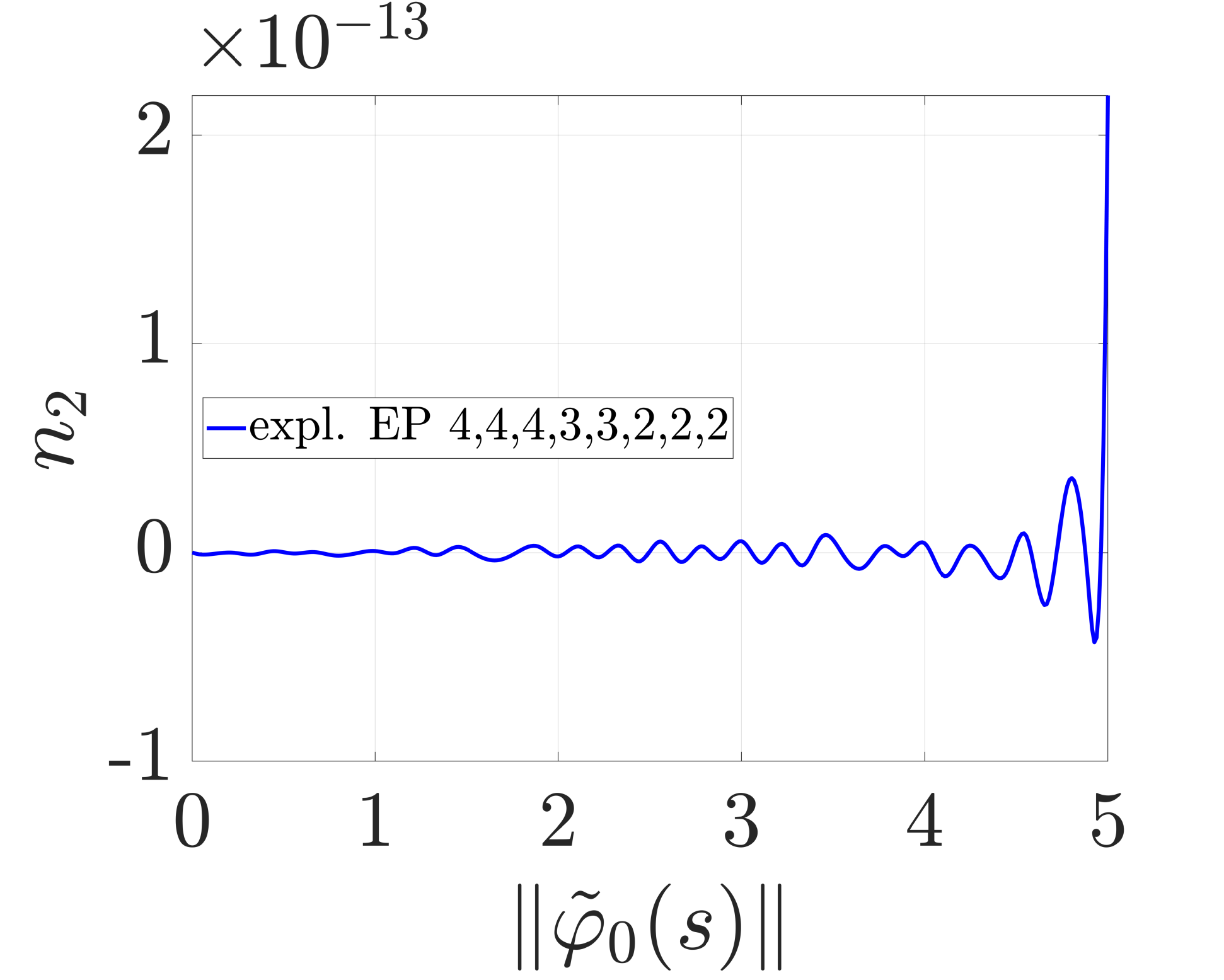}
	\end{minipage}
	\hfill
	\begin{minipage}[top]{0.325\textwidth}
		\includegraphics[width=\textwidth]{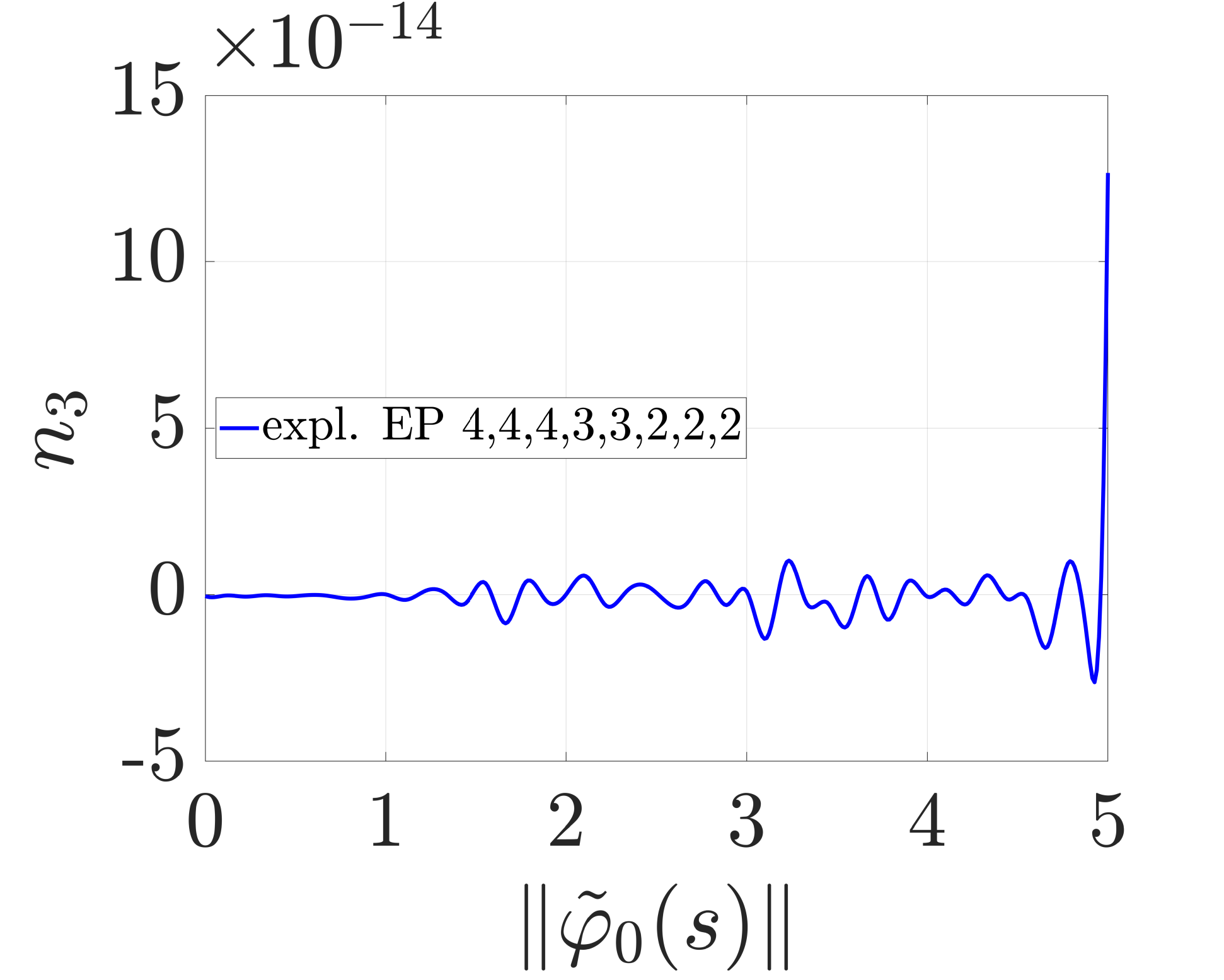}
	\end{minipage}
	\hfill\\
	\begin{minipage}[top]{0.325\textwidth}
		\includegraphics[width=\textwidth]{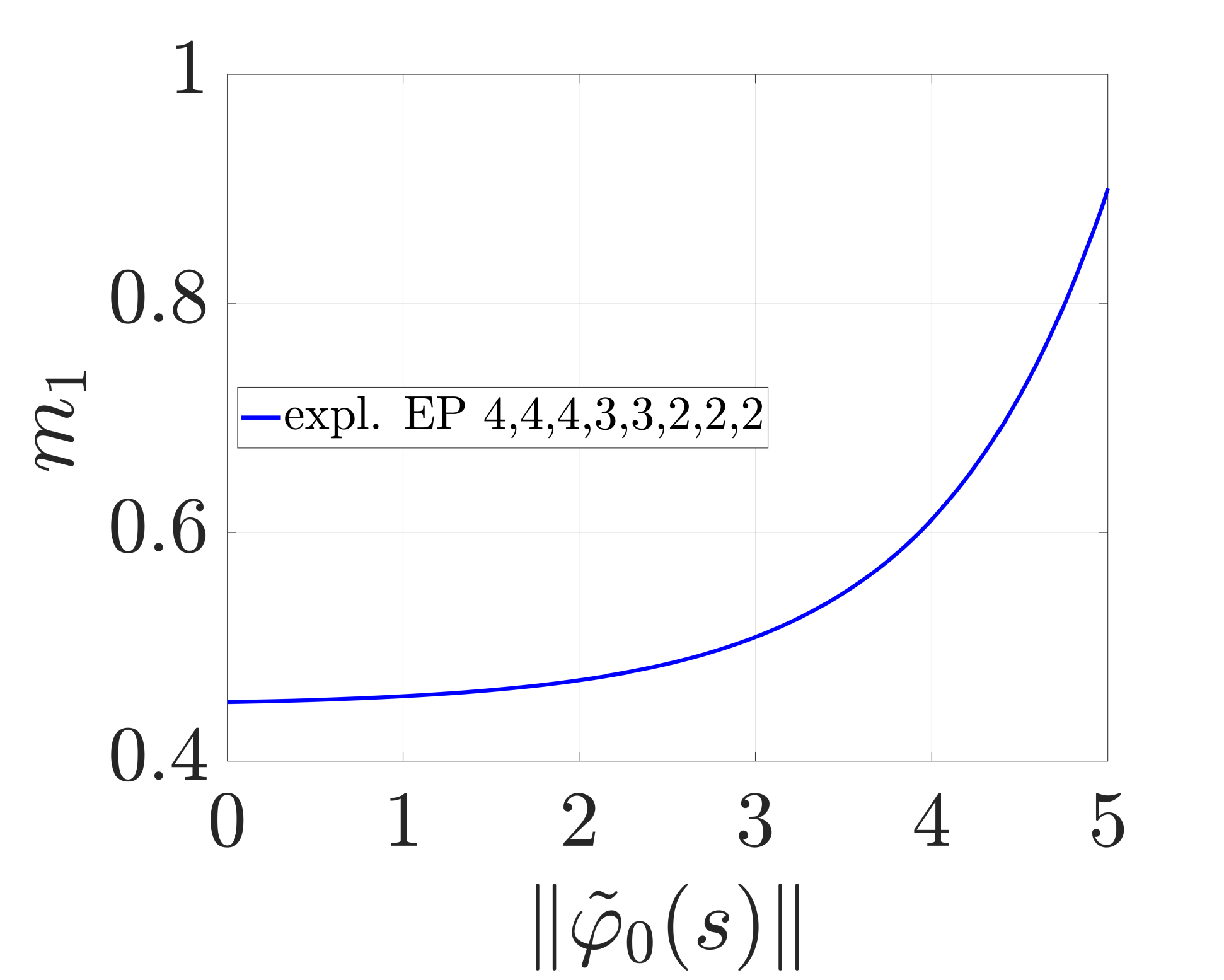}
	\end{minipage}
	\hfill
	\begin{minipage}[top]{0.325\textwidth}
		\includegraphics[width=\textwidth]{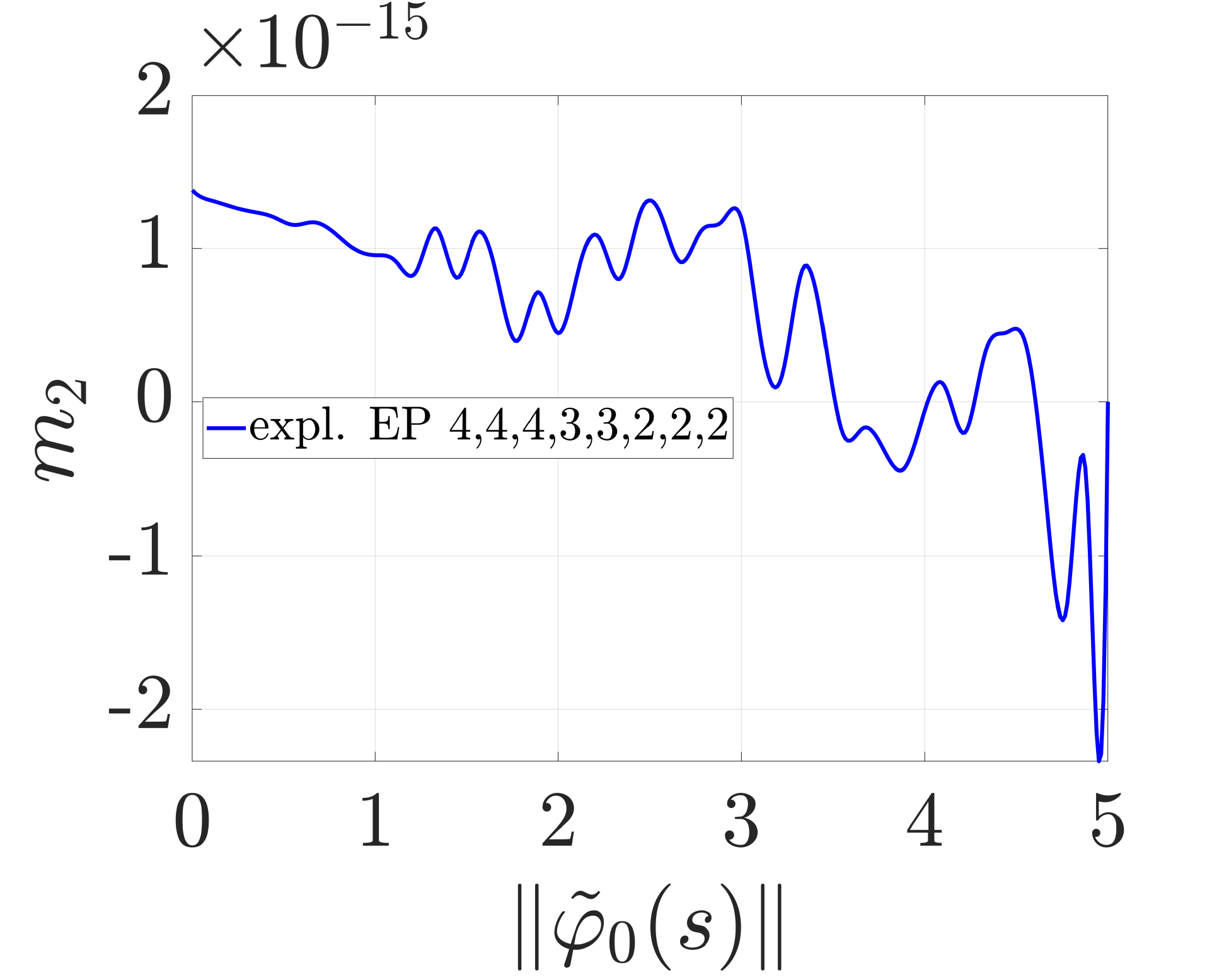}
	\end{minipage}
	\hfill
	\begin{minipage}[top]{0.325\textwidth}
		\includegraphics[width=\textwidth]{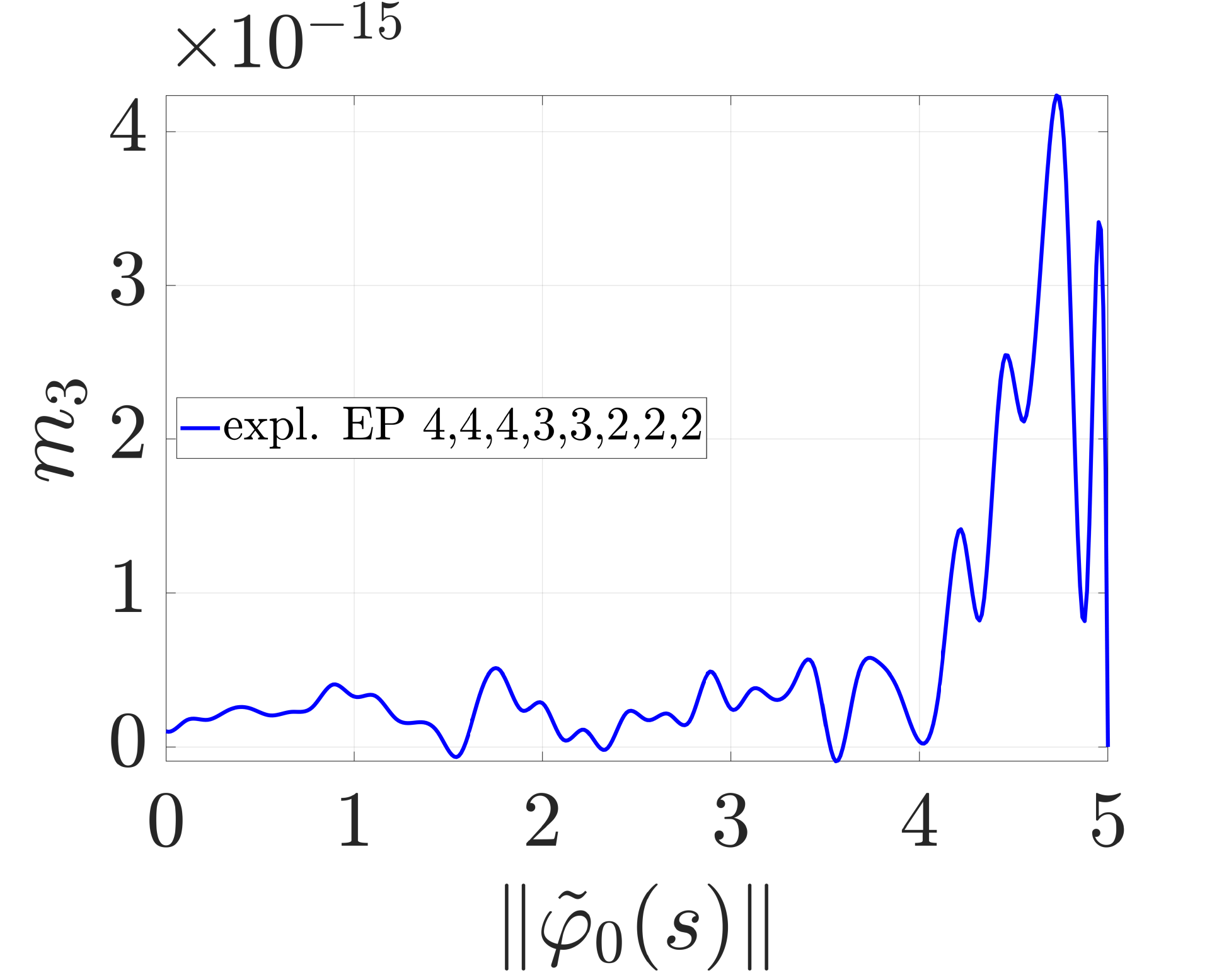}
	\end{minipage}
	\hfill\\
	\begin{minipage}[top]{0.325\textwidth}
		\includegraphics[width=\textwidth]{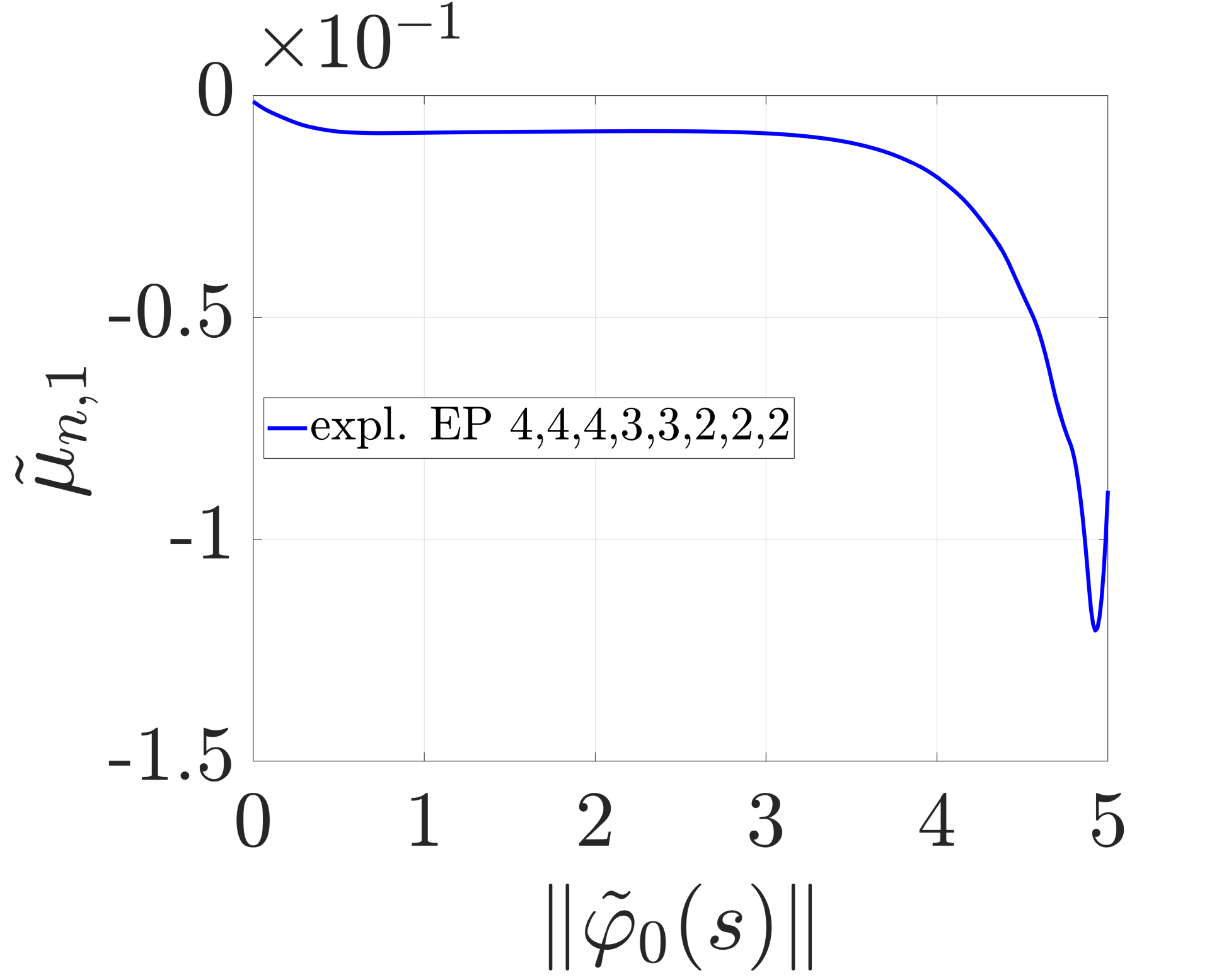}
	\end{minipage}
	\hfill
	\begin{minipage}[top]{0.325\textwidth}
		\includegraphics[width=\textwidth]{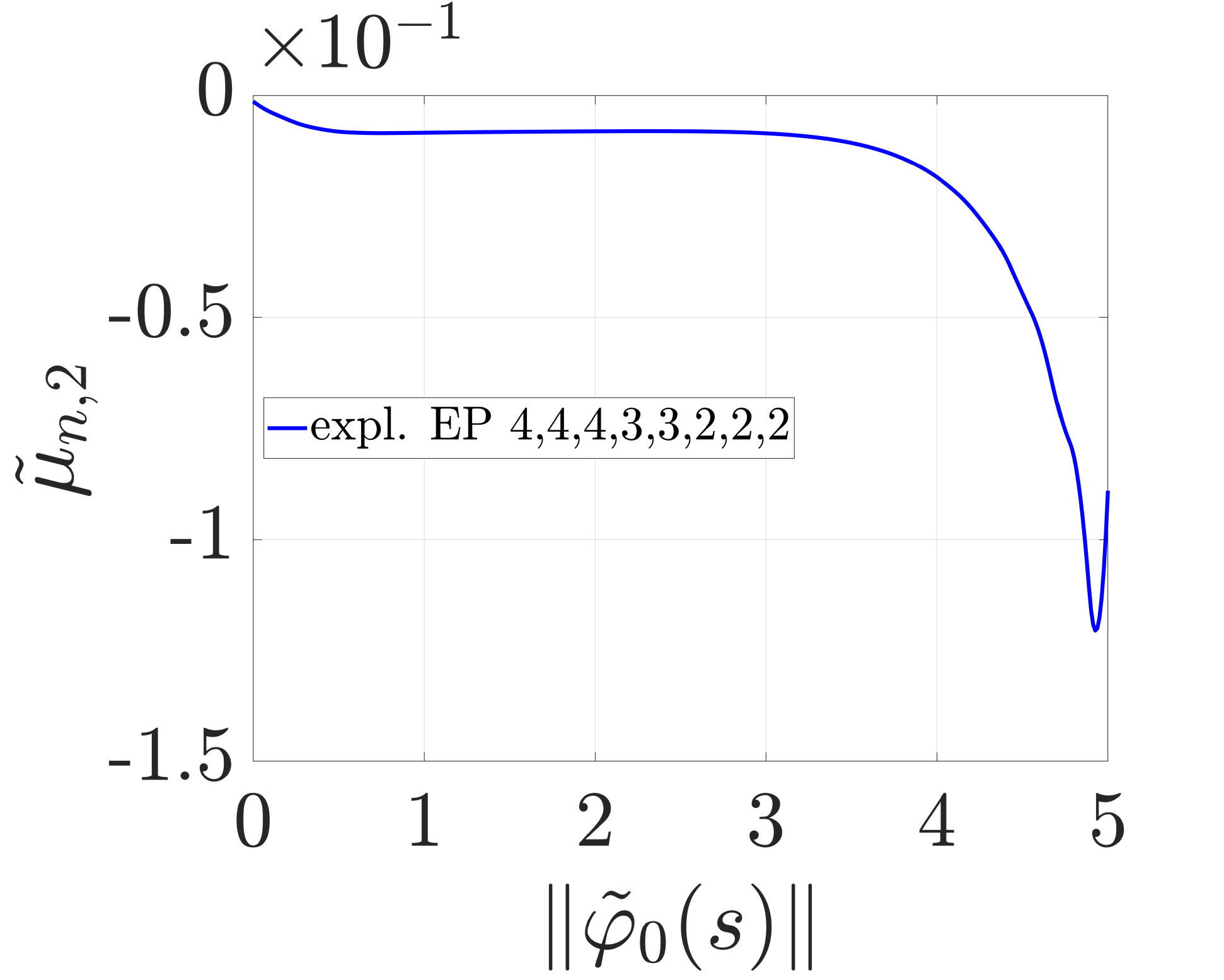}
	\end{minipage}
	\hfill
	\begin{minipage}[top]{0.325\textwidth}
		\includegraphics[width=\textwidth]{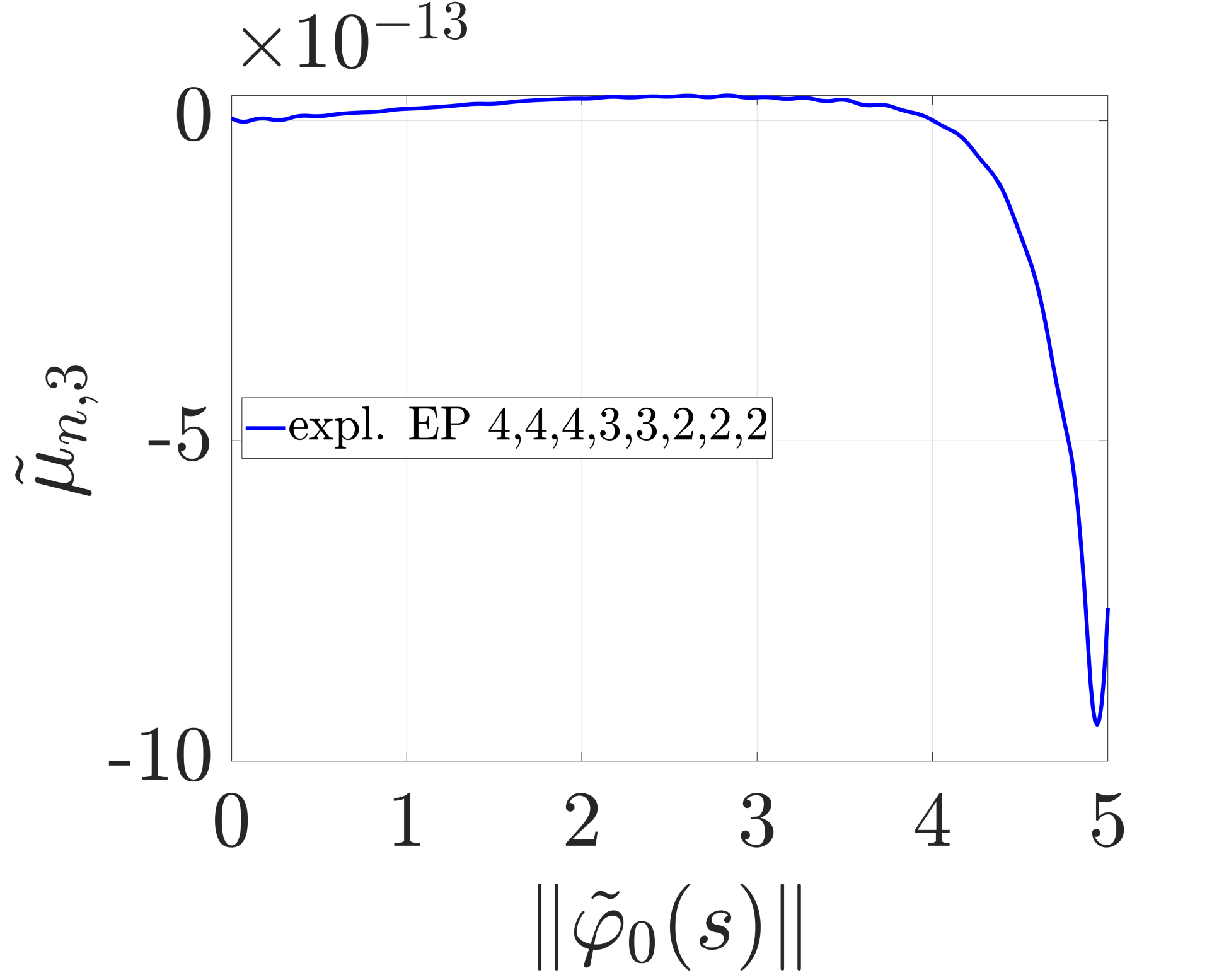}
	\end{minipage}
	\hfill
	\caption{\textbf{Torsion test:} From top to bottom, $\mathfrak{n},\,\mathfrak{m}$ and $\mathfrak{p}$, from left to right $1$-, $2$- and $3$-components are displayed.}
	\label{fig:BeamMatrixTorsionResults}
\end{figure}

\subsection{Shear test}
The last example using the above described geometry can be denoted as shear test. Therefore, we fix the lower surface of the matrix material and prescribe the displacement of the top surface. In Figure \ref{fig:BeamMatrixShear} the resulting von Mises stress distribution with (left) and without (right) constraints on hydrostatic pressure and shear are presented., see Remark \ref{re:shear} and Lemma \ref{th:area_cond_F}.

The shear test does not bend the beam, nor does the beam elongate. Additionally, the rotation in the center line is nearly constant, such that without the constraints the matrix is not affected by the beam. This can be observed in Figure \ref{fig:BeamMatrixShear}, right. In contrast, using the constraints in \eqref{eq:area_cond_F}, restricting the stretches in terms of the right Cauchy-Green tensor on the cross section of the beam, leads to a redistribution of the stresses around the virtual beam geometry, see left plot in Figure \ref{fig:BeamMatrixShear}. 

\begin{figure}[t]
	\centering
	\begin{minipage}[top]{0.49\textwidth}
		\includegraphics[width=\textwidth]{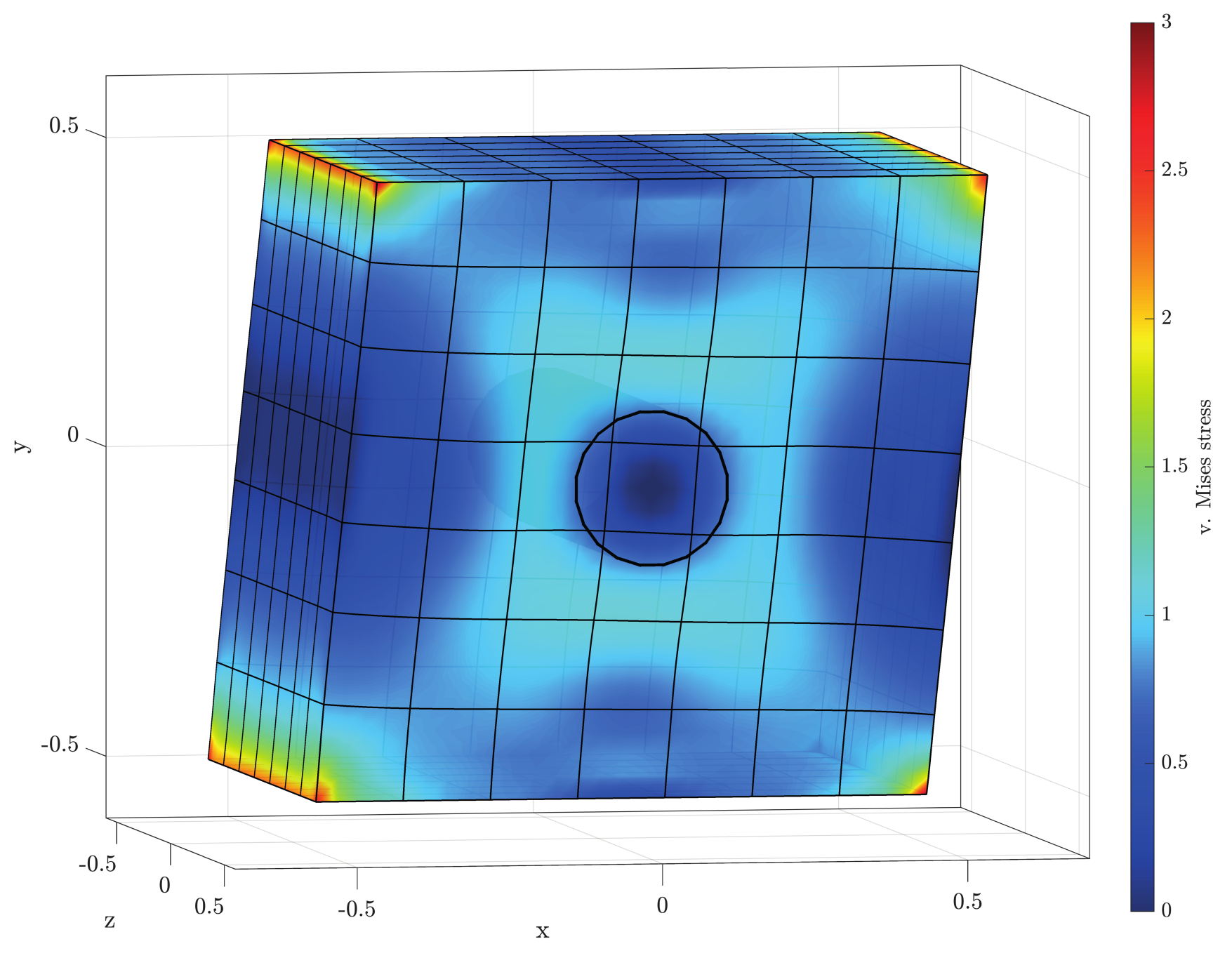}
	\end{minipage}
	\hfill
	\begin{minipage}[top]{0.49\textwidth}
		\includegraphics[width=\textwidth]{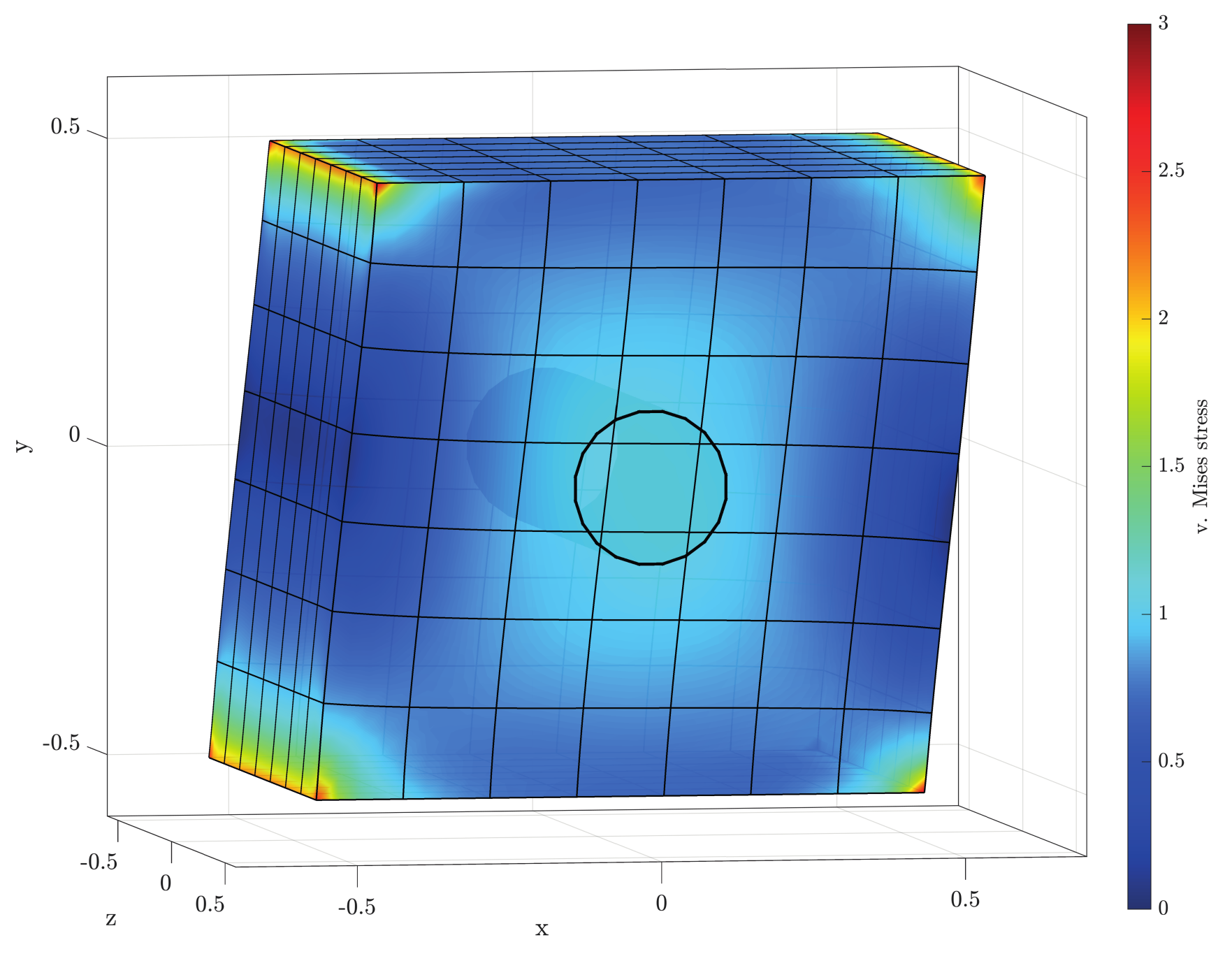}
	\end{minipage}
	\caption{\textbf{Shear Test:} Von Mises stress distribution with (left) and without (right) constraints on hydrostatic pressure and shear.}
	\label{fig:BeamMatrixShear}
\end{figure}

%%%%%%%%%%%%%%%%%%%

\subsection{Application of multiple beams}
\review{
%\begin{figure}[t]
%	\centering
%	\begin{minipage}[top]{0.49\textwidth}
%		\includegraphics[width=\textwidth]{figures/polym_d10_Ma1000_40beams_elem20_ord44433222_diriSurf_isotropic1_mesh_hollow}
%	\end{minipage}
%	\hfill
%	\begin{minipage}[top]{0.49\textwidth}
%		\includegraphics[width=\textwidth]{figures/polym_d10_Ma1000_40beams_elem20_ord44433222_diriSurf_anisotropic1_mesh_hollow}
%	\end{minipage}
%	\caption{\textbf{Application of multiple beams:} Reference configuration of the RVE, isotropic (left) and anisotropic (right) case with 40 beams.}
%	\label{fig:RVE}
%\end{figure}
As a proof of concept, we investigate in this final example a representative volume element (RVE) for fiber reinforced plastics. For the polymer we use the same Mooney-Rivlin material model as presented in \ref{num:Torsion}, but with $c_1 = 2000\,\mathrm{N/mm^2}$ and $c_2 = 1000\,\mathrm{N/mm^2}$. The glass fibers have a diameter of $d=0.01\,\mathrm{mm}$, an approximate length of $l = 0.2\,\mathrm{mm}$, a Young modulus of $E = 73000\,\mathrm{N/mm^2}$ and a Poisson ratio of $\nu = 0.3$.

The RVE has a size of $1\,\mathrm{mm}\times 1\,\mathrm{mm}\times 1\,\mathrm{mm}$ and we made use of $20\times 20\times 20$ elements. The order for [$\vec{\varphi},\,\tilde{\vec{\varphi}},\,\mathfrak{q},\,\mathfrak{n},\,\mathfrak{m},\,\bar{\vec{\mu}},\,\tilde{\vec{\mu}}_{\tau},\,\tilde{\vec{\mu}}_{n}$] is set to $[4,\,4,\,4,\,3,\,3,\,2,\,2,\,2]$. For both,  the nearly isotropic case as well as the anisotropic case we applied 40 fibres with each 2 to 5 elements depending on the length. All surfaces of the RVE are prescribed in a periodic sense to obey the Hill-Mandel criteria, equilibrating the virtual work on both scales. Therefore a prescribed, external deformation gradient $\bar{\vec{F}}$ is applied as Dirichlet boundary on the vertices and all three sets of opposing surfaces are restricted to equal fluctuations, i.e.,\ the deviation of the current configuration from the reference configuration deformed by $\bar{\vec{F}}$, with
{\begin{small}
\begin{equation}
\begin{aligned}
\bar{\vec{F}} &=
\begin{bmatrix}
\phantom{-}0.9985 & \phantom{-}0.025\phantom{0} & -0.002\phantom{0}\\
-0.01\phantom{00} & \phantom{-}1.0005 & -0.005\phantom{0}\\
-0.001\phantom{0} & \phantom{-}0.01\phantom{00} & \phantom{-}0.9985
\end{bmatrix}. \\
\end{aligned}
\end{equation}
\end{small}}

Figure \ref{fig:RVE_stresses} displays the corresponding stress distribution of the matrix material, where we obtain peak values of about $289\,\mathrm{N/mm^2}$ for the isotropic and $266\,\mathrm{N/mm^2}$ for the anisotropic case compared to average values of $\approx 170\,\mathrm{N/mm^2}$ for both cases. 

\begin{figure}[t]
	\centering
	\begin{minipage}[top]{0.49\textwidth}
		\includegraphics[width=\textwidth]{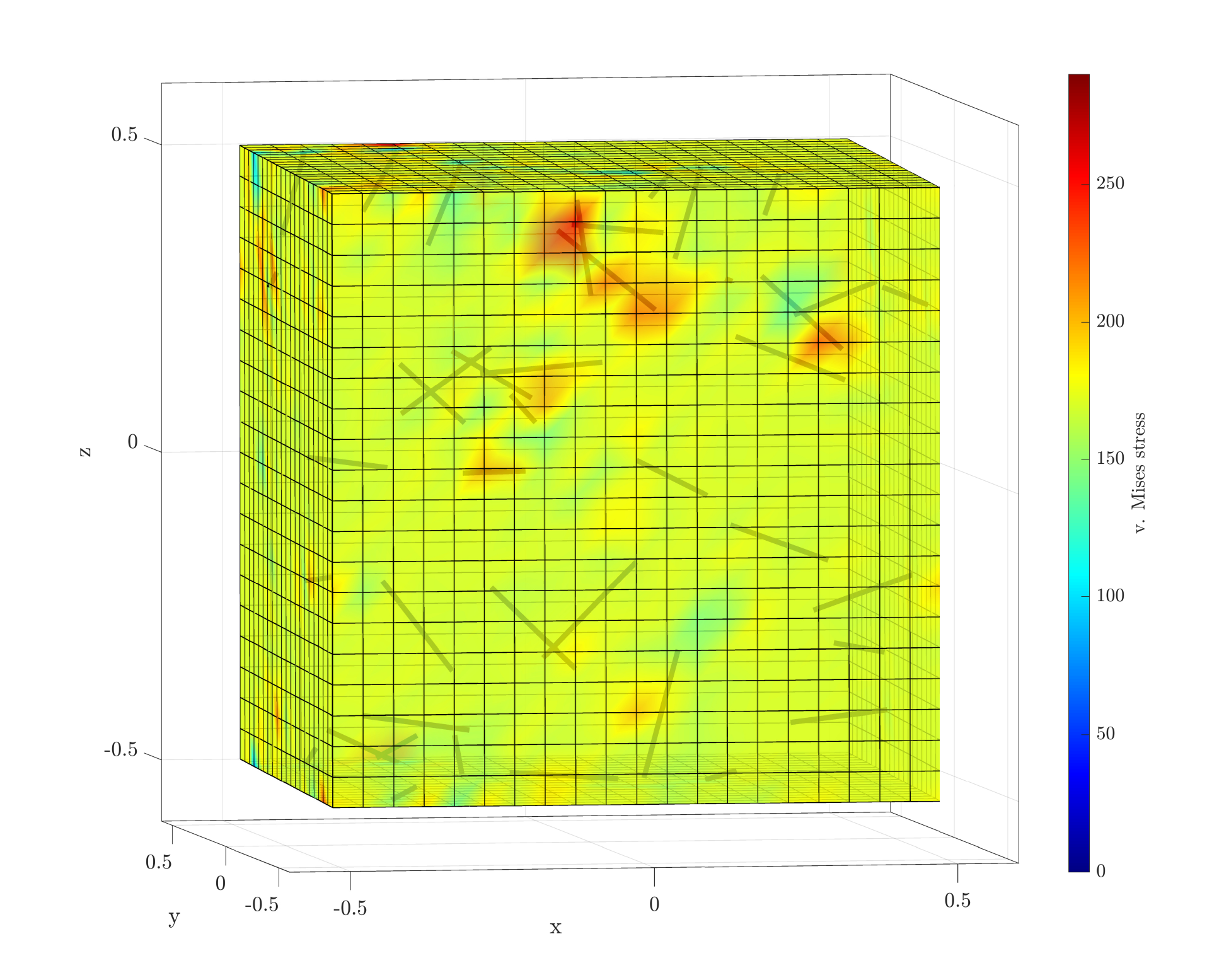}
	\end{minipage}
	\hfill
	\begin{minipage}[top]{0.49\textwidth}
		\includegraphics[width=\textwidth]{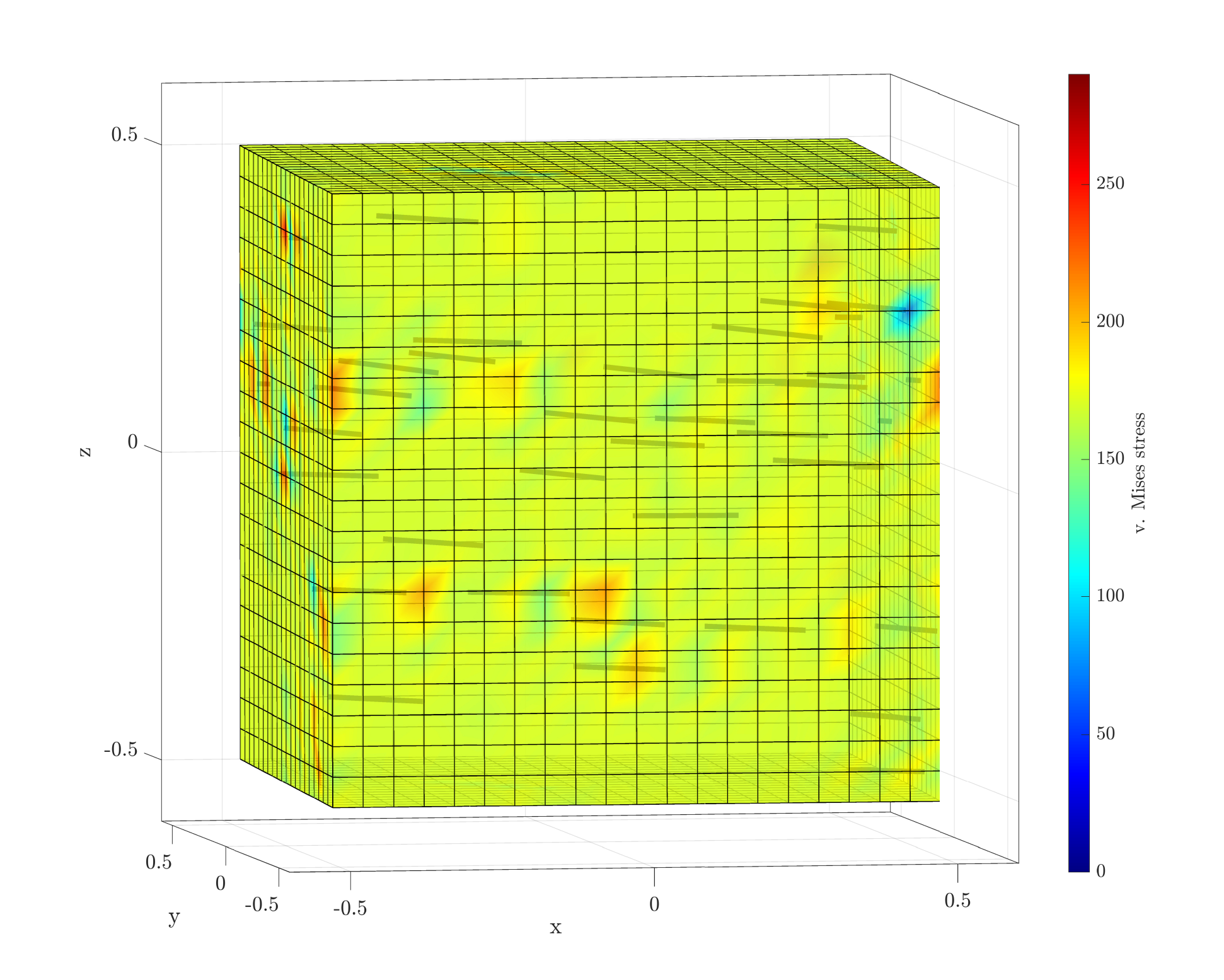}
	\end{minipage}
	\caption{\textbf{Application of multiple beams:} Von Mises stress distribution, isotropic (left) and anisotropic (right) case, plotted in the reference configuration.}
	\label{fig:RVE_stresses}
\end{figure}
}

\section{Conclusions}\label{sec:conclusions}
The framework provided in this contribution utilizes an overlapping domain decomposition technique similar to immersed techniques in fluid and solid mechanics to embed 1-dimensional fibers into a suitable 3-dimensional matrix material. For the first time, the coupling terms not only consider the forces, but also the bending and torsional moments of the beam model, leading to an enhanced formulation, which has shown to be superior to previous approaches. The chosen benchmark tests for bending and torsion demonstrate the flexibility as well as the accuracy of the proposed model. Especially the latter example regarding a torsional moment cannot be obtained from a simple coupling of forces.

The proposed static condensation procedures in the continuum as well as in the discrete setting reduces the computational effort on the calculation of the matrix material with regard to the enforcement of the shear and pressure terms of the beam, which kinematically constrains the cross sectional area. This condensation procedures for the multidimensional coupling of the immersed fibers give rise to a third order stress tensor, which is related to the concept of coupled stresses. This can be considered as first step of a homogenization, as we can now avoid the highly inefficient resolution of the beam as a classical Cauchy continuum.

\section*{Acknowledgements}
Support for this research was provided by the Deutsche Forschungsgemeinschaft (DFG) under grant DI2306/1-1. The author C. Hesch gratefully acknowledge support by the DFG. U. Khristenko and B. Wohlmuth gratefully acknowledge support by the European Union's Horizon 2020 research and innovation programme under grant agreement No 800898, the German Research Foundation by grants WO671/11-1 and WO671/15-2. We would like to thank the group of Alexander Popp at the Universit\"at der Bundeswehr M\"unchen and his Co-Worker Ivo Steinbrecher for providing us additional data for the benchmark test. We also thank the group of Torsten Leutbecher at the University of Siegen for providing all necessary information on fiber reinforced ultra high performance concrete in the final example.
\review{We also thank Patrick Le Tallec (\'Ecole Polytechnique Paris) for helpful discussions.}

\appendix
%\input{sections/sec_appendix0}
%\input{sections/sec_appendix}

%------------------------------------------------------------------
% Bibliography
%------------------------------------------------------------------	
	
\bibliographystyle{plain}
\bibliography{bibliography,literature}

\begin{thebibliography}{10}

\bibitem{antman1995}
S.S. Antman.
\newblock {\em {Nonlinear problems of elasticity}}.
\newblock Springer, 1995.

\bibitem{menzel2017a}
T.~Asmanoglo and A.~Menzel.
\newblock {A finite deformation continuum modelling framework for curvature
  effects in fibre-reinforced nanocomposites}.
\newblock {\em Journal of the Mechanics and Physics of Solids}, 107:411--432,
  2017.

\bibitem{menzel2017c}
T.~Asmanoglo and A.~Menzel.
\newblock {A multi-field finite element approach for the modelling of
  fibre-reinforced composites with fibre-bending stiffness}.
\newblock {\em Computer Methods in Applied Mechanics and Engineering},
  317:1037--1067, 2017.

\bibitem{betsch2009c}
P.~Betsch and R.~Siebert.
\newblock Rigid body dynamics in terms of quaternions: Hamiltonian formulation
  and conserving numerical integration.
\newblock {\em International Journal for Numerical Methods in Engineering},
  79:444--473, 2009.

\bibitem{betsch2002c}
P.~Betsch and P.~Steinmann.
\newblock {Conservation Properties of a Time FE Method. Part III: Mechanical
  systems with holonomic constraints}.
\newblock {\em International Journal for Numerical Methods in Engineering},
  53:2271--2304, 2002.

\bibitem{betsch2002b}
P.~Betsch and P.~Steinmann.
\newblock {Constrained dynamics of geometrically exact beams}.
\newblock {\em Computational Mechanics}, 31:49--59, 2003.

\bibitem{bonet2015a}
J.~Bonet, A.J. Gil, and R.~Ortigosa.
\newblock {A computational framework for polyconvex large strain elasticity}.
\newblock {\em Computer Methods in Applied Mechanics and Engineering},
  283:1061--1094, 2015.

\bibitem{cosserat1909}
E.~Cosserat and F.~Cosserat.
\newblock {\em {Sur la Theorie des Corps Deformables}}.
\newblock Herman, Paris, 1909.

\bibitem{d2008coupling}
C.~D'Angelo and A.~Quarteroni.
\newblock {On the coupling of 1D and 3D diffusion-reaction equations:
  Application to tissue perfusion problems}.
\newblock {\em Mathematical Models and Methods in Applied Sciences},
  18(08):1481--1504, 2008.

\bibitem{delIsola2015}
F.~dell{\textquoteright}Isola, I.~Giorgio, M.~Pawlikowski, and N.L. Rizzi.
\newblock {Large deformations of planar extensible beams and pantographic
  lattices: heuristic homogenization, experimental and numerical examples of
  equilibrium}.
\newblock {\em Proceedings of the Royal Society of London A: Mathematical,
  Physical and Engineering Sciences}, 472(2185), 2016.

\bibitem{hesch2019a}
F.~dell'Isola, P.~Seppecher, M.~Spagnuolo, E.~Barchiesi, F.~Hild, T.~Lekszycki,
  I.~Giorgio, L.~Placidi, U.~Andreaus, M.~Cuomo, S.R. Eugster, A.~Pfaff,
  K.~Hoschke, R.~Langkemper, E.~Turco, R.~Sarikaya, A.~Misra, M.~{De Angelo},
  F.~D'Annibale, A.~Bouterf, X.~Pinelli, A.~Misra, B.~Desmorat, M.~Pawlikowski,
  C.~Dupuy, D.~Scerrato, P.~Peyre, M.~Laudato, L.~Manzari, P.~G{\"o}ransson,
  C.~Hesch, S.~Hesch, P.~Franciosi, J.~Dirrenberger, F.~Maurin, Z.~Vangelatos,
  C.~Grigoropoulos, V.~Melissinaki, M.~Farsari, W.~Muller, B.E. Abali,
  C.~Liebold, G.~Ganzosch, P.~Harrison, R.~Drobnicki, L.~Igumnov, F.~Alzahrani,
  and T.~Hayat.
\newblock {Advances in pantographic structures: design, manufacturing, models,
  experiments and image analyses}.
\newblock {\em Continuum Mechanics and Thermodynamics}, 31:1231--1282, 2019.

\bibitem{dittman2020straingradient}
M.~Dittman, J.~Schulte, F.~Schmidt, and C.~Hesch.
\newblock A strain-gradient formulation for fiber reinforced polymers: Hybrid
  phase-field model for porous-ductile fracture, submitted 2020.

\bibitem{eringen1999}
A.C. Eringen.
\newblock {\em {Microcontinuum Field Theories I: Foundations and Solids}}.
\newblock Springer, 1999.

\bibitem{eugster2014}
S.R. Eugster, C.~Hesch, P.~Betsch, and C.~Glocker.
\newblock {Director-based beam finite elements relying on the geometrically
  exact beam theory formulated in skew coordinates}.
\newblock {\em International Journal for Numerical Methods in Engineering},
  97:111--129, 2014.

\bibitem{germain1973b}
P.~Germain.
\newblock {The Method of Virtual Power in Continuum Mechanics. Part 2:
  Microstructure}.
\newblock {\em SIAM Journal on Applied Mathematics}, 25:556--575, 1973.

\bibitem{gil2010}
A.J. Gil, A.~{Arranz Carre{\~n}o}, J.~Bonet, and O.~Hassan.
\newblock {The Immersed Structural Potential Method for haemodynamic
  applications}.
\newblock {\em Journal of Computational Physics}, 229:8613--8641, 2010.

\bibitem{giorgio2016}
I.~Giorgio.
\newblock {Numerical identification procedure between a micro-Cauchy model and
  a macro-second gradient model for planar pantographic structures.}
\newblock {\em Zeitschrift für angewandte Mathematik und Physik}, 67:95:1--17,
  2016.

\bibitem{glowinski1994}
R.~Glowinski, T.-W. Pan, and J.~P{\'e}riaux.
\newblock {{A} fictitious domain method for {D}irichlet problems and
  applications}.
\newblock {\em Computer Methods in Applied Mechanics and Engineering},
  111:283--303, 1994.

\bibitem{hesch2012b}
C.~Hesch, A.J. Gil, A.~{Arranz Carre{\~n}o}, and J.~Bonet.
\newblock {On immersed techniques for fluid-structure interaction}.
\newblock {\em Computer Methods in Applied Mechanics and Engineering},
  247-248:51--64, 2012.

\bibitem{hesch2014b}
C.~Hesch, A.J. Gil, A.~{Arranz Carre{\~n}o}, J.~Bonet, and P.~Betsch.
\newblock {A Mortar approach for Fluid-Structure Interaction problems: Immersed
  strategies for deformable and rigid bodies}.
\newblock {\em Computer Methods in Applied Mechanics and Engineering},
  278:853--882, 2014.

\bibitem{steinmann2013}
A.~Javili, F.~dell'Isola, and P.~Steinmann.
\newblock {Geometrically nonlinear higher-gradient elasticity with energetic
  boundaries}.
\newblock {\em Journal of the Mechanics and Physics of Solids},
  61(12):2381--2401, 2013.

\bibitem{koppl2018mathematical}
T.~K{\"o}ppl, E.~Vidotto, B.~Wohlmuth, and P.~Zunino.
\newblock Mathematical modeling, analysis and numerical approximation of
  second-order elliptic problems with inclusions.
\newblock {\em Mathematical Models and Methods in Applied Sciences},
  28(05):953--978, 2018.

\bibitem{liu2007}
W.K. Liu, D.W. Kim, and S.~Tang.
\newblock {Mathematical foundations of the immersed finite element method}.
\newblock {\em Computational Mechanics}, 39:211--222, 2007.

\bibitem{liu2006}
W.K. Liu, Y.~Liu, D.~Farrell, L.~Zhang, X.S. Wang, Y.~Fukui, N.~Patankar,
  Y.~Zhang, C.~Bajaj, J.~Lee, J.~Hong, X.~Chen, and H.~Hsu.
\newblock {{I}mmersed finite element method and its applications to biological
  systems}.
\newblock {\em Computer Methods in Applied Mechanics and Engineering},
  195:1722--1749, 2006.

\bibitem{marsden1983}
J.E. Marsden and T.J.R. Hughes.
\newblock {\em {Mathematical Foundations of Elasticity}}.
\newblock Pren\-tice-Hall, INC, 1983.

\bibitem{marsden2003}
J.E. Marsden and T.S. Ratiu.
\newblock {\em {Introduction to Mechanics and Symmetry}}.
\newblock Springer, 2003.

\bibitem{mcrobie1999}
F.A. McRobie and J.~Lasenby.
\newblock {Simo--Vu Quoc rods using Clifford algebra}.
\newblock {\em International Journal for Numerical Methods in Engineering},
  45:377--398, 1999.

\bibitem{meier2018}
C.~Meier, M.J. Grill, W.A. Wall, and A.~Popp.
\newblock {Geometrically exact beam elements and smooth contact schemes for the
  modeling of fiber-based materials and structures.}
\newblock {\em International Journal of Solids and Structures}, 154:124--1146,
  2018.

\bibitem{meier2014}
C.~Meier, A.~Popp, and W.A. Wall.
\newblock {An objective 3D large deformation finite element formulation for
  geometrically exact curved Kirchhoff rods.}
\newblock {\em Computer Methods in Applied Mechanics and Engineering},
  278:445--478, 2014.

\bibitem{meier2019geometrically}
C.~Meier, A.~Popp, and W.A. Wall.
\newblock {Geometrically exact finite element formulations for slender beams:
  Kirchhoff--\review{Love} theory versus \review{Simo--Reissner} theory}.
\newblock {\em Archives of Computational Methods in Engineering},
  26(1):163--243, 2019.

\bibitem{mindlin1964}
R.D. Mindlin.
\newblock {Micro-structure in linear elasticity}.
\newblock {\em Archive for Rational Mechanics and Analysis}, 16:51--78, 1964.

\bibitem{mindlin1965a}
R.D. Mindlin.
\newblock {Second gradient of strain and surface-tension in linear elasticity}.
\newblock {\em International Journal of Solids and Structures}, 1:417--438,
  1965.

\bibitem{ortigosa2016}
R.~Ortigosa, A.J. Gil, J.~Bonet, and C.~Hesch.
\newblock {A computational framework for polyconvex large strain elasticity for
  geometrically exact beam theory}.
\newblock {\em Computational Mechanics}, 57(2):277--303, 2016.

\bibitem{peskin2002}
C.S. Peskin.
\newblock {The immersed boundary method}.
\newblock {\em Acta Numerica}, 11:479--517, 2002.

\bibitem{Rei1981}
E.~Reissner.
\newblock {On finite deformations of space-curved beams.}
\newblock {\em Zeitschrift für angewandte Mathematik und Physik}, 32:734--744,
  1981.

\bibitem{romero2001}
I.~Romero and F.~Armero.
\newblock {An objective finite element approximation of the kinematics of
  geometrically exact rods and its use in the formulation of an energy-momentum
  conserving scheme in dynamics}.
\newblock {\em International Journal for Numerical Methods in Engineering},
  54:1683--1716, 2002.

\bibitem{rubin2010}
M.B. Rubin.
\newblock {\em {Cosserat Theories: Shells, Rods and Points}}.
\newblock Kluwer Academic Publishers, 2000.

\bibitem{puso2012}
J.~Sanders and M.A. Puso.
\newblock An embedded mesh method for treating overlapping finite element
  meshes.
\newblock {\em International Journal for Numerical Methods in Engineering},
  91(3):289--305, 2012.

\bibitem{Schulte2020b}
J.~Schulte, M.~Dittmann, S.R. Eugster, S.~Hesch, T.~Reinicke, F.~Dell'Isola,
  and C.~Hesch.
\newblock {Isogeometric analysis of fiber reinforced composites using
  Kirchhoff--Love shell elements}.
\newblock {\em Computer Methods in Applied Mechanics and Engineering},
  362:112845, 2020.

\bibitem{Simo1985}
J.C. Simo.
\newblock {A finite strain beam formulation. {T}he three-dimensional dynamic
  problem. {Part I}}.
\newblock {\em Computer Methods in Applied Mechanics and Engineering},
  49(1):55--70, 1985.

\bibitem{simo1992h}
J.C. Simo, D.D. Fox, and T.J.R. Hughes.
\newblock {Formulations of finite elasticity with independent rotations}.
\newblock {\em Computer Methods in Applied Mechanics and Engineering},
  95:277--288, 1992.

\bibitem{Simo1986}
J.C. Simo and L.~Vu-Quoc.
\newblock {A three-dimensional finite-strain rod model. {Part II:
  Computational} aspects}.
\newblock {\em Computer Methods in Applied Mechanics and Engineering},
  58:79--116, 1986.

\bibitem{soldatos2010}
K.P. Soldatos.
\newblock {Second-gradient plane deformations of ideal fibre-reinforced
  materials: implications of hyper-elasticity theory}.
\newblock {\em Journal of Engineering Mathematics}, 68:99--127, 2010.

\bibitem{spencer2007}
A.J.M. Spencer and K.P. Soldatos.
\newblock {Finite deformations of fibre-reinforced elastic solids with fibre
  bending stiffness}.
\newblock {\em International Journal of Non-Linear Mechanics}, 42:355--368,
  2007.

\bibitem{steinbrecher2019c}
I.~Steinbrecher, M.~Mayr, M.J. Grill, J.~Kremheller, C.~Meier, and A.~Popp.
\newblock A mortar-type finite element approach for embedding 1{D} beams into
  3{D} solid volumes.
\newblock {\em Computational Mechanics}, 66:1377--1398, 2020.

\bibitem{toupin1962}
R.A. Toupin.
\newblock {Elastic materials with couple-stresses}.
\newblock {\em Archive for Rational Mechanics and Analysis}, 11:385--414, 1962.

\bibitem{toupin1964}
R.A. Toupin.
\newblock {Theories of elasticity with couple stress}.
\newblock {\em Archive for Rational Mechanics and Analysis}, 17:85--112, 1964.

\bibitem{Weeger2017}
O.~Weeger, S.K. Yeung, and M.L. Dunn.
\newblock {Isogeometric collocation methods for Cosserat rods and rod
  structures}.
\newblock {\em Computer Methods in Applied Mechanics and Engineering},
  316:100--122, 2017.

\end{thebibliography}

\end{document}